\documentclass[11pt]{article}

\usepackage{natbib,amsmath,amssymb,amsfonts,amsthm,dsfont,enumitem,float}
\RequirePackage[colorlinks,citecolor=blue,urlcolor=blue,breaklinks]{hyperref}
\usepackage[normalem]{ulem}
\usepackage{epsfig,longtable,xcolor,algorithm,algpseudocode, bigints}
\usepackage{makecell}
\usepackage{setspace}
\usepackage{titlesec}
\usepackage{longtable}

          \def\cA{{\cal  A}}     
          \def\cB{{\cal  B}}     
               
          \def\cD{{\cal  D}}     
          \def\cE{{\cal  E}}     
         \def\cF{{\cal  F}}     
               
          \def\cH{{\cal  H}}     
               
          \def\cJ{{\cal  J}}     
          \def\cK{{\cal  K}}     
          \def\cL{{\cal  L}}     
               
          \def\cN{{\cal  N}}     
               
          \def\cP{{\cal  P}}     
               
          \def\cR{{\cal  R}}     
          \def\cS{{\cal  S}}     
          \def\cT{{\cal  T}}

          \def\cW{{\cal  W}}

          \def\cZ{{\cal  Z}}

\def \Up {{\Upsilon}}
 
\DeclareMathOperator*{\argmin}{argmin}

\DeclareMathOperator{\Cov}{Cov}

\def \EE {\mathbb{E}}
\def \RR {\mathbb{R}}
\def \PP {\mathbb{P}}
\def \SS {\mathbb{S}}

\def \NN {\mathbb{N}}
\def \FF {\mathbb{F}}

\newcommand*{\boldone}{\text{\usefont{U}{bbold}{m}{n}1}}

\newcommand{\beq}  {\begin{equation}}
\newcommand{\eeq}  {\end{equation}}
\newcommand{\beqn} {\begin{eqnarray}}
\newcommand{\eeqn} {\end{eqnarray}}
\newcommand{\beqnn}{\begin{eqnarray*}}
\newcommand{\eeqnn}{\end{eqnarray*}}

\theoremstyle{definition}

\newtheorem{innercustomasp}{Assumption}
\newenvironment{aspt}[1]
  {\renewcommand\theinnercustomasp{#1}\innercustomasp}
  {\endinnercustomasp}

\newtheorem*{rem}{Remark}
\theoremstyle{plain}
\newtheorem{lem}{Lemma}
\newtheorem{cor}{Corollary}
\newtheorem{prop}{Proposition}
\newtheorem{thm}{Theorem}

\newcommand{\true}[0]{\overline b_{n}^{E}}
\newcommand{\truej}[0]{\overline b_{n}^{j\cdot E}}
\newcommand{\truenj}[0]{g_{n}^{j\cdot E}}
\newcommand{\truestar}[0]{b_{n}^{E}}
\newcommand{\truestarj}[0]{b_{n}^{j\cdot E}}

\newcommand{\trueprime}[0]{\overline b_{n}^{\prime E}}
\newcommand{\truejprime}[0]{\overline b_{n}^{\prime j\cdot E}}

\newcommand{\truenjprime}[0]{g_{n}^{\prime j\cdot E}}

\newcommand{\mle}[0]{\widehat{\beta}_{n}^{E}}
\newcommand{\mlej}[0]{\widehat{\beta}_{n}^{j\cdot E}}

\newcommand{\mlenj}[0]{\widehat{\Gamma}_{n}^{j \cdot E}}
\newcommand{\mlenjr}[0]{\widehat{\gamma}_{n}^{j \cdot E}}

\newcommand{\mleprime}[0]{\widehat{\beta}_{n}^{\prime E}}
\newcommand{\mlejprime}[0]{\widehat{\beta}_{n}^{\prime j\cdot E}}

\newcommand{\mlenjprime}[0]{\widehat{\Gamma}_{n}^{\prime j \cdot E}}

\newcommand{\Rnj}[0]{A_1}
\newcommand{\Tn}[0]{A_2}

\newcommand{\lasso}[0]{\widehat{B}_{n,\lambda}}
\newcommand{\Z}[0]{\widehat{Z}_{n,\lambda}}
\newcommand{\zr}[0]{z_{\lambda}}
\newcommand{\Sgn}[0]{\widehat{S}_{n,\lambda}}
\newcommand{\sgnr}[0]{s_{\lambda}}
\newcommand{\D}[0]{\widehat{D}_{n,\lambda}}
\newcommand{\dr}[0]{d_{\lambda}}

\newcommand{\Ej}[0]{\eta_{j\cdot E}}
\newcommand{\Vj}[0]{\widehat{V}_{n}^{j\cdot E}}
\newcommand{\vjr}[0]{v^{j\cdot E}}
\newcommand{\Uj}[0]{\widehat{U}_{n}^{j\cdot E}}

\newcommand{\varmle}[0]{\Sigma_{E,E}}
\newcommand{\covmlej}[0]{\Sigma_{E, j}}
\newcommand{\varmlej}[0]{\sigma^2_{j}}

\newcommand{\rand}[0]{\omega_n}
\newcommand{\varrand}[0]{\Omega}

\newcommand{\Zn}[0]{\cZ_n} 
 
\newcommand{\pivot}[0]{\cP^{j\cdot E} \left(\cZ_n\right)}
\newcommand{\pivotup}[0]{\cP^{j\cdot E} \left(\Upsilon_n\right)}
\newcommand{\pivotz}[0]{\cP^{j\cdot E} \left(\cZ\right)}
\newcommand{\event}[0]{\left\{\D = \dr, \Vj= \vjr\right\}}

\newcommand{\RD}[0]{\operatorname{RD}_n} 
\newcommand{\rv}[0]{\cT_n} 
\newcommand{\tail}[0]{L_n} 
\newcommand{\arb}[0]{\cH} 
\newcommand{\mn}[0]{R} 
\newcommand{\W}[0]{\cW_{\alpha, \kappa}} 

\providecommand{\keywords}[1]
{
  \small	
  \textbf{Keywords:} #1
}

\begin{document}

\title{Asymptotically-exact selective inference for quantile regression}
\date{}
\author{Yumeng Wang \\ 
Department of Statistics, University of Michigan \\
and \\
Snigdha Panigrahi \\
Department of Statistics, University of Michigan \\
and \\
Xuming He \\
Department of Statistics and Data Science, Washington University in St. Louis
}

\maketitle

\begin{abstract}
In modern data analysis, it is common to select a model before performing statistical inference.
Selective inference tools make adjustments for the model selection process in order to ensure reliable inference post selection.
In this paper, we introduce an asymptotic pivot to infer about the effects of selected variables on conditional quantile functions.
Utilizing estimators from smoothed quantile regression, our proposed pivot is easy to compute and yields asymptotically-exact selective inference without making strict distributional assumptions about the response variable. 

At the core of our pivot is the use of external randomization variables, which allows us to utilize all available samples for both selection and inference, without partitioning the data into independent subsets or discarding samples at any step.
From simulation studies,  we find that: (i) the asymptotic confidence intervals based on our pivot achieve the desired coverage rates, even in cases where sample splitting fails due to insufficient sample size for inference; (ii) our intervals are consistently shorter than those produced by sample splitting across various models and signal settings.
We report similar findings when we apply our approach to study risk factors for low birth weights in a publicly accessible dataset of US birth records from 2022.
\end{abstract}


\setstretch{1.5}

\section{Introduction}
\label{sec:introduction}

Quantile regression, proposed by \cite{Koenker_and_Bassett}, estimates the conditional quantiles of the response by minimizing a piecewise linear loss function, also known as the check loss function.
Quantile regression does not assume a specific parametric family of distributions or a constant variance for the response variable and is therefore more robust against heavy-tailed errors and outliers than least squares regression.
Due to its robustness and versatility, quantile regression is a popular tool for analyzing the relationship between variables in heterogeneous data.   
See, for example, papers by \cite{ecology, quantile_panel, reference_growth, QRank} on a variety of application domains. 

There is a computational downside to consider, though. 
Estimating conditional quantiles using the non-differentiable check loss function does not scale easily with big data. 
This makes it challenging to select important covariates from a large pool of variables when using penalized versions of the standard quantile regression.
The task of performing inference with big data is even more difficult, and is complicated by the need to estimate nuisance parameters that rely on the true population conditional densities.

Recent papers by \cite{smoothing_fernandes} and \cite{smooth_quantile} have introduced a new method called the Smoothed Quantile Regression (SQR) that overcomes these limitations of the standard quantile regression. 
The SQR method uses a convolution-type smoothing technique to make the  loss function convex and differentiable. 
By adding a penalty term to the loss function, this method can be used to estimate conditional quantiles when analyzing big data with high dimensional covariates.
Although the SQR method can select important covariates and provide point estimators, it does not attach uncertainties to these point estimators or address the practical issue of inference for the effects of the selected variables on the conditional quantile function.
Using the post-selection inference framework, also known as selective inference, our paper provides a theoretically-justified and an easy to implement solution to this problem.

Our paper introduces an asymptotic pivot after selecting variables with the $\ell_1$-penalized SQR method. 
We form this pivot from a conditional distribution of quantile regression estimators given the event of selection.
This method allows us to use the entire set of observed samples for both selection and inference, without partitioning the data into independent subsets  or discarding a portion of the samples at either step.
Inverting our pivot yields asymptotically-exact selective inference for the effects of the selected variables, with coverage probability converging to the desired level as the sample size grows to infinity.

The appeal of our newly introduced pivot lies in its four key features:
\begin{enumerate}
\item it is easy to calculate and only involves one-dimensional integrals for its numerical calculation;
\item it circumvents the need to estimate the true conditional density functions, which are nuisance parameters in our inferential problem;
\item it takes the same form as pivots used in drawing selective inference with least squares regression,  but it stays true to the essence of quantile regression by making no parameteric assumptions about the conditional distribution of the response variable;
\item the confidence intervals based on this pivot continue to assure asymptotically-exact coverage guarantees even when selection events are rare and occur with vanishing probabilities.
\end{enumerate}

We start by giving a brief overview of the $\ell_1$-penalized SQR approach for quantile regression.
Then we present a data example demonstrating the limitations of directly adapting existing post-selection inference methods to the quantile regression problem, followed by a concise summary of our key contributions.

\subsection{The $\ell_1$-penalized SQR method}
\label{sec1.2}

Consider a $p$-dimensional covariate $x \in \RR ^ p$ and a scalar response variable $y \in \RR$. 
Let $\{(y_i, x_i)\}_{i=1} ^ n$ be a set of $n$ independent and identically distributed realizations of $(y, x)$. 
Suppose that $(y, x)$ follows the distribution $\FF$.
It follows that $\{(y_i, x_i)\}_{i=1} ^ n$ will have a joint distribution $\FF_n = \FF \times \cdots \times \FF \; (n \text{ times})$.
Using matrix notation, let $Y = (y_1, y_2, \ldots, y_n) ^ \top\in \RR^n$ and $X = (x_1, x_2, \ldots, x_n) ^ \top \in \RR^{n\times p}$ denote the response vector and the design matrix, respectively.

Suppose that we want to estimate the $\tau$-th  conditional quantile of $y$ given $x$ at a pre-specified quantile level $\tau$.
Let 
$$
\rho_\tau(u) = u\{\tau - 1(u<0)\}
$$ 
be the non-differentiable check loss function.
Additionally, let $K:\RR \to [0,\infty)$ be a symmetric, non-negative function that integrates to $1$ and let
$$K_h(u) = \frac{1}{h} K\left(\frac{u}{h}\right)$$
be a kernel function with a fixed bandwidth parameter $h>0$.
One commonly used kernel function, for example, is the Gaussian kernel function, which is given by
$K_h(u) = \frac{1}{\sqrt{2 \pi}h} e^{-\frac{u^2}{2h ^ 2}}.$
More examples of kernel functions can be found in \cite{smooth_quantile}.

Denote by
\begin{equation}
\widehat Q_{n;h} (\theta; Y) = \frac{1}{n} \sum_{i=1} ^ n \int_{-\infty} ^ \infty \rho_\tau(u) K_h(u + \theta_i - y_i) du
\label{convoluted:loss} 
\end{equation}
the convolution smoothed quantile loss with the kernel function $K_h$ and bandwidth parameter $h$, where $\theta$ is an $n$-dimensional vector.
To select a subset of the $p$ variables, the $\ell_1$-penalized SQR method solves
\begin{equation*}
	\underset{\beta \in \RR ^ p}{\operatorname{minimize}} \; \left\{\sqrt{n} \widehat Q_{n; h} (X\beta; Y)  + \lambda \|\beta\|_1\right\}.
\end{equation*}
When bandwidth parameter $h$ is of the order $\{\log (p) / n\}^{1 / 4}$, the difference between $\ell_1$-penalized SQR estimator and the $\ell_1$-penalized quantile regression estimator is negligible, as shown in Theorem 1 of \cite{high_smooth_quantile}.

In the rest of the paper, we focus on performing selective inference for the effects of variables in the support set estimated with the $\ell_1$-penalized SQR method.

\subsection{Selective inference and a first example}

Let $E\subset \{1,2,\ldots, p\}$ be the selected subset of variables using the $\ell_1$-penalized SQR estimator and let $q = |E|$.
Post selection, it is natural to consider selective inference for the population parameter 
\begin{equation}
\label{eq:infer_target:inf}
    \truestar = \argmin_{b \in \RR ^ q} \EE_{\FF_n} \left[\frac{1}{\sqrt{n}} \sum_{i = 1} ^ n \rho_\tau(y_i - x_{i, E} ^ \top b)\right],
\end{equation}
which is obtained by using the check loss function $\rho_\tau(\cdot)$ on the selected subset of variables with size $q$.

Note that in \eqref{eq:infer_target:inf}, $\EE_{\FF_n} [D]$ denotes the expectation with respect to the true distribution of the variables. 
This type of target is similar to the projection-based targets discussed in the survey article by \citep{zhang2022post}. 
Moreover, this population target provides meaningful parameters even when  the conditional quantile function given $x_E$ is  nonlinear. 
As shown in  \cite{angrist2006quantile}, this target parameter represents the best linear approximation to the conditional quantile function based on a weighted square  loss function.

In reality, the set $E$ is dependent on data. 
Valid selective inference for $\truestar$, without making any unrealistic assumptions about the quality of selection, can be achieved by conditioning on the selection event.
This approach has been investigated by several authors for least squares regression.  
Starting with \cite{lee_seletive}, this line of work, used properties of the $\ell_1$-penalty to show that the selection event in the fixed-$X$ least squares regression setting can be described by a set of linear inequalities in $Y$ and takes the shape of a polyhedron.
In the case that the response is a Gaussian variable, truncating the normal distribution of $Y$ to this polyhedron results in a truncated Gaussian pivot.
This method is commonly referred to as the polyhedral method for selective inference.
In the same setting, \cite{tibshirani_bootstrap} investigated large sample properties of the polyhedral method within a nonparametric family of distributions.

These existing methods may suggest that switching from a quadratic loss  to the convolution smoothed quantile loss using the same approach would yield valid selective inference for our problem.
However, Figure \ref{fig:randomized} demonstrates that the polyhedral method can be quite brittle, as the reliability of inference for $\truestar$, with a data-dependent set $E$, gets severely compromised.
An adaptation of the polyhedral method to the $\ell_1$-penalized SQR framework is shown as ``Previous'' for three regimes of signal strength, ``Low'', ``Mid'' and ``High'' in this figure.
The results are based on $500$ independent Monte Carlo experiments for each setting.
The data in this first example obeys: $y = x ^ \top \beta + \varepsilon - F_\varepsilon ^ {-1} (\tau)$ where $x \in \RR ^ p$ follows a multivariate Gaussian distribution with $p = 200$, $\varepsilon \in \RR$ follows a Gaussian distribution.
Here, $F_\varepsilon ^ {-1} (\tau)$ denotes the $\tau ^ {\text{th}}$ quantile of $\varepsilon$.
The coefficient vector $\beta \in \RR ^ p$ has $5$ nonzero components whose magnitudes vary according to the three different signal strengths. 
We set the sample size to $n = 800$, the quantile level to $\tau = 0.7$ and choose the tuning parameter as $\lambda = \sqrt{\log p/n}$.

The confidence intervals based on ``Previous'' not only fall short in coverage in the ``Low'' and ``Mid'' signal regimes (see leftmost panel of the figure), but also show  numerical instability with low power. 
As seen in the rightmost panel of the figure, the percentage of unbounded intervals with ``Previous'' can be as high as $50\%$ in the ``Low'' signal regime. 
In fact, the paper by \cite{kivaranovic2020tight} showed that the “Previous” method can generate infinitely long intervals, as seen in our ``Low'' and ``Mid'' signal regimes. 
However, it might be surprising to see that this method even falls short of attaining the intended level of coverage. 
This occurs especially in the ``Low'' signal regime when selection events become rare or have low probabilities.
The large sample guarantees for “Previous” in \cite{tibshirani_bootstrap} no longer hold in these settings.

\begin{figure}[t]
	\centering
	\includegraphics[scale = 0.55]{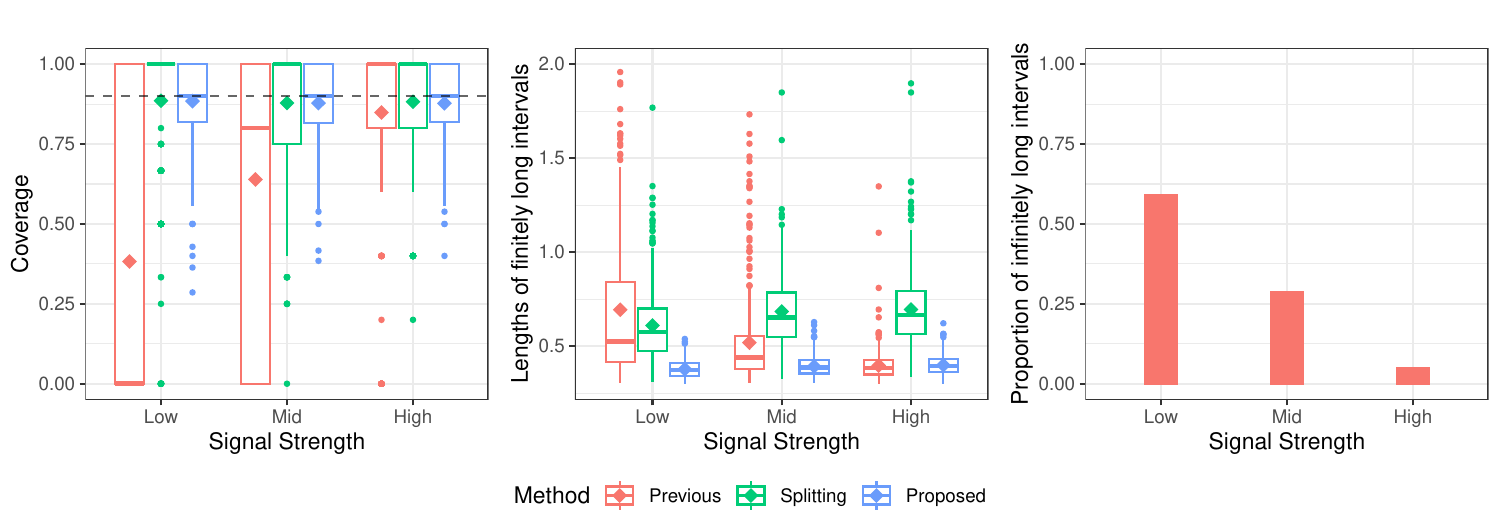}
	\caption{
    \small{Comparison between an adaptation of the polyhedral method (``Previous''), data splitting (``Splitting") and our proposed method (``Proposed'').
    \textit{Left:} Box plots for coverage probabilities of 90\% confidence intervals with the diamond symbol  denoting the mean coverage rate of intervals. The ``Proposed'' method provides valid selective inference across all signal regimes, as does ``Splitting''. However, the coverage rates of the ``Previous'' method are significantly lower than the desired level in both ``Low'' and ``Mid'' signal regimes.
    \textit{Middle:} Lengths of confidence intervals. The confidence intervals generated by ``Proposed'' are substantially shorter than that of ``Previous” in both ``Low'' and ``Mid'' signal regimes, and are also shorter than ``Splitting'' across all regimes.
    \textit{Right:} Proportions of infinitely long intervals. The ``Previous'' method has a high probability of generating infinite intervals in ``Low'' and ``Mid'' signal regimes, which is consistent with the finding in \cite{kivaranovic2020tight}.}
}
\label{fig:randomized}
\end{figure}

\subsection{Randomized selective inference}

In Figure \ref{fig:randomized}, we also present the empirical coverage rates and lengths of intervals based on our newly developed pivot in the paper, which is depicted as ``Proposed''. 
We observe that this method provides more reliable inference, producing bounded intervals that are always narrower than those formed by adapting the polyhedral method to provide inference. 
This observation holds true even after excluding the unbounded intervals produced by the ``Previous'', as shown in the middle panel of this figure.

In the same figure, we also observe that our intervals are significantly shorter compared to the intervals obtained through data splitting. 
To conduct data splitting, we use a random subsample that contains two-thirds of the total sample size for variable selection, and then use the remaining one-third samples for selective inference on the conditional quantile.

The success of our new pivot lies in the addition of external randomization variables to the optimization objective used by the $\ell_1$-penalized SQR method. 
Consider an external randomization variable $\sqrt{n}\rand$ from $\cN(0_p, \varrand)$, independent from our data, where $\varrand\in \RR^{p\times p}$ is a user-specified covariance matrix and $0_p$ is a vector of zeros of length $p$. 
Using $\rand$, we implement a randomized version of the $\ell_1$-penalized SQR method, given by
\begin{equation}
\label{eq:est_random}
  \underset{\beta \in \RR ^ p}{\text{minimize}} \;\left\{\sqrt{n} \widehat Q_{n;h} (X\beta; Y) + \lambda \|\beta\|_1  - \sqrt{n} \rand ^ \top \beta \right\}.
\end{equation}
As an example, consider fixing $\varrand= \delta^2 \cdot I_{p,p}$, i.e., we draw $p$ i.i.d. external Gaussian random variables with mean $0$ and variance $\delta^2$. 
The parameter $\delta^2$ allows us to control the amount of randomization introduced at the time of model selection.

In what follows, we develop an asymptotic pivot to infer about $\truestar$, where $E$ now depends on the set of variables selected by \eqref{eq:est_random}.
Including a randomization variable in this form is indeed crucial for obtaining a pivot that is easy to calculate.
Note that this type of randomization differs from the randomness involved in data splitting. 
As emphasized earlier, the form of randomization in \eqref{eq:est_random} enables us to use all samples during selection and inference, rather than discarding any of our samples at either of the two steps. 

In a second experiment illustrated in Figure \ref{fig:randomized_level}, we vary the randomization variance parameters $\delta^2 \in \{0.4, 0.6, 0.8, 1\}$ in the randomization scheme with $\varrand= \delta^2 \cdot I_{p,p}$. 
These different levels of randomization variance correspond to the Settings $1,2,3,4$ as shown on the x-axis of the figure, with Setting $0$ representing the case without randomization. 
In Setting $0$, we reuse the same data for inference without accounting for the selection process, which we refer to as the ``Naive'' method.

Note that with a small amount of randomization in the model selection process, we achieve a model with comparable power to the one trained on the full dataset.
This is evident from the ``Recall'' values in the leftmost panel of the figure, where ``Recall'' value is the proportion of active variables (with non-zero coefficients in the true model) correctly identified by the $\ell_1$-penalized SQR method.
While the ``Naive'' approach does not guarantee valid inference after selection—potentially leading to an inflated rate of false positives—our method ensures valid inference following randomized selection, as shown in the central panel of the figure.
Moreover, the coverage probabilities and interval lengths in our randomized approach remain stable, exhibiting low variability across different signal settings.

\begin{figure}[h]
	\centering
	\includegraphics[scale = 0.55]{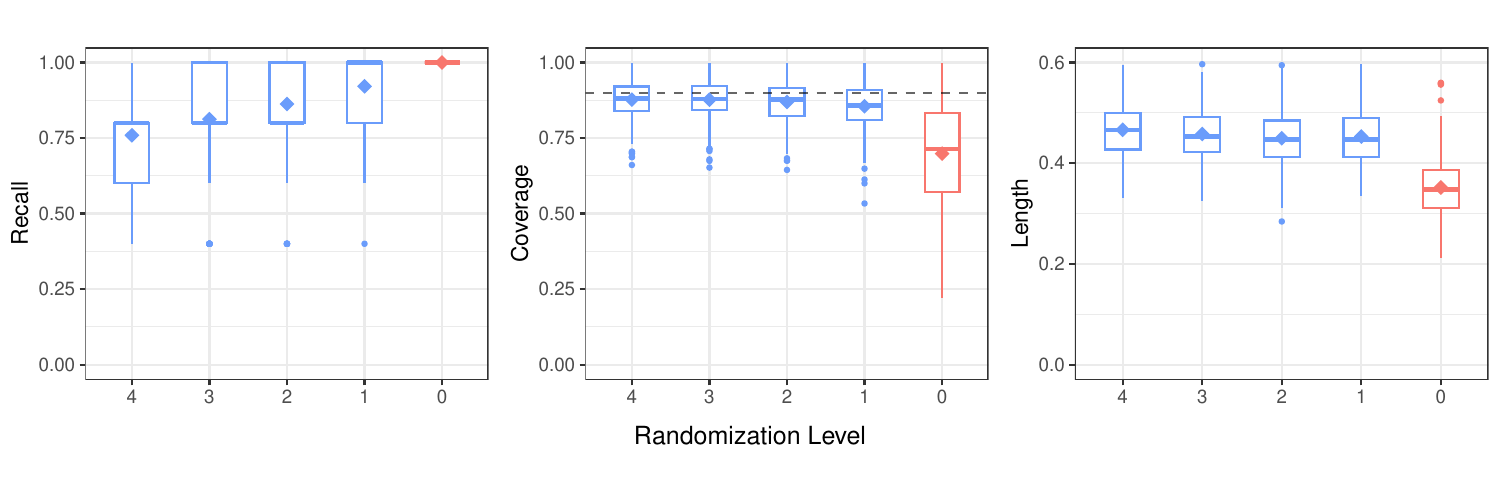}
    \vspace{-15pt}
	\caption{\small{The performance of selection and inference of our proposed method with varying randomization variance levels. Randomization levels 1, 2, 3, and 4 correspond to $\delta^2 = 0.4, 0.6, 0.8, 1$, while randomization level 0 corresponds to the ``Naive'' method.
    \textit{Left:} ``Recall'' across different randomization levels, showing the impact of varying $\delta^2$ on model selection. 
    \textit{Middle:} Coverage probabilities of 90\% confidence intervals.
    \textit{Right:} Lengths of confidence intervals.
    With a small randomization variance, our proposed method provides valid inference, whereas the ``Naive'' method underperforms.}}
\label{fig:randomized_level}
\end{figure}

A summary of our technical contributions is as follows. 
\begin{enumerate}
\item We derive two important asymptotic expansions in Section \ref{sec:pivot}. The first one is for the refitted SQR estimators using the selected set of variables in Proposition \ref{prop:beta_gamma_expression}, and the second one connects the refitted SQR estimators to the randomized $\ell_1$-penalized SQR estimators through randomization in Proposition \ref{prop:taylor_expansion}.  These expansions help us establish the marginal distributional properties of these estimators and guide the construct of our asymptotic pivot, linking it to the pivot form used in the rather well-studied least squares setting.
\item By making mild assumptions on the errors in the two asymptotic expansions, we prove rigorous inference guarantees for our pivot in Theorem \ref{thm:main_weak_convergence}. 
The asymptotically-exact guarantees in Section \ref{sec:theory} apply uniformly across a range of data generating distributions $\mathbb{F}_n$, assuming sub-Gaussian covariates but making no strong assumptions about the conditional distribution of $Y$. 
The proofs of our results rely on tail and moment bounds for the score function based on the SQR loss, which allow us to establish convergence rates for our asymptotic pivot.
\item In Sections \ref{sec:simulation} and \ref{sec:real}, we evaluate the performance of our method on simulated datasets and a publicly available dataset on US birth weights in $2022$.
To evaluate coverage rates, power, and estimation accuracy, we analyze various models and signal settings using our post-selection inferential pipeline.
Additional simulations, including data from misspecified models, are provided in the Supplementary Material.
Although we use simple white noise as randomization variables in our numerical analysis, we emphasize that our theory in Section \ref{sec:theory} covers the wider range of mean-centered Gaussian randomization variables. 

\end{enumerate}

\section{Related work}
\label{sec:related}

\subsection*{Inference in high dimensional quantile regression.} 
Statistical inference with penalized quantile regression has been widely acknowledged as an interesting and challenging task in the literature. 
Various strategies have been employed for fitting quantile regression, including kernel smoothing \citep{horowitz1998bootstrap, galvao2016smoothed}, using information criteria for model selection \citep{behl2014focussed}, and employing fast algorithms like the alternating direction method of multipliers \citep{gu2018admm}.
Due to its ease of implementation and ability to scale up with large datasets, the SQR method for estimating conditional quantile functions has become increasingly popular and has been adapted for different types of data. 
See, for example, the papers by \cite{Unified_algorithm_SQR} and \cite{Censored_SQR}.

We also draw attention to some earlier work on hypothesis testing and inference in quantile regression.
Marginal testing methods for detecting significant associations with conditional quantiles were developed by \cite{wang2018testing, tang2022conditional}. 
Work by \cite{belloni2019valid} tackled high dimensional inference using orthogonal score functions.
While \cite{smooth_quantile} analyzed the asymptotic properties of the SQR estimator in the low dimensional setting, \cite{debiased_quantile} proposed a debiased $\ell_1$-penalized SQR estimator and established confidence intervals using their debiased approach for high dimensional data. 
Note that these papers provide inference for the entire $p$-dimensional parameter vector in a predetermined model, not a selected model.
More specifically, when selection is done prior to making inference, and analysts report and use the selected subset of variables for downstream inferential tasks, these techniques are no longer applicable for the selected set of post-selection parameters.  
Our work fills this gap by providing analysts with a rigorous toolbox to construct and report valid confidence intervals (and p-values) for the effects of these selected variables on the conditional quantile function.

\subsection*{Post-selection inference or selective inference.}
The challenges of characterizing the finite sample distribution of post-selection estimators, which exhibit intricate dependence on model parameters, have been documented in \cite{leeb2003finite, leeb2005model}.
Instead of using plug-in estimators for nuisance parameters, the conditional approach, discussed in the earlier section, addresses this limitation by conditioning on the sufficient statistics for these parameters, constructing a pivot free of nuisance parameters. 
Alternatively, a simultaneous approach has been developed for selective inference in the normal linear model and the logistic regression model by \cite{BerkBrownBujaZhangZhao2013} and \cite{BachocPreinerstorferSteinberger2016}, respectively. 
An advantage of this approach is that it remains valid regardless of the method of selection. 
However, the long intervals generated by the  marginal viewpoint taken by this approach and the lack of easy computational tools to construct these intervals hinder its widespread application.

\subsection*{Randomized selective inference.}
Several recent papers have used the concept of utilizing an external randomization variable for performing selective inference. 
For example,
\cite{tian_randomized} and \cite{kivaranovic2020tight} added a randomization variable to the response for gaining power,
\cite{rasines2021splitting} used a similar randomization scheme to split their data into a selection set and validation set, and this idea was generalized to other distributions in \cite{LeinerDuanWassermanRamdas2021} and \cite{dharamshi2023generalized}. 
Furthermore, \cite{zrnic2023post} connected the use of randomization with algorithmic stability, and \cite{DaiLinXingLiu2023} added a Gaussian random variable to predictors for controlling false discovery rate control.
Note that the selective inference method in our paper employs a different form of randomization.
Specifically, our randomization scheme involves adding noise variables to the gradient of the SQR objective, rather than directly adding them to the response. 
In contrast to previous approaches that rely on specific parametric distributions for randomizing their data, our setup does not involve such modeling restrictions. 
Our randomization is inspired from \cite{panigrahi2021integrative} and \cite{panigrahi2018carving} that used a similar form of randomness in a class of techniques known as ``carving''.
This involves conducts selection on a random subset of the data as done during data splitting, while selective inference is performed on the full data after conditioning on the selection event from the first stage data. 

Drawing selective inference from conditional distributions can be quite challenging, as these distributions are usually computationally intractable.
Prior work has achieved approximate selective inference by utilizing feasible approximations of these conditional distributions. 
For example, tractable likelihood functions were developed by \cite{panigrahi2017mcmc} to approximate a pivot, and by \cite{panigrahi2023integrative} and \cite{panigrahi2023approximate} for sampling from a post-selection posterior distribution.
An approximate maximum likelihood approach was developed in \cite{approximate_snigdha} for normal data, which was generalized to cover generalized linear models and a broad category of M-estimation problems in \cite{huang2023selective}. 

Instead of constructing such approximate likelihood functions, we take a different approach in our paper by providing a pivot that yields asymptotically-exact post-selection inference for the effects of selected variables on the conditional quantile function.
Our method does not rely on approximations, ensuring consistent performance across a variety of scenarios without being affected by potential fluctuations in the quality of such approximations.
Moreover, the pivot we construct using the SQR estimators mirrors the form of the pivot introduced by \cite{panigrahi2023exact} in the least squares setting. 
It is worth pointing out that the pivot in previous work ensures valid selective inference for normal data. 
What sets our contribution apart in the current paper is that our pivot allows for valid inference without placing strict assumptions on the conditional distribution of the response variable.

\section{Pivot using SQR estimators}
\label{sec:pivot}

\subsection{Basics}

Let $[p]$ denote the set $\{1, 2,\ldots, p\}$.
Given a matrix $M \in \mathbb{R}^{p \times p}$ and two nonempty subsets of $[p]$, $E_1$ and $E_2$, we define $M_{E_1, E_2}$ as the submatrix of $M$ that consists of the rows indexed by $E_1$ and the columns indexed by $E_2$.

Consider $E_1=\{j_1, j_2, \ldots, j_r\} \subseteq [p]$. 
Then we let $\mathbb{S}_{E_1}$ be an $r \times p$ matrix with the $(k, j_k)^{\text{th}}$ entry equal to 1 for each $k \in [r]$ and 0 in the remaining entries, so that $\mathbb{S}_{E_1}M$ extracts the rows of $M$ in the set $E_1$. 
If $E_1=\{j\}$, a singleton set, then we simply denote the $1\times p$ matrix that extracts the $j^{\text{th}}$ row of $M$ as $\mathbb{S}_{j}$.
We denote by $\phi(x ; \mu, \Sigma)$ the Gaussian density function with mean vector $\mu$ and covariance matrix $\Sigma$. 

Let $\widehat{E}$ denote the set of non-zero entries of the solution to the randomized $\ell_1$-penalized SQR problem in \eqref{eq:est_random}.
We denote by $E$ the realized value of $\widehat{E}$ and by $E'$ its complement. 
We let $|E| = q$ and $|E'| = q' = p - q$.

After selection, we consider inference for each component of $\truestar$ in \eqref{eq:infer_target:inf}.
Without loss of generality, let us focus on the $j^{\text{th}}$ component of $\truestar$, denoted by $\truestarj$.
For this selective inference task, we will use SQR estimators, which will be defined as we develop our method in the paper. 
To analyze their distribution for inference, it is necessary to introduce the pseudo-parameters:
\begin{equation*}
    \trueprime = \argmin_{b \in \RR ^ {|E|}} \EE_{\FF_n} \left[\sqrt{n} \widehat Q_{n; h'}(X_E b; Y)\right],
\end{equation*}
where $\widehat Q_{n; h'}(X_E b; Y)$ represents the convolution smoothed quantile loss in \eqref{convoluted:loss} with the bandwidth parameter $h'$. 
Note that the pseudo-parameter vector $\trueprime$ depends on the bandwidth parameter $h'$, which may differ from the bandwidth $h$ used for selection. However, for the sake of simplicity, we choose to suppress this dependence in our notations.

We conclude this section with Proposition \ref{prop:diff_target}, showing that the difference between $\trueprime$ and $\truestar$, the pseudo-parameter and our target parameter, is negligible when $h’=o(n^{-1/4})$, similar to the bandwidth parameter $h$ in \eqref{eq:est_random}.
The implication of this result will become evident in the following section, when we introduce the refitted SQR estimators, based on the selected variables, to construct inference for the post-selection parameter $\truestar$. 

Before we present this result, we state the following assumptions on the data generating distribution. 

\begin{aspt}{A}
\label{aspt:moment_bound}
Let $\sigma \in \RR$ be a constant and let $x\in \RR^p$ represent the $p$-dimensional vector of covariates.
For any $\lambda \in \RR$ and a unit vector $u \in \RR ^ p$, assume that 
\[
	\EE_{\FF} \big[\exp \big(\lambda u ^ \top x\big)\big] 
	\leq \exp \left(\sigma ^ 2 \lambda ^ 2\right).
\]
\end{aspt}

\begin{aspt}{B}
\label{aspt:Lip}
Denote by $f_{y \mid x_E}$ the conditional density of $y$ given $x_E$. 
There exists a constant $l_{0, E} > 0$ such that $\left|f_{y \mid x_E}(u)-f_{y \mid x_E}(v)\right| \leq l_{0, E} |u-v|$ for all $u, v \in \RR$ almost surely (over $x_E$). 
Furthermore, for $m_{0,E}>0$, suppose that
$$\EE_{\mathbb{F}} \left[K_h(y - x_E^ \top \truestar) x_E x_E ^ \top\right] \succeq m_{0,E} > 0.$$
\end{aspt}

The condition in Assumption \ref{aspt:moment_bound} guarantees that the covariates $\{x_i\}_{i = 1} ^ n$ are sub-Gaussian.
This includes variables from common distributions such as Gaussian, Bernoulli, Binomial, and Uniform distributions, among others.
The first part of Assumption \ref{aspt:Lip} assumes that the conditional density of the response variable, given the selected set of covariates, is Lipschitz continuous. 
This holds true, for example, when this density has a bounded first derivative. 
The second part of the assumption is satisfied as long as the residual vector from the population fit $(y - x_E^ \top \truestar)$ is stochastically bounded, and $\EE_{\mathbb{F}} \left[x_E x_E ^ \top\right]$ is a positive definite matrix.

\begin{prop}
\label{prop:diff_target}
Let $\kappa = \int_{-\infty}^{\infty}|u|^2 K(u) d u$.
Under Assumptions \ref{aspt:moment_bound} and \ref{aspt:Lip}, we have that
\[
    \sup_{E} \frac{m_{0,E}}{l_{0,E} ^ 2} \times \big\|\trueprime - \truestar\big\| \leq 6 \kappa  \sigma^4  h ^ {\prime 2},
\]
where $h'$ denotes the bandwidth used for inference and $\|\cdot\|$ represents the $\ell_2$-norm.
\end{prop}

A proof for this result is provided in Appendix \ref{suppl:diff_target} of the Supplementary Material and follows similar arguments as the proof of Proposition 4.1 in \cite{smooth_quantile}.
As an immediate consequence of Proposition \ref{prop:diff_target}, for any fixed $E\subseteq [p]$, we have that 
$$
\sqrt{n}\big\|\trueprime - \truestar\big\| = o(1) 
$$
as long as $h' = o(n^{-1/4})$.

In the fixed $p$ and growing $n$ regime, we note that the bandwidth parameters $h$ and $h'$, which are used at the time of selection and inference, can have the same asymptotic order. 
For ease of exposition, we will hereafter assume that $h'=h$ in the subsequent theoretical results; see Remark \ref{hh}. 

With $h'=h$, we consider
\begin{equation}
\label{conv:smooth:target}
    \true = \argmin_{b \in \RR ^ {|E|}} \EE_{\FF_n} \left[\sqrt{n} \widehat Q_{n; h}(X_E b; Y)\right]
\end{equation}
as the pseudo-parameter vector, i.e., $\trueprime=\true$ when $h'=h$.
Let $\truej$ be the $j^{\text{th}}$ component of $\true$. The dependence of the  pseudo-parameter on $h$ is again suppressed in the notation here.

\subsection{The SQR estimators}
\label{subsec:sqr_estimators}

Now we introduce the refitted SQR estimators based on the selected variables, which are used in constructing our pivot.

Define the refitted SQR estimator obtained from regressing $Y$ against $X_E$ as
\begin{equation}
\label{eq:MLE}
	\mle = \argmin_{b \in \RR ^ {q}} \sqrt{n} \widehat Q_{n;h} (X_E b; Y).
\end{equation}
Let $\mlej= \mathbb{S}_j \mle$ denote the $j^{\text{th}}$ component of this estimator. 
We make two remarks here. 
\begin{rem}
The estimator $\mle$  depends on the bandwidth parameter $h$.
As in the case of the pseudo-parameter vector in \eqref{conv:smooth:target}, we have again suppressed the explicit dependence of the refitted SQR estimator on the value of the bandwidth parameter.
\end{rem}

\begin{rem} \label{hh}
Recall from the previous section that we fix $h' = h$ and develop the results in the main paper under the assumption that the bandwidth parameters, used for selection and for inference (after selection), are the same. 
The more general results involving the refitted SQR estimator with the bandwidth parameter $h'$, allowing for $h$ and $h'$ to differ, are presented in Appendix \ref{suppl:general} of the Supplementary Material.
\end{rem}

Unlike the case with a prefixed set $E$, we need to specify additional statistics to carry out our post-selection inferential task, because the set of variables selected by the randomized $\ell_1$-penalized SQR problem does not depend solely on $\mle$.
To this end, let
\begin{align*}
\begin{gathered}
	J =  \EE_{\FF_n} \left[X ^ \top \nabla ^ 2 \widehat Q_{n;h} (X_E \true; Y) X\right] 
    = \begin{bmatrix} J_{E, E} & J_{E, E'} \\ J_{E', E} & J_{E', E'}\end{bmatrix} \\
	H = \Cov_{\FF_n} \Big(\sqrt{n} X ^ \top \nabla \widehat Q_{n;h} (X_E \true; Y), \sqrt{n} X ^ \top \nabla \widehat Q_{n;h} (X_E \true; Y)\Big) 
    = \begin{bmatrix} H_{E, E} & H_{E, E'} \\ H_{E', E} & H_{E', E'}\end{bmatrix}
\end{gathered}	
\end{align*}
be two matrices in $\RR ^ {p\times p}$, based on moments of the Hessian and the gradient of the SQR loss, and let $\varmle = J_{E, E} ^ {-1} H_{E, E} J_{E, E} ^ {-1}$ and let $\varmlej$ be the $j^{\text{th}}$ diagonal entry of $\Sigma_{E, E}$.
Now, define the estimator
$$
    \mlenj =\begin{pmatrix} \mathbb{S}_{[E]\setminus j}\left(\mle - \frac{1}{\varmlej}\covmlej\mlej\right) \\ X_{E'} ^ \top \nabla \widehat Q_{n;h}(X_E \mle; Y) + (H_{E', E} H_{E, E} ^ {-1} J_{E, E} - J_{E', E}) \mle \end{pmatrix},
$$
which depends on  both $\mle$ and the variables that were not selected.
As we will see later, conditioning on $\mlenj$ enables us to get rid of the nuisance parameters, 
$$
    \truenj =\begin{pmatrix} \mathbb{S}_{[E]\setminus j}\left(\true - \frac{1}{\varmlej}\covmlej\truej\right) \\ \EE \left[X_{E'} ^ \top \nabla \widehat Q_{n;h}(X_E \mle; Y)\right] + (H_{E', E} H_{E, E} ^ {-1} J_{E, E} - J_{E', E}) \true \end{pmatrix},
$$
at the time of performing selective inference and leads to an asymptotic pivot for our target parameter.
Hereafter, we will refer to $\mle$ and $\mlenj$ as the refitted SQR estimators.

We start with an asymptotic representation for $\mlej \in \RR$  and $\mlenj \in \RR^{p-1}$, which holds for any fixed $E\subseteq [p]$ and $j\in E$. 
This asymptotic representation involves the $j^{\text{th}}$ component of the pseudo parameter, $\truej$, as defined in \eqref{conv:smooth:target}.

\begin{prop}
\label{prop:beta_gamma_expression}
Let $E$ be a fixed subset of $[p]$.
Define 
\begin{align*}
\begin{gathered}
\Lambda_1 = \begin{bmatrix} - \mathbb{S}_j J_{E,E} ^ {-1}~ & 0 ^ \top_{p-q} \end{bmatrix} H ^ {1/2} \in \RR^{1\times p},\\
\Lambda_2 =\begin{bmatrix} \Lambda_{2,1} \\ \Lambda_{2,2}\end{bmatrix}= \begin{bmatrix}  
		 \mathbb{S}_{[E] \setminus j}\left(\frac{1}{\varmlej} \covmlej \mathbb{S}_j J_{E, E} ^ {-1}  -  J_{E, E} ^ {-1}\right) & 0_{q-1, p-q} \\
		- H_{E', E} H_{E, E} ^ {-1} & I_{p-q, p-q} \end{bmatrix} 
		H ^ {1/2} \in \RR^{(p-1) \times p}.
\end{gathered}
\end{align*}		
For $j \in E$, it holds that
\begin{equation*}
	\sqrt{n} \begin{pmatrix}
		\mlej -  \truej \\ 
		\mlenj - \truenj
	\end{pmatrix}
	= \begin{pmatrix} \Lambda_1 \\ \Lambda_2 \end{pmatrix} \Up_n + \bar\Delta_1,
\end{equation*}
where 
\begin{align*}
\begin{gathered}
    \Up_n = \sqrt{n} H ^ {- 1/2} \left( \begin{pmatrix}  X_E^\top \nabla \widehat Q_{n;h} (X_E \true; Y)\\ X_{E'}^\top \nabla \widehat Q_{n;h} (X_E \true; Y) \end{pmatrix}- \begin{pmatrix} 0_q \\ \EE_{\FF_n} [X_{E'} ^ \top \nabla \widehat Q_{n;h} (X_E \true; Y)] \end{pmatrix}\right), \\ 
\bar\Delta_1 = o_p(1).
\end{gathered}
\end{align*}
\end{prop}

Combining Proposition \ref{prop:diff_target} with the above-stated representation for the refitted SQR estimators, we have an immediate corollary.
Note that we have a similar asymptotic expansion for the refitted SQR estimator $\mlej$ when centered at $\truestarj$, rather than at the pseudo parameter $\truej$.
\begin{cor}
\label{cor:beta_gamma_expression}
Suppose that the conditions in Assumptions \ref{aspt:moment_bound} and \ref{aspt:Lip} are met.
Then, we have
\begin{equation*}
	\sqrt{n} \begin{pmatrix}
		\mlej -  \truestarj \\ 
		\mlenj - \truenj
	\end{pmatrix}
	= \begin{pmatrix} \Lambda_1 \\ \Lambda_2 \end{pmatrix} \Up_n + \Delta_1,
\end{equation*}
where $\Up_n$ is as defined in Proposition \ref{prop:beta_gamma_expression} and $\Delta_1=o_p(1)$.
\end{cor}

As a follow-up result, we also have the following asymptotic normal distribution for the refitted SQR estimators.

\begin{cor}
\label{cor:asy_distribution}
For a fixed set $E$, we have
\[
    \sqrt{n} \begin{pmatrix}
		\mlej -  \truestarj \\ 
		\mlenj - \truenj
	\end{pmatrix} 
    \Rightarrow
    \cN \left(0_{p},
    \begin{bmatrix}
        \varmlej & 0_{1, p-1} \\
        0_{p-1, 1} & \Lambda_2 \Lambda_2 ^ \top 
    \end{bmatrix}
    \right).
\]
\end{cor}

\begin{rem}
As mentioned earlier, a similar asymptotic representation for the refitted SQR estimators holds with minor adjustments, even if $h$ and $h'$ differ. The more general version of Proposition \ref{prop:beta_gamma_expression}, along with the corresponding corollaries, is provided in Appendix \ref{suppl:general} of the Supplementary Material; refer to Proposition \ref{prop:beta_gamma_expression_general}, Corollary \ref{cor:beta_gamma_expression_general} and Corollary \ref{cor:asy_distribution_general}.
The difference between the representations in Proposition \ref{prop:beta_gamma_expression_general}  and Proposition \ref{prop:beta_gamma_expression} is now reflected in the matrices $\Lambda'_1$ and $\Lambda'_2$, which reduce to $\Lambda_1$ and $\Lambda_2$, as stated in Proposition \ref{prop:beta_gamma_expression}, when $h = h'$.
\end{rem}

Our next result in Proposition \ref{prop:taylor_expansion} establishes the connection between the randomized $\ell_1$-penalized SQR estimators and the refitted SQR estimators based on the selected set of variables, $\mlej$ and $\mlenj$.

Fixing some more notations, let $\lasso \in \RR ^ {q \times 1}$ denote the nonzero components of the randomized $\ell_1$-penalized SQR estimator from solving \eqref{eq:est_random}.
Let $|\lasso|$ denote the absolute values of $\lasso$ in a component-wise sense, and let
$$
    \D = \lambda \cdot \begin{pmatrix} \Sgn \\ \Z \end{pmatrix} \in \RR^{p\times 1}
$$
denote the subgradient of the $\ell_1$-penalty at the solution, where 
 $\Sgn \in \RR ^ {q \times 1}$ represents the sign vector of $\lasso$, and $\Z \in \RR^{q'\times 1}$ represents the inactive components of the subgradient vector that satisfies the constraint $\|\Z\|_{\infty} \leq 1$. 

Define the matrices
\begin{align*}
\begin{gathered}
	M ^ j = - \frac{1}{\varmlej} \begin{bmatrix} H_{E, E} \\ H_{E', E} \end{bmatrix} J_{E, E} ^ {-1} \mathbb{S}_j ^ \top  \in \RR ^ {p \times 1}, \ 
    N ^ j = \begin{bmatrix} - J_{E, E} \mathbb{S}_{[E]\setminus j} ^ \top & 0_{q, q'} \\ - H_{E', E} H_{E, E} ^ {-1} J_{E, E} \mathbb{S}_{[E]\setminus j} ^ \top & I_{q', q'}\end{bmatrix} \in \RR ^ {p \times(p-1)},\\
	T = \begin{bmatrix} J_{E, E} \\ J_{E', E} \end{bmatrix} \text{\normalfont Diag}(\sgnr) \in \RR ^ {p \times q},
\end{gathered}	
\end{align*}
where $\sgnr$ is the observed value $\Sgn$, the signs of the $\lasso$.

\begin{prop}
\label{prop:taylor_expansion}
Given $E$ and $\sgnr$, the observed selected set of variables and the signs for the nonzero components of the randomized $\ell_1$-penalized SQR estimator, we have that
\begin{equation*}
	T \sqrt{n} |\lasso| + \D  = \sqrt{n} \rand - M ^ j \sqrt{n} \mlej - N ^ j \sqrt{n} \mlenj + \Delta_2,
\end{equation*}
where $\Delta_2 = o_p(1)$.
\end{prop}

\begin{rem}
Similar to the asymptotic representation of the refitted SQR estimators, their relationship with the randomized $\ell_1$-penalized SQR estimator, as presented in Proposition \ref{prop:taylor_expansion}, holds with slightly modified matrices when the smoothing parameters $h$ and $h'$ differ. Specifically, Proposition \ref{prop:taylor_expansion_general} in Appendix \ref{suppl:general} of the Supplementary Material establishes this relationship using the modified matrices $T'$, $M^{\prime j}$, and $N^{\prime j}$.
It is straightforward to observe that the more general representation simplifies to the one presented above when $h = h'$.
\end{rem}

\subsection{Our pivot}

Constructing a pivot for $\truestarj$ is a natural means to inference, allowing for the construction of confidence intervals or p-values for hypotheses related to this parameter.

First, we turn to our conditioning event.
Let
\begin{align*}
\begin{gathered}
	\Ej = T ^ \top \varrand ^ {-1} M ^ j \in \RR ^ {q \times 1}, \ \ 
\Psi = \left\{T ^ \top \varrand ^ {-1} T\right\} ^ {-1} \in \RR ^ {q \times q}.
\end{gathered}	
\end{align*}
Then, based on the nonzero components of the randomized $\ell_1$-penalized SQR estimator $\lasso$, define the variables
\begin{align*}
\label{def:UV}
\begin{gathered}
   \Uj = \sqrt{n}\Ej ^ \top |\lasso| \in \RR,\\
   \Vj = \sqrt{n}\bigg(I_{q, q} - \frac{\Psi \Ej}{\Ej ^ \top \Psi \Ej} \Ej ^ \top\bigg) |\lasso| \in \RR ^ {q \times 1}.
\end{gathered}
\end{align*}

Proposition \ref{prop:conditional_event} states  our conditioning event and provides a characterization of this event in terms of the randomized $\ell_1$-penalized SQR estimators.

\begin{prop}
\label{prop:conditional_event}
Fix $j\in E$.
We have that
\begin{equation*}
	\left\{\D = \dr, \Vj= \vjr\right\}
    = \left\{I^j_1 \leq \Uj \leq I^j_2, \Z = \zr, \Vj = \vjr\right\},
    \quad\quad
\end{equation*} 
where 
\[
\begin{aligned}
	I^j_1 = -\min_{k: \mathbb{S}_k  \Psi \Ej > 0}  \frac{\Ej ^ \top \Psi \Ej}{\mathbb{S}_k  \Psi \Ej} \mathbb{S}_k \vjr,
    \quad
    I^j_2 = -\max_{k: \mathbb{S}_k  \Psi \Ej < 0} \frac{\Ej ^ \top \Psi \Ej}{\mathbb{S}_k  \Psi \Ej} \mathbb{S}_k \vjr.
\end{aligned}
\]  
\end{prop}

There are a few things to note about our conditioning event in the above-stated result. 

For selective inference to be valid, we must condition on a subset of our selection event $\{\widehat{E}=E\}$.
In order to achieve a tractable conditional distribution for this purpose and an easy-to-calculate pivot from it, the subset that we choose must also admit a simple description in the involved estimators.
The conditioning event on the left-hand side of the claim in Proposition \ref{prop:conditional_event} satisfies both.
Firstly, it is easy to see that this event is a proper subset of $\{\widehat{E}=E\}$. 
Secondly, we can represent the same event by truncating $\Uj$ to a simple interval, as shown on the right-hand side of the claim in Proposition \ref{prop:conditional_event}.
This simplification reduces our event to an interval on the real line, making it possible to  construct a pivot in closed form.

We are now ready to present our pivot informally.
Let $I_1^j$ and $I_2^j$ be constants as defined in Proposition \ref{prop:conditional_event}.
For $z_1 \in \RR$ and $z_2 \in \RR ^ {p - 1}$, define
$W_0: \RR \times \RR ^ {p - 1} \rightarrow \RR$ as:
\begin{equation}
\label{eq:defn_W0}    
    W_0 (z_1, z_2)
    = \int_{I^j_1} ^ {I^j_2} \phi \left(Q ^ j t + M ^ j z_1 + N ^ j z_2 + P ^ j; 0_p, \varrand\right) dt,
\end{equation}
where $Q ^ j \in \RR ^ {p \times 1}$ and $P ^ j \in \RR ^ {p \times 1}$ are equal to:
\begin{align*}
\begin{gathered}
	Q ^ j = \frac{T \Psi \Ej}{\Ej ^ \top \Psi \Ej},
	\quad
	P ^ j = T \vjr + \dr.
\end{gathered}
\end{align*}
Then, our pivot for $\truestarj$ is equal to: 
\begin{equation}
\label{pivot}
    \mathrm{Pivot} ^ {j\cdot E} (\sqrt{n}\mlej, \sqrt{n}\mlenj; \truestarj)= \dfrac{
	\bigintss_{-\infty} ^ {\sqrt{n} \mlej} \phi \left(x; \sqrt{n} \truestarj, \varmlej\right)
	W_0 \left(x, \sqrt{n} \mlenj\right) dx}
	{
	\bigintss_{-\infty} ^ \infty \phi \left(x; \sqrt{n} \truestarj, \varmlej\right)
	W_0 \left(x, \sqrt{n} \mlenj\right) dx
	}.
\end{equation}

The form of our pivot in \eqref{pivot} suggests that an adjustment for the selection process is achieved by multiplying the weight function $W_0(\cdot, \sqrt{n} \mlenj)$ to the marginal normal density of the refitted SQR estimator $\sqrt{n}\mlej$. 

We highlight the role of the smoothing bandwidth parameters, used for variable selection and for obtaining the refitted SQR estimators, in the construct of our pivot.

\begin{rem}
We note that when the two smoothing bandwidth parameters are equal (i.e., $h = h'$), the matrices in Proposition \ref{prop:beta_gamma_expression} and \ref{prop:taylor_expansion} can be directly applied to construct the weight function $W_0$. 
If $h$ and $h'$ differ, then the modified matrices provided in Propositions \ref{prop:beta_gamma_expression_general} and \ref{prop:taylor_expansion_general} in the Supplementary Material are used to construct this weight function.
Importantly, the form of our pivot remains unchanged in both scenarios.
\end{rem}

At last, consider the standardized variable
$$
\Zn =\sqrt{n} \begin{pmatrix} \Lambda_1 \\ \Lambda_2 \end{pmatrix}^{-1}  \begin{pmatrix}
		\mlej -  \truestarj \\ 
		\mlenj - \truenj
	\end{pmatrix}.
$$
Alternatively, we can express the pivot in \eqref{pivot} in terms of the standardized variable $\Zn$ as:
\begin{equation}
\label{pivot:std}
    \cP ^ {j\cdot E} (\Zn)
    = \dfrac{
	\bigintsss_{-\infty} ^ {\Lambda_1 \Zn + \sqrt{n} \truestarj} \phi \left(x; \sqrt{n} \truestarj, \varmlej\right)
	W_0 \left(x, \Lambda_2 \Zn + \sqrt{n} \truenj\right) dx}
	{
	\bigintsss_{-\infty} ^ \infty \phi \left(x; \sqrt{n} \truestarj, \varmlej\right)
	W_0 \left(x, \Lambda_2 \Zn + \sqrt{n} \truenj\right) dx
	}.
\end{equation}
We use the latter form of our pivot when studying its asymptotic properties. In this form, we mainly analyze the pivot as a function of the standardized variable and therefore omit the parameter in its argument for the sake of simplicity in notation.

\subsection{Link with the least squares regression}

Let us consider the well-studied least squares regression with a fixed design matrix for normal data and homoscedastic errors with variance $\sigma^2$. 
After selecting the subset of variables $E$, our post-selection targets of inference are the components of 
$$
\truestar = \argmin_{b\in \RR^q} \mathbb{E}_{\mathbb{F}_n} \left[\|Y- X_Eb\|^2\right].
$$

In this scenario, we observe that 
$$\mle = (X_E ^ \top X_E) ^ {-1} X_E Y$$
is the least squares estimator using $(Y, X_E)$, and $\mlej$ is its $j^{\text{th}}$ component. 
Note, with the quadratic loss function in the least squares problem, $H = \sigma^2 J$ for the fixed matrix $J = X ^\top X$, and that $\Sigma_{E, E} = \sigma^2 \left(X_E ^\top X_E\right)^{-1}$, and
$$\mlenj = \begin{pmatrix} \mathbb{S}_{[E]\setminus j}\left(\mle - \frac{1}{\varmlej}\covmlej\mlej\right) \\ X_{E'} ^\top \left(X_E\mle - Y\right) \end{pmatrix}.$$
Moreover, in this special setting, the two representations in Propositions \ref{prop:beta_gamma_expression} and \ref{prop:taylor_expansion} are not asymptotic but rather exact.
This is formalized in our next result.
\begin{prop}
\label{prop:ls_exact}
In the setting described above, we have:
$$	\sqrt{n} \begin{pmatrix}
		\mlej -  \truestarj \\ 
		\mlenj - \truenj
	\end{pmatrix}
	= \begin{pmatrix} \Lambda_1 \\ \Lambda_2 \end{pmatrix} \cZ
$$	
where $\cZ \sim N(0, I_{p})$. Furthermore, the randomized $\ell_1$-penalized estimators satisfy
\[
    T \sqrt{n} |\lasso| + \D  = \sqrt{n} \rand - M ^ j \sqrt{n} \mlej - N ^ j \sqrt{n} \mlenj.
\]
\end{prop}

A probability integral transform, applied to the conditional distribution of $\mlej$—the $j^{\text{th}}$ entry of the least squares estimator—conditioned on the event in Proposition \ref{prop:conditional_event}, yields the following pivot.

\begin{prop}
\label{cor:pivot_cdf}
An exact pivot is equal to
\begin{equation*}
\label{eq:pivot_cdf}
	\dfrac{
	\bigintsss_{-\infty} ^ {\sqrt{n} \mlej} \phi \big(x; \sqrt{n} \truestarj, \varmlej\big)
	W_0 \big(x, \sqrt{n} \mlenj\big) dx}
	{
	\bigintsss_{-\infty} ^ \infty \phi \big(x; \sqrt{n} \truestarj, \varmlej\big)
	W_0 \big(x, \sqrt{n} \mlenj\big) d x
	},
\end{equation*} 
or, in terms of the standardized variable $\cZ$, is equal to
\[ \dfrac{
	\bigintsss_{-\infty} ^ {\Lambda_1 \cZ + \sqrt{n} \truestarj} \phi \left(x; \sqrt{n} \truestarj, \varmlej\right)
	W_0 \left(x, \Lambda_2 \cZ+ \sqrt{n} \truenj\right) dx}
	{
	\bigintsss_{-\infty} ^ \infty \phi \left(x; \sqrt{n} \truestarj, \varmlej\right)
	W_0 \left(x, \Lambda_2 \cZ + \sqrt{n} \truenj\right) dx.
	}
\]	
This pivot is distributed as a $\operatorname{Unif} (0, 1)$ random variable conditional on the event in Proposition \ref{prop:conditional_event}. 
\end{prop}

To conclude, Lemma \ref{lem:pivot_equal} in Appendix \ref{appdx:A} of the Supplementary Material confirms that the expression of this pivot matches the one given by \cite{panigrahi2023exact} for the $\ell_1$-penalized least squares regression.

From this, we observe that the pivot introduced in \eqref{pivot} has the same form as the exact pivot in least squares regression with normal data, except that the refitted SQR estimators replace the least squares regression estimators. 
As previously emphasized, our theory remarkably demonstrates that we can draw asymptotically-exact inference from this pivot without imposing strong assumptions on the conditional distribution of the response.

We provide the details in the next section.

\section{Asymptotic theory}
\label{sec:theory}

Consider the object that was defined in \eqref{pivot}. 
Let
\[
	\left[\mathrm{LCB}^\alpha_{n},\, \mathrm{UCB}^\alpha_{n}\right]
	= \left\{ b \in \RR: \frac{\alpha}{2} \leq \mathrm{Pivot}^{j\cdot E}(\sqrt{n} \mlej, \sqrt{n} \mlenj; b) \leq 1 - \frac{\alpha}{2} \right\}
\]
denote the $(1-\alpha)\cdot 100\%$ two-tailed confidence interval for $\truestarj$.
Our main result in this section shows that the confidence interval $\left[\mathrm{LCB}^\alpha_{n},\, \mathrm{UCB}^\alpha_{n}\right]$ has coverage probability converging to $1-\alpha$ as the sample size grows to $\infty$.

Throughout this section, we consider that the parameters in our problem vary with $n$ as: 
$$\sqrt{n}  \begin{pmatrix}  \truestarj ~ & (\truenj)^\top \end{pmatrix}^\top = r_n \beta,$$ for a given vector $\beta$ and $\{r_n: n \in \mathbb{N}\}$ a fixed sequence of non-negative numbers such that $r_n= o(n ^ {1 / 6})$.
The rate $r_n$ controls the probability of the selection event defined in Proposition \ref{prop:conditional_event}. 
This probability remains bounded away from zero when $r_n$ is bounded by a constant $C>0$. 
Asymptotic selective inferential guarantees for events of this type have been studied in previous work by \cite{tibshirani_bootstrap, rasines2021splitting}.
However, the selection probability may converge to zero when $r_n$ grows to $\infty$, and may result in rare selection events. 
As we will see later, our theory allows us to construct asymptotically-exact selective inference even for these rare events, provided that $r_n$ grows at a slower rate than $n ^ {1/6}$.
Before we state this result, we specify some regularity conditions for our theory to be valid.

Consider the representation in Proposition \ref{prop:taylor_expansion}.
Let 
\begin{equation}
    \sqrt{n}\widetilde{\omega}_n = \sqrt{n}\rand + \Delta_2
    \label{tilde:rand}.
\end{equation}
It is evident from the definition that $\sqrt{n}\widetilde{\omega}_n \Rightarrow \cN(0_p, \Omega)$.
\begin{aspt}{C}
Suppose that there exists $n_0$ such that distribution for $\sqrt{n}\widetilde{\omega}_n$ admits a Lebesgue density $q_n$ for all $n\geq n_0$.
For $z= (z_1, z_2)$ where $z_1 \in \RR$ and $z_2 \in \RR^{p-1}$, define
\[
    W_{0, n} \big(z_1, z_2\big)
    = \int_{I^j_1} ^ {I^j_2} q_n \left(Q ^ j t + M ^ j z_1 + N ^ j z_2 + P ^ j; 0_p, \varrand\right) dt,
\]
Then, assume that
\[
    \lim_n \sup_{z} \left|\frac{W_{0, n} (z)}{W_0 (z)} - 1\right| = 0.
\]
\label{aspt:randomization}
\end{aspt}

\begin{aspt}{D}
Assume for $\Delta_1$ defined in Corollary \ref{cor:beta_gamma_expression} that 
\[
    \lim_n \dfrac{1}{r_n^2}  \log \mathbb{P}_{\FF_n}\left[ \frac{1}{r_n}\|\Delta_1\|> \epsilon \right]= -\infty
\]
holds for any $\epsilon>0$, if $r_n\to \infty$ as $n\to \infty$.
\label{aspt:error}
\end{aspt}

Note that the conditions in Assumption \ref{aspt:randomization} and Assumption \ref{aspt:error} place some mild constraints on the errors in the asymptotic representations provided in Proposition \ref{prop:taylor_expansion} and Corollary \ref{cor:beta_gamma_expression}, respectively.
Assumption \ref{aspt:randomization} controls the behavior of $\sqrt{n}\widetilde{\omega}_n$ in \eqref{tilde:rand}, ensuring that this sequence of variables exhibits similar behavior as its limiting normal counterpart $\sqrt{n}\rand$.
Assumption \ref{aspt:error}, on the other hand, is only required when $r_n$ grows to $\infty$ with increasing sample size, and in this scenario, it imposes a condition on the rate of convergence of the error variable $\Delta_1$.

\begin{thm}
\label{thm:main_weak_convergence}
For each $n\in \NN$, denote by $\cF_n=\{\FF_n\}$ a collection of data-generating distributions satisfying Assumptions \ref{aspt:moment_bound}, \ref{aspt:Lip}, \ref{aspt:randomization} and \ref{aspt:error}. 
Then, we have
$$
\lim_{n} \sup_{\FF_n \in \cF_n} \Big|\mathbb{P}_{\FF_n} \left[\truestarj \in  \left[\mathrm{LCB}^\alpha_{n},\, \mathrm{UCB}^\alpha_{n}\right] \Big\lvert \event\right] -(1-\alpha)\Big| =0.
$$
\end{thm}

We make a few comments on the implications of our main result. 
\begin{enumerate}
\item First, note that the statement regarding the coverage of our confidence intervals is conditional on 
$$\event,$$
the event that was stated in Proposition \ref{prop:conditional_event}.
Since this event is a strict subset of $\big\{\widehat{E} =E\big\}$, the same guarantee applies to the larger event based on the total law of probability.

\item Second, we observe that the guarantees of inference in our theory are strong in that they not only ensure asymptotically-exact inference on an average for a range of selection events, but also for the specific selection event that was observed with the data at hand.
The guarantees of the former type, though, do not apply to any individual event or specifically the event observed in the dataset.

\item Third, the asymptotically-exact inference we offer holds uniformly across all distributions in the collection $\cF_n$. 
This ensures that, for any given value of $\epsilon > 0$, there exists a number $N(\epsilon)$ such that, for all $n \ge N(\epsilon)$, the confidence intervals obtained by inverting $$\mathrm{Pivot} ^ {j\cdot E} (\sqrt{n} \mlej, \sqrt{n} \mlenj; \truestarj)$$ will provide coverage of at least $1-\alpha-\epsilon$ post selection, regardless of the data-generating distribution in $\cF_n$. 
\end{enumerate}

The empirical results in the next section support our theory and show that the proposed confidence intervals achieve the desired coverage for the quantile parameters using the check loss function.
In the remaining section, we state two key results that establish the validity of our main theorem.

For the proof of our main result, we rely on the asymptotic pivot presented in terms of the standardized variable $\Upsilon_n$ as defined in Proposition \ref{prop:beta_gamma_expression}. Fixing some more notations, we express the weight function $W_0(\cdot)$ used in our pivot directly in terms of $\Upsilon_n$, and compactly rewrite it as
\[
	W_0 \big(\Upsilon_n\big)
	= \int_{I_1 ^ j} ^ {I_2 ^ j} \phi \big(Q ^ j t + \mn ^ j \Upsilon_n + \tail; 0_p, \varrand\big) dt,
\]
where  $\mn ^ j \in \RR ^ {p \times p}$ and $\tail \in \RR ^ {p \times 1}$ are defined as
\[
	\mn ^ j = M ^ j \Lambda_1 + N ^ j \Lambda_2,
	\quad \tail = P ^ j + M ^ j \sqrt{n} \truestarj + N ^ j \sqrt{n} \truenj.
\]

In order to prove Theorem \ref{thm:main_weak_convergence}, we first provide sufficient conditions for our main result on the asymptotic coverage guarantees with our pivot.

We define $\EE_{\cN} \big[g\left(\cZ\right)\big]$ as the expectation obtained by replacing $\Upsilon_n$ in $\EE_{\FF_n} \big[g \left(\Upsilon_n\right)\big]$ with the variable $\cZ$ from Proposition \ref{prop:ls_exact}.
Note that $\cN$ in the subscript of this expectation highlights that $\cZ$ is a normal variable.

\begin{thm}
\label{thm:relative_diff}
Let $\widetilde{D}_n$ be an increasing sequence of sets in $\mathbb{R}^p$ such that 
$$\lim_n \sup_{\FF_n\in \cF_n} \mathbb{P}_{\mathbb{F}_n}[\Upsilon_n \in \widetilde{D}_n^c] = 0.$$
For $\arb \in \mathbb{C} ^ 3 (\RR, \RR)$ an arbitrary function with bounded derivatives up to the third order, define
\begin{equation*}
\begin{gathered}
\RD ^ {(1)} = \dfrac{\Big|\EE_{\FF_n} \big[W_0 \left(\Upsilon_n\right)\boldone_{ \widetilde{D}_n} (\Upsilon_n)\big] - \EE_{\cN} \big[W_0 \left(\cZ\right)\boldone_{\widetilde{D}_n} (\cZ)\big]\Big|}
	{\EE_{\cN} \left[W_0 \left(\cZ\right)\right]} \\[7pt]
\RD ^ {(2)} = \dfrac{\Big|\EE_{\FF_n} \big[\arb \circ \pivotup \times W_0 \left(\Upsilon_n\right)\boldone_{\widetilde{D}_n}(\Upsilon_n)\big] - \EE_{\cN} \big[\arb \circ \pivotz \times W_0 \left(\cZ\right)\boldone_{\widetilde{D}_n}(\cZ)\big]\Big|}
	{\EE_{\cN} \left[W_0 \left(\cZ\right)\right]}
\end{gathered}
\end{equation*}
Suppose that
\begin{equation*}
\begin{gathered}
	\lim_{n} \sup_{\FF_n \in \cF_n}  \RD ^ {(1)} =0, \quad
	 \lim_{n} \sup_{\FF_n \in \cF_n} \RD ^ {(2)}  =0.
\end{gathered}
\end{equation*}
Then, under Assumptions \ref{aspt:randomization} and \ref{aspt:error}, the assertion stated in Theorem \ref{thm:main_weak_convergence} holds. 
\end{thm}

In the next step of our proof, we focus on proving the sufficient conditions in Theorem \ref{thm:relative_diff} that ensure the validity of our main result.

\begin{thm}
\label{thm:rd_limit}
Under the conditions stated for Theorem \ref{thm:main_weak_convergence}, we have 
\[
	\lim_{n} \sup_{\FF_n \in \cF_n}  \RD ^ {(1)} =0, \quad
	 \lim_{n} \sup_{\FF_n \in \cF_n} \RD ^ {(2)}  =0.
\]
\end{thm}
Details for the proof of Theorem \ref{thm:rd_limit} are provided in Appendix \ref{appdx:B} of the Supplementary Material.

\section{Simulation study}
\label{sec:simulation}

In this section, we investigate the quality of selective inference generated by our pivot on simulated data. 
We generate datasets with $n = 800$ i.i.d. observations and a set of $p = 200$ potential covariates from the following models.
\begin{enumerate}[label={Model \arabic*:}, leftmargin=*]
    \item $y_i = \beta_0 + x_i ^ \top \beta + \varepsilon_i$, where $\beta_0 = 0.2$, $\beta = (c, c, c, c, c, 0 \ldots, 0) ^ \top \in \RR ^ p$, $\varepsilon_i \sim \cN(0, 4)$ and $x_i \sim \cN(0, \Sigma)$ for an autoregressive design matrix with $\Sigma_{j, k} = 0.5 ^ {|j - k|}$, where $\varepsilon_i $ and $x_i$ are independent. 
	\item $y_i = \beta_0 + x_{i, 1} ^ \top \beta + 1.5 x_{i, 2} \varepsilon_i$, where $\beta_0 = 0.2$, $\beta = (c, c, c, c, c, 0 \ldots, 0) ^ \top \in \RR ^ p$ and $\varepsilon_i \sim \cN(0, 4)$. In this model, $x = (x_{i, 2}, x_{i, 1} ^ \top) ^ \top \in \RR ^ p$, where the variable $x_{i, 2} \in \RR$ is drawn from $U(0 ,2)$ and the variables $x_{i, 1} \in \RR ^ {p - 1}$ are drawn independently from $x_{i, 2}$ and from $\cN(0, \Sigma)$ for an autoregressive design matrix with $\Sigma_{j, k} = 0.5 ^ {|j - k|}$.
	\item $y_i = \beta_0(u_i) + x_i ^ \top \beta(u_i)$ for $i=1, \ldots, n$, where $\beta_0(u_i) = 2 c u_i$ and $\beta(u_i) = (c u_i, c u_i, c, c, c, 0 \ldots, 0) ^ \top \in \RR ^ p$, $u_i \sim U(0, 1)$. 
	In this model, $x_{i, 1}, x_{i, 2} \in \RR$, the first and second elements of $x_i$, are drawn from $U(0, 2)$ and $x_{i, 3} \in \RR ^ {p - 2}$, the subvector of $x_i$ after removing $x_{i, 1}, x_{i, 2}$, is drawn from $\cN(0, \Sigma)$ for an autoregressive design matrix $\Sigma_{j, k} = 0.5 ^ {|j - k|}$.
\end{enumerate}

We draw selective inference in each of our models after estimating $F_{y|x} ^ {-1} (\tau)$, which is the $\tau$-th population conditional quantile of $y$ given $x$. 
Note that $F_{y|x} ^ {-1} (\tau) = \beta_0 (\tau) + x ^ \top \beta (\tau)$, where in 
\begin{equation}
\begin{gathered}
\text{ Model 1: } \beta_0 (\tau) = 0.2 + \Phi ^ {-1} (\tau; 0, 4) \text{ and } \beta (\tau) = (c, c, c, c, c, 0 \ldots, 0) ^ \top, \\
\text{ Model 2: } \beta_0 (\tau) = 0.2 \text{ and } \beta (\tau) = (1.5 \Phi ^ {-1} (\tau; 0, 4), c, c, c, c, c, 0 \ldots, 0) ^ \top, \\ 
\text{ Model 3: } \beta_0 (\tau) = 2 c \tau \text{ and } \beta (\tau) = (c \tau, c \tau, c, c, c, 0 \ldots, 0) ^ \top. 
\end{gathered}
\end{equation}

In our simulation, we set the quantile level at $\tau = 0.7$. 
We apply the randomized $\ell_1$-penalized SQR problem with tuning parameter $\lambda = 0.6 \sqrt{\log p / n}$.
To construct the smoothed quantile loss functions, both during variable selection and for obtaining the refitted SQR estimators for inference, we use Gaussian kernels with the bandwidth parameter
$$h = \max \big\{0.05, \sqrt{\tau(1-\tau)} (\log (p) / n)^{1 / 4}\big\}.$$
Using these formulae, $\lambda=0.049$ and $h=0.131$ in our simulations.
Our randomized method is implemented with white noise $\omega_n$
drawn from $\cN(0_p, \Omega/n)$ with $\varrand = I_{p, p}$.

In our simulation, the signal strength settings are categorized as ``Low'', ``Medium'',  and ``High'', depending on the value of $c$ from the set $\{0.1, 0.5, 1\}$.
The reported findings are based on $500$ independent Monte Carlo datasets for each pair of model and signal setting.

Additional simulations on a high dimensional setting with $p>n$ and a misspecified model are included in Appendix \ref{appdx:simulation} of the Supplementary Material.

\subsection{Coverage rates}
\label{subsec:coverage}

We start by assessing the coverage properties of intervals produced by our method, which is labeled as ``Proposed'' in the plots.

We present comparisons between the proposed method and two common baselines: (1) ``Naive'': which utilizes all the data for model selection and reuses the same data for inference without accounting for the double usage of data; (2) ``Splitting'': which divides the data into two independent parts, using  two-thirds of our data for selection and reserving the other third exclusively for inference.

For a prespecified significance level $\alpha = 0.1$ and an interval 
$\left[\mathrm{LCB}^\alpha_{n},\, \mathrm{UCB}^\alpha_{n}\right]$ produced by each method, we compute the coverage rate for the selected population parameters defined as:
\[
	\frac{\left|\left\{j \in E: \truestarj \in  \left[\mathrm{LCB}^\alpha_{n},\, \mathrm{UCB}^\alpha_{n}\right]\right\} \right|}{\max (|E|, 1)}.
\]

The coverage rates of different methods in three models and signal settings are depicted in Figure \ref{fig:coverage}. 
The gray dashed line in the figure represents the prespecified target coverage rate, which is fixed at $0.9$, and the diamond marks highlight the mean coverage rates over all replications.

\begin{figure}[!t]
	\centering
 	\includegraphics[scale = 0.62]{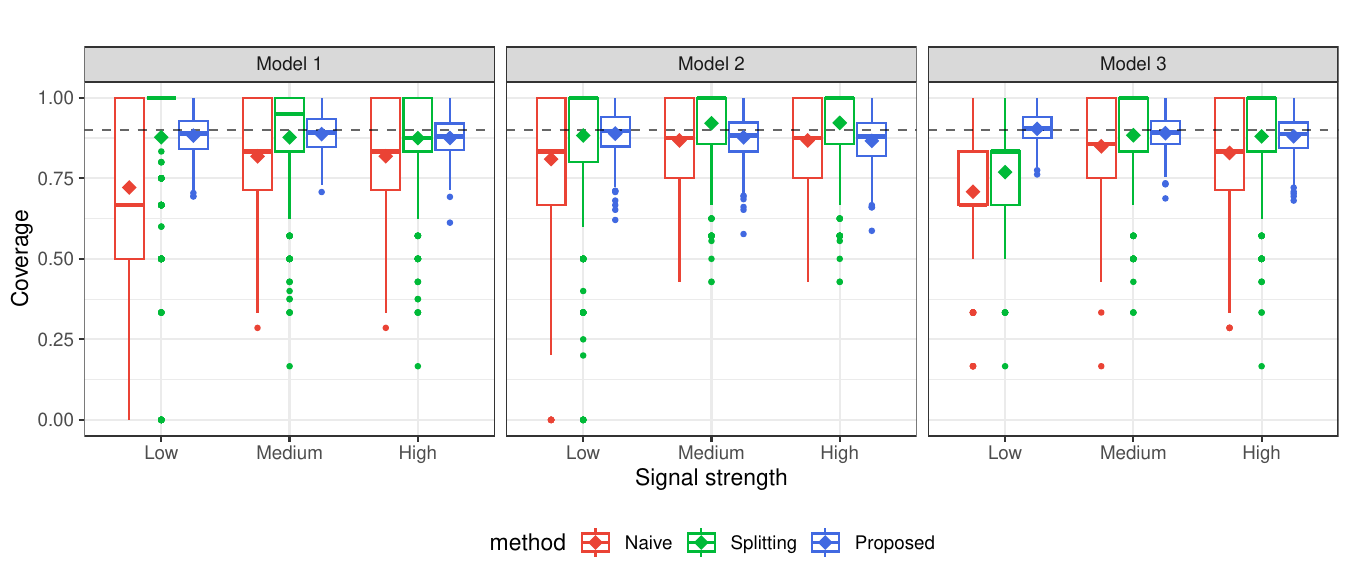}
	\caption{
    \small{Coverage rates of different methods across different models and signal settings. 
    The gray dashed line represents the prespecified target coverage rate at $0.9$, and the diamond marks highlight the averaged coverage rates over all replications. 
    We observed that the ``Proposed'' method consistently achieves the target coverage rate across all scenarios, whereas ``Naive'' and ``Splitting'' underperform.}}
	\label{fig:coverage}
\end{figure}

It is evident that the ``Naive'' method falls remarkably short of coverage, emphasizing the need to account for the impact of selection while constructing inference for the effects of the selected variables on the conditional quantile function. 
The ``Splitting'' method achieves the target coverage rate in all scenarios, except in the ``Low'' signal setting for Model 3. 
The poor performance of ``Splitting'' in this scenario can be attributed to the lack of sufficient data needed to attain the asymptotic coverage rate. 

Remarkably, the ``Proposed'' method achieves the desired rate of coverage across all scenarios, including the challenging setting where ``Splitting'' performs poorly.

Unlike ``Splitting'', our approach utilizes the entire dataset for making inference, achieving the asymptotic approximation even with moderately sized datasets.
Furthermore, it is apparent that the variability of coverage rates in the proposed intervals is lower than that of ``Splitting'', indicating a more stable performance across experiments.

\subsection{Inferential power}

We will now compare inferential power between ``Proposed'' and ``Splitting'', leaving out the invalid ``Naive'' from our comparison.
Although not the exclusive metric, lengths of intervals are a frequently employed and practical way to gauge inferential power.

Box plots for the ratio of average interval lengths:
 \[
    \text{Ratio}= \frac{\text{Average intervals length of ``Proposed''}}{\text{Average intervals length of ``Splitting''}}
\]
are shown in Figure \ref{fig:length}.

\begin{figure}[!t]
	\centering
	\includegraphics[scale = 0.62]{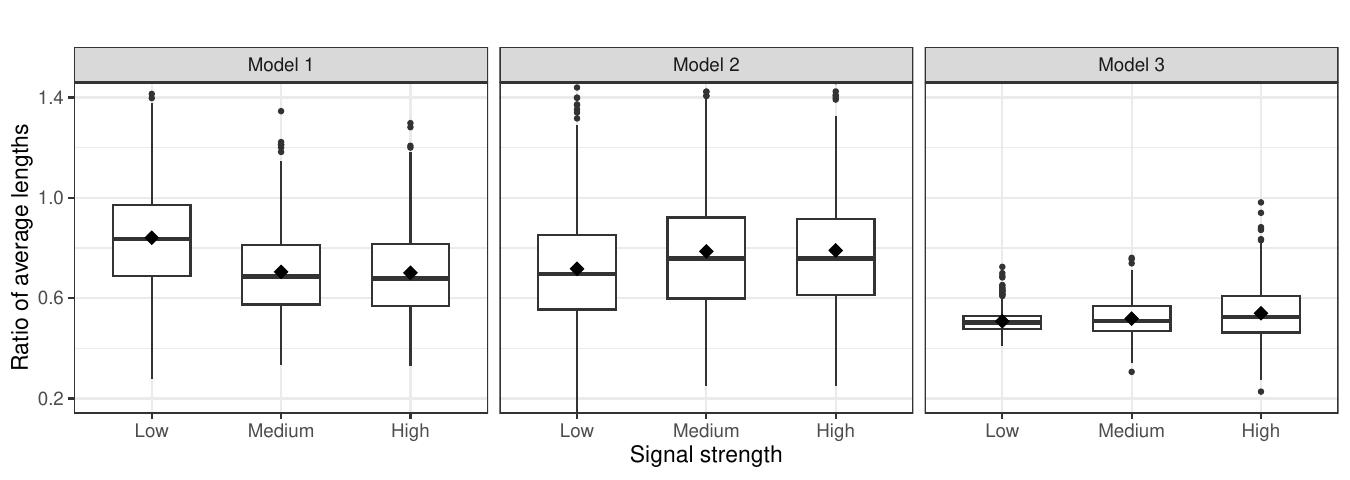}
	\caption{
    \small{The boxplots present the ratio of average interval lengths for the selected parameters between the ``Proposed'' method and the ``Splitting'' method across different models and signal strengths. ``Proposed'' yields significantly shorter intervals than ``Splitting'' in all settings.}}
	\label{fig:length}
\end{figure}

This plot shows that the ``Proposed'' method outperforms ``Splitting'' by consistently producing shorter intervals in all three signal settings and models.
Note that in the challenging setting of ``Model 3'' and the ``Low'' signal regime, our intervals not only give valid inference, but they are also almost half the length of the intervals produced by ``Splitting''.

\subsection{Estimation accuracy}
Finally, we evaluate whether the proposed selective inferential method improves the accuracy of estimating the true signal structure in our data. 
We consider two stages here, the selection stage (before applying ``Proposed'') and the inference stage (after applying ``Proposed'').

To measure the overall accuracy, we compute the F1 score which is defined as 
\[
    \text{{F1 score}} = \dfrac{\text{{True Positives}}}{\text{True Positives}+\dfrac{1}{2}(\text{False Positives}+ \text{False Negatives})},
\]
which can be evaluated for the selection stage before inference is made, or  evaluated after inference is made.  Both F1 scores for the proposed method are reported in Table \ref{table:F1scores}.

To be more specific, at the selection stage, True Positives refer to the active variables (with non-zero coefficients in the model) that are correctly selected by the randomized SQR method; 
False Positives refer to inactive (with zero-coefficients in the model) variables that are incorrectly selected by the same method. 
Similarly, False Negatives refer to the active variables that are missed by the selection method. 
After the inference stage, True Positives refer to variables where the post-selection confidence intervals of active variables do not include zero,   
False Positives refer to cases when the intervals of inactive variables do not include zero, and 
False Negatives refer to cases when the intervals of active variables include zero.

\begin{table}[t]
\caption{\small{Accuracy based on F1 scores before and after applying ``Proposed'' in the three models and signal settings. The improvement in accuracy of estimating the true signal structure from the inference stage is shown in the third row of the Table.}}
\vspace{10pt}
\footnotesize
\begin{tabular*}{\textwidth}{@{\extracolsep{\fill}}*{10}{c}}
\hline
\hline
    & & Model 1 & & & Model 2 & & & Model 3 & \\
     \cline{2 - 4} \cline{5 - 7} \cline{8 - 10} 
Signal Strength & Low & Medium & High & Low & Medium & High & Low & Medium & High \\ 
\hline
\hline
F1 score before inference
    & 0.10 & 0.19 & 0.22 & 0.12 & 0.19 & 0.23 & 0.22 & 0.21 & 0.21 \\
F1 score after inference
    & 0.26 & 0.63 & 0.70 & 0.29 & 0.58 & 0.66 & 0.73 & 0.70 & 0.68 \\
\hline
\hline
\end{tabular*}
\label{table:F1scores}
\end{table}

Consistent with our expectation, the ``Proposed'' method improves the accuracy of identifying the true signal structure from the data by conducting inference post selection. 
As evident from this table, we note a significant improvement in accuracy after conducting inference, which can be as high as $50\%$.
It is crucial to bear in mind that during the estimation process, the selection stage can sometimes mistakenly identify noise variables and include them in the model. 
However, after the selection stage, the ``Proposed'' approach enables the analyst to remove these noise variables through inference.

\section{Analysis of a birth weight dataset}
\label{sec:real}

\subsection{Risk factors for low birth weight}
In this section, we use the proposed method to investigate the association between low birth weight in twins and various risk factors. 
The dataset we use is derived from the $2022$ U.S. birth records collected by the Centers for Disease Control and Prevention (CDC). 
It consists of $114,763$ observations across $79$ potential risk factors that include maternal age, order of live-birth, race, marital status, tobacco use, prenatal care, method of delivery, and gestational age, among others.
We exclude from this data observations with missing data. 
Additionally, we removed three variables to avoid multicollinearity in our design matrix and a few more categorical variables that did not have sufficient observations for each of their categories.
After applying one hot encoding to our categorical variables, we obtain a total of $33,798$ observations of birth weight in twins across a total of $83$ factors. 
The entire list of variables in our data and their description are included in Appendix \ref{appdx:data} of the Supplementary Material.

Note that our method ``Proposed'' is implemented  as described in Section \ref{sec:simulation}. 
We use the tuning parameter of $\lambda=0.4 \sqrt{\log p / n}$ and bandwidth $h = \max \big\{0.05, \sqrt{\tau(1-\tau)} (\log (p) / n)^{1 / 4}\big\} $ in this example, and draw Gaussian white noise $\omega_n$ from $\cN\left(0_p, \frac{1}{2n}I_{p, p}\right)$. 
For data splitting, we use two-thirds of the samples for model selection and the remaining one-third of the samples for constructing confidence intervals for the effects of the selected variables on the $10\%$ conditional quantile. 
We focus on $\tau=0.1$ in our analysis, because studying risk factors for low birth weight is typically more relevant in public health.
See, for example, investigations by \cite{low_birthweight1, low_birthweight2}.

First, we apply our method to model the $10\%$ quantile in the birth weight data on the full dataset by selecting the pertinent risk factors and then constructing interval estimators for their effects. 
Given the large sample size, we take the results obtained from the entire dataset as a benchmark for assessing the power of statistical analysis based on a smaller dataset of size $500$. 
We refer to this analysis performed on the full dataset as ``Baseline'', which will aid us to evaluate if the conclusions drawn from our method on the much smaller dataset align with those drawn from the larger dataset, and determine how they compare with data splitting on the subsets of the data.

In Figure \ref{fig:real1}, we show the $90\%$ confidence intervals for the variables selected by ``Baseline''. 
For this same set of variables, we show the results from ``Proposed'' and ``Splitting'' when implemented on a randomly drawn subsample of size $n=500$. In this case, we use $\lambda =0.0377$ and $h=0.092$.
Note that we do not show an interval for the variables that were not selected by the method in the first place.
\begin{figure}[t]
	\centering
	\includegraphics[scale = 0.68]{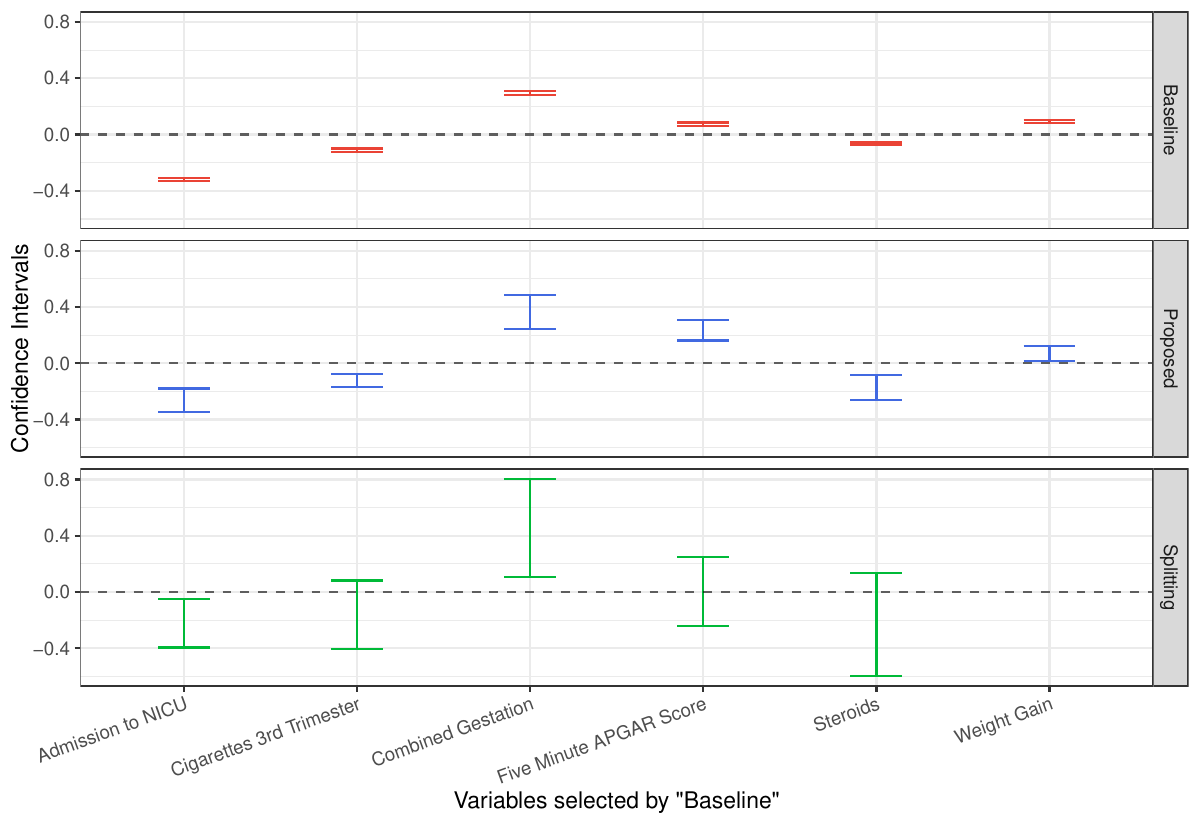}
	\caption{
   \small{$90\%$ confidence intervals for the variables chosen as significant by the full data analysis. 
     ``Splitting'' method does not select the ``Weight Gain'' factor and fails to identify the significance of ``Cigarettes 3rd Trimester'', ``Five Minute APGAR Score'' and ``Steroids'' based on subsamples.
     In contrast, ``Proposed'' identifies the association between these factors and the low birth weight in twins even when a ``Baseline'' on the full data. 
     The average length of the confidence intervals produced by ``Proposed'' is $0.125$, while ``Splitting'' results in an average interval length of $0.471$}.}
	\label{fig:real1}
\end{figure}

We observe that ``Combined Gestation'' is positively associated with the low birth weight of twins and  ``Admission to NICU'' are negatively associated with the low birth weight of twins, as identified by both ``Proposed'' and ``Splitting''.
These findings align with those from ``Baseline''.
However, ``Splitting'' failed to select the ``Weight Gain'' factor during the model selection step and did not identify the associations of low birth weight with ``Cigarettes 3rd Trimester'', ``Five Minute APGAR Score'', and ``Steroids'', possibly due to lack of sufficient data. 
On the other hand, ``Proposed'', which uses all samples for both selection and inference, not only identifies these factors but also yields shorter intervals on average than ``Splitting''. 
Overall, the proposed intervals are substantially shorter than the intervals produced by ``Splitting'' on the same  subsamples.

\subsection{Replicating our analysis}

In this section, we repeat the above-described analysis on $100$ random subsamples of  size of $500$. We report in Table \ref{table:real1}, the fraction of times that a selected variable was reported as significant with the post-selection interval estimators.
Note that we considered the association significant if the corresponding interval did not include $0$. 
Additionally, in Figure \ref{fig:real2}, we display the interval lengths for each selected variable as well as the average interval lengths for all selected variables.

\begin{table}[b]
\caption{\small{The fraction of times the post-selection interval for a variable did not include zero out of the total number of times the same variable was selected.
This fraction is in general higher for ``Proposed'' compared to ``Splitting'', and is especially true for the variables ``Five Minute APGAR Score'' and ``Weight Gain'', indicating the higher efficacy of ``Proposed'' in detecting significance.
}}
\vspace{10pt}
\small
\centering
\begin{tabular}{ccccccc}
\hline 
\hline
Variables & \makecell{Admission to\\ NICU} & \makecell{Cigarettes 3rd\\ Trimester} & \makecell{Combined\\ Gestation} & \makecell{Five Minute\\ APGAR Score} & \makecell{Steroids} & \makecell{Weight\\ Gain} \\ 
\hline 
\hline
Proposed  & 0.979 & 0.688 & 0.990 & 0.366 & 0.431 & 0.609   \\ 
Splitting & 0.880 & 0.674 & 0.830 & 0.277 & 0.327 & 0.467  \\
\hline 
\hline
\end{tabular}
\label{table:real1}
\end{table}

\begin{figure}[t]
	\centering
	\includegraphics[scale = 0.68]{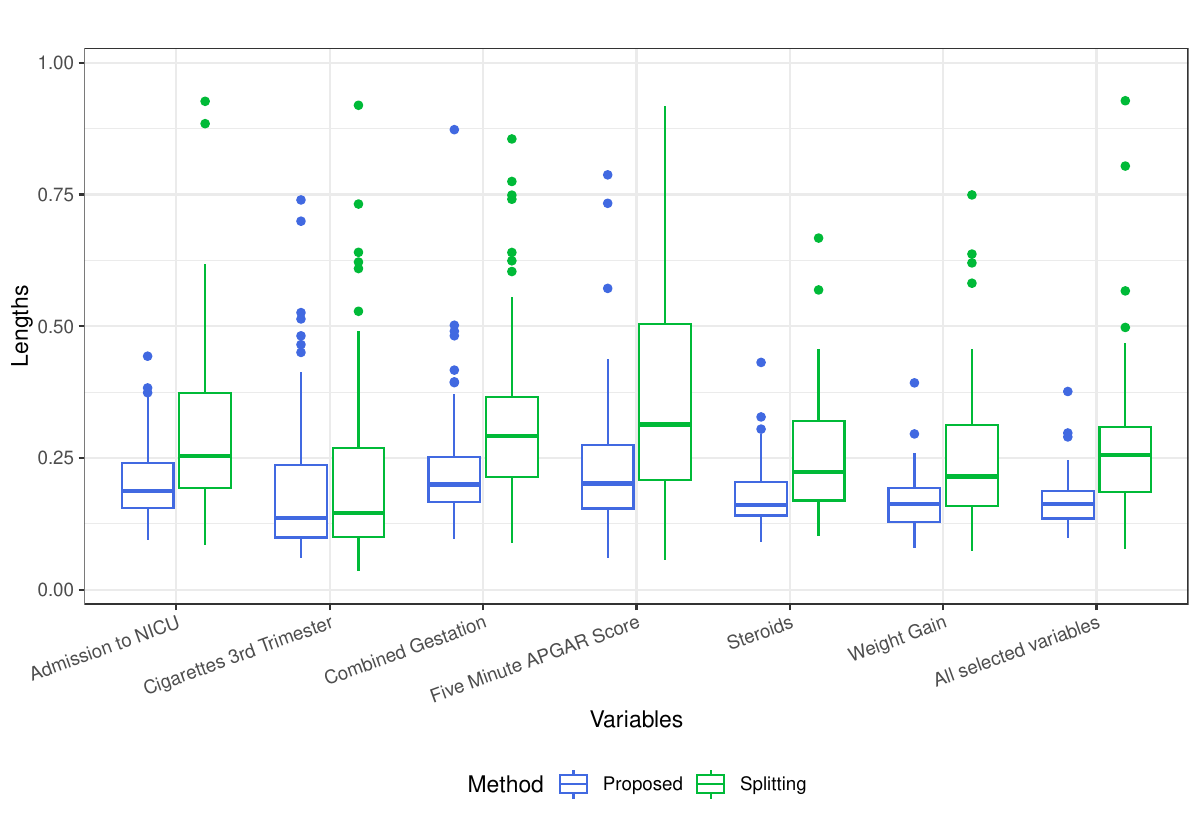}
	\caption{
    \small{
    Box plots for the lengths of $90\%$ confidence intervals of each selected variable and the average lengths of all selected variables. 
    ``Proposed'' results in shorter intervals compared to ``Splitting'' for each variable and overall on average for all variables.}}
	\label{fig:real2}
\end{figure}

These results support our findings on simulated data, indicating that ``Proposed'' is far more effective at detecting significance than ``Splitting''. 
The box plots for the lengths of these intervals confirm that ``Proposed'' consistently produces shorter intervals than ``Splitting'', for each selected variable and overall for all of them.

In Appendix \ref{appdx:data} of the Supplementary Material , we apply our approach on a subsample of size $500$ to analyze the risk factors associated with the $1\%$ conditional quantile. 
This analysis was repeated across a $100$ random subsamples.
The interval lengths for the selected variables are displayed in Figure \ref{fig:birth_length_1_old}. 
The findings of our experiment remain consistent at a lower quantile level, as the ``Proposed'' approach produced shorter intervals than ``Splitting''.

\section{Discussion}
\label{sec:discussion}

In this paper, we have addressed the challenge of conducting selective inference on quantile effects by introducing an asymptotic pivot. 
Our method, which incorporates estimators from smoothed quantile regression and an external randomization variable, ensures accurate inference and is computationally straightforward. 

Our experiments with simulated and real data show that our approach outperforms existing techniques, especially when dealing with small sample sizes or weak signals. 
For example, modifying the existing polyhedral toolbox for quantile regression results in infinitely long intervals and fails to achieve the desired coverage in low signal-to-noise ratio scenarios. 
Similarly, the common practice of splitting samples into two independent subsets can sometimes lead to undercoverage due to insufficient sample sizes for inference. 
In contrast, our method allows us to use the full dataset for both selection and inference in these instances, leading to improved statistical accuracy, efficiency, and numerical stability.

For future work, we hope that  asymptotic pivots of similar nature can be extended to encompass a wider range of nonparametric or semi-parametric models. 
Exploring the potential of our asymptotic pivot for honest inference in nonlinear models is a promising direction for further investigation.
Additionally, we believe that extensions of our approach to other penalties such as Smoothly Clipped Absolute Deviation (SCAD) and Minimax Concave Penalty (MCP) will yield a versatile toolbox for making selective inference. 

\section*{Acknowledgements}
The Python package \texttt{conquer} from \url{https://github.com/WenxinZhou/conquer} was used by the authors to build their code for selective inference.

Y. Wang is supported by NSF DMS grant 1951980.
S. Panigrahi's research is supported in part by NSF DMS grants 1951980 and 2113342, and NSF Career Award DMS 2337882.
X. He's research is supported in part by NSF DMS grants 1951980 and 2345035.

\bibliographystyle{ims}
\bibliography{ref.bib}

\appendix
\titleformat{\section}[block]{\normalfont\Large\bfseries}{Appendix \Alph{section}:}{1em}{}
\renewcommand{\thesection}{\Alph{section}}

\begin{spacing}{1.3}

\section{Asymptotic representation of SQR estimators in Section \ref{subsec:sqr_estimators}}
\label{suppl:general}
Denote the bandwidth parameters $h$ and $h'$ as the bandwidths used for selection and inference, respectively.
Recall the definition of the pseudo-parameter $\trueprime$ in \eqref{conv:smooth:target} and let $\truejprime$ denote the $j^{\text{th}}$ component of $\trueprime$.
Firstly, define the matrices
\[
\begin{aligned}
	&\quad\quad\quad J = \EE_{\FF_n} \left[X ^ \top \nabla ^ 2 \widehat Q_{n;h'} (X_E \trueprime; Y) X\right] 
    = \begin{bmatrix} J_{E, E} & J_{E, E'} \\ J_{E', E} & J_{E', E'} \end{bmatrix} 
    = \begin{bmatrix} J_{\cdot, E} \ & J_{\cdot, E'} \end{bmatrix} \in \RR ^ {p\times p} \\
    &\quad\quad\quad \widetilde J = \EE_{\FF_n} \left[X ^ \top \nabla ^ 2 \widehat Q_{n;h} (X_E \trueprime; Y) X\right] 
    = \begin{bmatrix} \widetilde J_{E, E} & \widetilde J_{E, E'} \\ \widetilde J_{E', E} & \widetilde J_{E', E'} \end{bmatrix}
    = \begin{bmatrix} \widetilde J_{\cdot, E} \ & \widetilde J_{\cdot, E'} \end{bmatrix} \in \RR ^ {p\times p} \\
	H &= \Cov_{\FF_n} \Big(\sqrt{n} X ^ \top \nabla \widehat Q_{n;h'} (X_E \trueprime; Y), \sqrt{n} X ^ \top \nabla \widehat Q_{n;h'} (X_E \trueprime; Y)\Big)
    = \begin{bmatrix} H_{E, E} & H_{E, E'} \\ H_{E', E} & H_{E', E'}\end{bmatrix} \in \RR ^ {p\times p} \\
    \widetilde H &= \Cov_{\FF_n} \Big(\sqrt{n} X ^ \top \nabla \widehat Q_{n;h} (X_E \trueprime; Y), \sqrt{n} X ^ \top \nabla \widehat Q_{n;h} (X_E \trueprime; Y)\Big) 
    = \begin{bmatrix} \widetilde H_{E, E} & \widetilde H_{E, E'} \\ \widetilde H_{E', E} & \widetilde H_{E', E'}\end{bmatrix} \in \RR ^ {p\times p} \\
    &\quad\quad\quad\quad K = \Cov_{\FF_n} \Big(\sqrt{n} X ^ \top \nabla \widehat Q_{n;h} (X_E \trueprime; Y), \sqrt{n} X ^ \top \nabla \widehat Q_{n;h'} (X_E \trueprime; Y)\Big) \\
    &\quad\quad\quad\quad\quad = \begin{bmatrix} K_{E, E} & K_{E, E'} \\ K_{E', E} & K_{E', E'}\end{bmatrix} 
    = \begin{bmatrix} K_{\cdot, E} \ & K_{\cdot, E'} \end{bmatrix} \in \RR ^ {p\times p}
\end{aligned}
\]
and let $\varmle = J_{E, E} ^ {-1} H_{E, E} J_{E, E} ^ {-1}$ and let $\varmlej$ be the $j^{\text{th}}$ diagonal entry of $\Sigma_{E, E}$.
Then, define the refitted SQR estimators
$$
    \mleprime = \argmin_{b \in \RR ^ {q}} \sqrt{n} \widehat Q_{n;h'} (X_E b; Y)
$$
and let $\mlejprime = \mathbb{S}_j \mleprime$ denote the $j^{\text{th}}$ component of this estimator.
Define the statistic
$$
    \mlenjprime =\begin{pmatrix} \mathbb{S}_{[E]\setminus j}\left(\mleprime - \frac{1}{\varmlej}\covmlej\mlejprime\right) \\ X ^ \top \nabla \widehat Q_{n;h}(X_E \mleprime; Y) + (K_{\cdot, E} H_{E, E} ^ {-1} J_{E, E} - \widetilde J_{\cdot, E}) \mleprime \end{pmatrix}
    \in \RR ^ {p + q - 1},
$$
and the nuisance parameters 
$$
    \truenjprime =\begin{pmatrix} \mathbb{S}_{[E]\setminus j}\left(\trueprime - \frac{1}{\varmlej}\covmlej\truejprime\right) \\ \EE \left[X ^ \top \nabla \widehat Q_{n;h}(X_E \trueprime; Y)\right] + (K_{\cdot, E} H_{E, E} ^ {-1} J_{E, E} -\widetilde J_{\cdot, E}) \trueprime \end{pmatrix}
    \in \RR ^ {p + q - 1}.
$$
Finally, define the matrices
\begin{equation}
\label{eq:matrix_prime}
\begin{gathered}
    M ^ {\prime j} = - \frac{1}{\varmlej} K_{\cdot, E} J_{E, E} ^ {-1} \mathbb{S}_j ^ \top \in \RR ^ {p \times 1}, \\ 
    N ^ {\prime j} = \begin{bmatrix} - K_{\cdot, E} H_{E, E} ^ {-1} J_{E, E} \mathbb{S}_{[E]\setminus j} ^ \top & I_{p, p}\end{bmatrix} \in \RR ^ {p \times(p+q-1)}, \ T ^ {\prime} = \widetilde{J}_{\cdot, E} \text{\normalfont Diag}(\sgnr) \in \RR ^ {p \times q}.
\end{gathered}	
\end{equation}

In this section, we present the general versions of the results from Section \ref{subsec:sqr_estimators} where $h$ and $h'$ take different values.
The results in Section \ref{subsec:sqr_estimators} correspond to the special case where $h = h'$, where $\widetilde J = J$, $H = \widetilde H = K$, $\mleprime = \mle$, and $\mlenjprime$ equivalent to $\mlenj$, ignoring the zero entries.


\subsection{General statement for  Proposition \ref{prop:beta_gamma_expression}}
\begin{prop}
\label{prop:beta_gamma_expression_general}
Let $E$ be a fixed subset of $[p]$.
Define 
\begin{align*}
\begin{gathered}
\Lambda_1' = \begin{bmatrix} - \mathbb{S}_j J_{E,E} ^ {-1}~ & 0 ^ \top_{p} \end{bmatrix} \begin{bmatrix}
    	H_{E, E} & K_{E, \cdot} \\
    	K_{\cdot, E} & \widetilde H
    \end{bmatrix} ^ {1/2} \in \RR^{1\times (p+q)},\\
\Lambda_2' =\begin{bmatrix} \Lambda_{2,1}' \\ \Lambda_{2,2}'\end{bmatrix}= \begin{bmatrix}  
		 \mathbb{S}_{[E] \setminus j}\left(\frac{1}{\varmlej} \covmlej \mathbb{S}_j J_{E, E} ^ {-1}  -  J_{E, E} ^ {-1}\right) & 0_{q-1, p} \\
		- K_{\cdot, E} H_{E, E} ^ {-1} & I_{p, p} \end{bmatrix} 
		\begin{bmatrix}
    	H_{E, E} & K_{E, \cdot} \\
    	K_{\cdot, E} & \widetilde H
    \end{bmatrix} ^ {1/2} \in \RR^{(p+q-1)\times (p+q)}.
\end{gathered}
\end{align*}
For $j \in E$, it holds that
\begin{equation*}
	\sqrt{n} \begin{pmatrix}
		\mlejprime -  \truejprime \\ 
		\mlenjprime - \truenjprime
	\end{pmatrix}
	= \begin{pmatrix} \Lambda_1' \\ \Lambda_2' \end{pmatrix} \Up_n' + \bar \Delta_1',
\end{equation*}
where 
\begin{align*}
\begin{gathered}
    \Up_n' = \sqrt{n} \begin{bmatrix}
    	H_{E, E} & K_{E, \cdot} \\
    	K_{\cdot, E} & \widetilde H
    \end{bmatrix} ^ {- 1/2} \left( \begin{pmatrix}  X_E^\top \nabla \widehat Q_{n;h'} (X_E \trueprime; Y)\\ X^\top \nabla \widehat Q_{n;h} (X_E \trueprime; Y) \end{pmatrix}- \begin{pmatrix} 0_q \\ \EE_{\FF_n} [X ^ \top \nabla \widehat Q_{n;h} (X_E \trueprime; Y)] \end{pmatrix}\right), \\ 
\bar \Delta_1' = o_p(1).
\end{gathered}
\end{align*}
\end{prop}

\begin{proof}
Define
\begin{align*}
\begin{gathered}
    \Rnj = \mathbb{S}_{[E]\setminus j}\left(\mleprime - \frac{1}{\varmlej}\covmlej\mlejprime\right) \in \RR^{q-1}, \\
    \Tn =  X ^ \top \nabla \widehat Q_{n;h}(X_E \mleprime; Y) + (K_{\cdot, E} H_{E, E} ^ {-1} J_{E, E} - \widetilde J_{\cdot, E}) \mleprime \in \RR^{p}.
\end{gathered}	
\end{align*}
Applying the Taylor expansion to 
\[  
    \sqrt{n} X_{E} ^ \top \nabla \widehat Q_{n;h'} (X_E \mleprime; Y)
\]
at $\trueprime$, we have
\[
\begin{aligned}
	\sqrt{n} &X_E ^ \top \nabla \widehat Q_{n;h'} (X_E \mleprime; Y) \\
	&= \sqrt{n} X_E ^ \top \nabla \widehat Q_{n;h'} (X_E \trueprime; Y)
	+ X_E ^ \top \nabla ^ 2 \widehat Q_{n;h'} (X_E \trueprime; Y) X_E \sqrt{n} (\mleprime - \trueprime) + o_p(1) \\
	&= \sqrt{n} X_E ^ \top \nabla \widehat Q_{n;h'} (X_E \trueprime; Y)
    + J_{E, E} \sqrt{n}  (\mleprime - \trueprime) + o_p(1).
\end{aligned}
\]

Observe that
$$X_E ^ \top \nabla \widehat Q_{n;h'} (X_E \mleprime; Y) = 0.$$
Therefore, the Taylor expansion lead us to note that
\begin{equation}
\label{eq:taylor_beta}
	\sqrt{n} \mleprime 
	= \sqrt{n} \trueprime - J_{E, E} ^ {-1} \sqrt{n}  X_E ^ \top \nabla \widehat Q_{n,h'}(X_E \trueprime; Y) + o_p(1).
\end{equation}
Based on the previous display, we can write
\begin{equation}
\label{eq:taylor:mlej}
\begin{aligned}
    \sqrt{n} \mlejprime &= \mathbb{S}_j \sqrt{n}\mleprime \\
    &= \sqrt{n} \truejprime - \mathbb{S}_j J_{E, E} ^ {-1} \sqrt{n} X_E ^ \top \nabla \widehat Q_{n,h'}(X_E \trueprime; Y) + o_p(1)\\
    &= \sqrt{n} \truejprime + \Lambda_1' \sqrt{n} \Up_n' + o_p(1),
\end{aligned}
\end{equation}
and 
\begin{equation}
\label{eq:taylor:mlenj1}
\begin{aligned}
    \sqrt{n} \Rnj 
    &= \mathbb{S}_{[E]\setminus j}  \sqrt{n}\left(\trueprime -\frac{1}{ \varmlej }  \covmlej  \truejprime\right) \\
    &\;\; + \mathbb{S}_{[E]\setminus j} \left(\frac{1}{\varmlej} \covmlej \mathbb{S}_j J_{E, E} ^ {-1} - J_{E, E} ^ {-1}\right) \sqrt{n}X_E ^ \top \nabla \widehat Q_{n;h'} (X_E \trueprime; Y) + o_p(1)\\
    &=\sqrt{n} \mathbb{S}_{[E]\setminus j}  \left(\trueprime - \frac{1}{\varmlej}\covmlej \truejprime\right) + \Lambda_{2,1}' \sqrt{n} \Up_n' + o_p(1).
\end{aligned}
\end{equation}

Similarly, the Taylor expansion lead us to:
\[
\begin{aligned}
	\sqrt{n} X ^ \top \nabla \widehat Q_{n;h} (X_E \mleprime; Y) &= \sqrt{n}  X ^ \top \nabla \widehat Q_{n;h} (X_E \trueprime; Y) + \widetilde{J}_{\cdot, E} \sqrt{n} (\mleprime - \trueprime) + o_p(1).
\end{aligned}
\]
Adding $(K_{\cdot, E} H_{E, E} ^ {-1} J_{E, E} - \widetilde{J}_{\cdot, E}) \sqrt{n} \mleprime$ to both sides of the equation and using \eqref{eq:taylor_beta}, we obtain
\begin{equation}
\label{eq:express_gamma2}
\begin{aligned}
	\sqrt{n}\Tn 
	&= \sqrt{n} X ^ \top \nabla \widehat Q_{n;h} (X_E \trueprime; Y) + K_{\cdot, E} H_{E, E} ^ {-1} J_{E, E}\sqrt{n} \mleprime - \widetilde{J}_{\cdot, E} \sqrt{n} \trueprime + o_p(1) \\
	&= \sqrt{n} \EE_{\FF_n} \left[X ^ \top \nabla \widehat Q_{n;h} (X_E \trueprime; Y) \right] + \big(K_{\cdot, E} H_{E, E} ^ {-1} J_{E, E} - \widetilde{J}_{\cdot, E}\big) \sqrt{n} \trueprime \\
	&\;\; - K_{\cdot, E} H_{E, E} ^ {-1} \sqrt{n} X_E ^ \top \nabla \widehat Q_{n;h'} (X_E \trueprime; Y) + \sqrt{n} \big(X ^ \top \nabla \widehat Q_{n;h} (X_E \trueprime; Y) \\
	&\;\; - \EE_{\FF_n} \left[X ^ \top \nabla \widehat Q_{n;h} (X_E \trueprime; Y) \right]\big) + o_p(1)\\
	&= \sqrt{n} \EE_{\FF_n} \left[X ^ \top \nabla \widehat Q_{n;h} (X_E \trueprime; Y) \right] + \big(K_{\cdot, E} H_{E, E} ^ {-1} J_{E, E} - \widetilde{J}_{\cdot, E}\big) \sqrt{n} \trueprime + \Lambda_{2,2}' \sqrt{n} \Up_n' + o_p(1).
\end{aligned}	
\end{equation}
By combining the equations in \eqref{eq:taylor:mlej},  \eqref{eq:taylor:mlenj1} and \eqref{eq:express_gamma2}, we obtain the asymptotic linear representation for our estimators in the Proposition.
\end{proof}

\subsection{General statement for Corollary \ref{cor:beta_gamma_expression}}

\begin{cor}
\label{cor:beta_gamma_expression_general}
Suppose that the conditions in Assumptions \ref{aspt:moment_bound} and \ref{aspt:Lip} are met. 
Then, we have
\begin{equation*}
	\sqrt{n} \begin{pmatrix}
		\mlejprime -  \truestarj \\ 
		\mlenjprime - \truenjprime
	\end{pmatrix}
	= \begin{pmatrix} \Lambda_1' \\ \Lambda_2' \end{pmatrix} \Up_n' + \Delta_1',
\end{equation*}
where $\Delta_1' = o_p(1)$.
\end{cor}

\begin{proof}
The result follows directly from the application of Proposition \ref{prop:diff_target} and Proposition \ref{prop:beta_gamma_expression_general}.
\end{proof}

\subsection{General statement for Corollary \ref{cor:asy_distribution}}

\begin{cor}
\label{cor:asy_distribution_general}
For a fixed set $E$, we have
$$
    \sqrt{n} \begin{pmatrix}
		\mlejprime -  \truestarj \\ 
		\mlenjprime - \truenjprime
	\end{pmatrix} 
    \Rightarrow
    \cN \left(0_{p+q},
    \begin{bmatrix}
        \varmlej & 0_{1, p+q-1} \\
        0_{p+q-1, 1} & \Lambda_2' \Lambda_2 ^ {\prime\top} 
    \end{bmatrix}
    \right).
$$ 
\end{cor}

\begin{proof}
The asymptotic distribution is obtained immediately by noting that
\[
\begin{aligned} 
	\begin{bmatrix} \Lambda_1' \\ \Lambda_2' \end{bmatrix} \begin{bmatrix} \Lambda_1 ^ {\prime\top} & \Lambda_2 ^ {\prime \top} \end{bmatrix}
	&= \begin{bmatrix}
        \varmlej & 0_{1, p+q-1} \\
        0_{p+q-1, 1} & \Lambda_2' \Lambda_2 ^ {\prime \top} 
    \end{bmatrix}.
\end{aligned}
\]
\end{proof}

\subsection{General statement for Proposition \ref{prop:taylor_expansion}}

\begin{prop}
\label{prop:taylor_expansion_general}
Recall the definition of $T'$, $M ^ {\prime j}$ and $N ^ {\prime j}$ in \eqref{eq:matrix_prime}.
We have that
\begin{equation*}
	T' \sqrt{n} |\lasso| + \D  = \sqrt{n} \rand - M ^ {\prime j} \sqrt{n} \mlejprime - N ^ {\prime j} \sqrt{n} \mlenjprime + \Delta_2',
\end{equation*}
where $\Delta_2' = o_p(1)$.
\end{prop}

\begin{proof}
Define
\begin{align*}
\begin{gathered}
    \Rnj = \mathbb{S}_{[E]\setminus j}\left(\mleprime - \frac{1}{\varmlej}\covmlej\mlejprime\right) \in \RR^{q-1}, \\
    \Tn = X ^ \top \nabla \widehat Q_{n;h}(X_E \mleprime; Y) + (K_{\cdot, E} H_{E, E} ^ {-1} J_{E, E} - \widetilde{J}_{\cdot, E}) \mleprime \in \RR^{p}.
\end{gathered}	
\end{align*}
Applying a Taylor expansion of the gradient of the SQR loss around $\mle$, we write
\begin{equation}
\begin{aligned}
    \sqrt{n} X ^ \top \nabla \widehat Q_{n;h} (X \lasso; Y)
    &= \sqrt{n} X ^ \top \nabla \widehat Q_{n;h} (X_E \mleprime; Y) 
	+  X ^ \top \nabla ^ 2 \widehat Q_{n;h} (X_E \mleprime; Y) X_E \sqrt{n} (\lasso - \mleprime) \\
    &\;\;\;\;\;\;+ o_p (1) \\
    &= \sqrt{n} X ^ \top \nabla \widehat Q_{n;h} (X_E \mleprime; Y) 
	+ \widetilde{J}_{\cdot, E} \sqrt{n} (\lasso - \mleprime) +  o_p (1) \\
    &= \sqrt{n} X ^ \top \nabla \widehat Q_{n;h}(X_E \mleprime; Y) + (K_{\cdot, E} H_{E, E} ^ {-1} J_{E, E} - \widetilde{J}_{\cdot, E}) \sqrt{n} \mleprime \\
    &\;\; + \widetilde{J}_{\cdot, E} \sqrt{n} \lasso - K_{\cdot, E} H_{E, E} ^ {-1} J_{E, E} \sqrt{n} \mleprime
    + o_p (1).
\end{aligned}
\label{eq:KKTtaylor1}
\end{equation}
Note that the second term on the right-hand side can be expressed as
\begin{equation*}
\label{eq:KKT_taylor2}
\begin{aligned}
    K_{\cdot, E} H_{E, E} ^ {-1} J_{E, E} \sqrt{n} \mleprime
    &= K_{\cdot, E} H_{E, E} ^ {-1} J_{E, E} \bigg(\sqrt{n} \mleprime - \frac{1}{\varmlej} \covmlej \sqrt{n} \mlejprime + \frac{1}{\varmlej} \covmlej \sqrt{n} \mlejprime\bigg) \\
    &= K_{\cdot, E} H_{E, E} ^ {-1} J_{E, E} \mathbb{S}_{[E]\setminus j} ^ \top \mathbb{S}_{[E]\setminus j} \bigg(\sqrt{n} \mleprime - \frac{1}{\varmlej} \covmlej \sqrt{n} \mlejprime\bigg) \\
    &\;\;+ K_{\cdot, E} H_{E, E} ^ {-1} J_{E, E} \frac{1}{\varmlej} \covmlej \sqrt{n} \mlejprime \\
    &= K_{\cdot, E} H_{E, E} ^ {-1} J_{E, E}  \mathbb{S}_{[E]\setminus j} ^ \top \sqrt{n} \Rnj
    + \frac{1}{\varmlej} K_{\cdot, E} J_{E, E} ^ {-1}\mathbb{S}_j ^ \top \sqrt{n} \mlejprime.
\end{aligned}
\end{equation*}
Plugging the above-stated representation into \eqref{eq:KKTtaylor1}, we obtain
\[
\begin{aligned}
    &\sqrt{n} X ^ \top \nabla \widehat Q_{n;h} (X \lasso; Y)\\
    &= - \frac{1}{\varmlej} K_{\cdot, E} J_{E, E} ^ {-1} \mathbb{S}_j ^ \top \sqrt{n} \mlejprime
    + \begin{bmatrix} - K_{\cdot, E} H_{E, E} ^ {-1} J_{E, E} \mathbb{S}_{[E]\setminus j} ^ \top & I_{p, p}\end{bmatrix}
   \sqrt{n} \begin{pmatrix} \Rnj \\ \Tn \end{pmatrix}
    + \widetilde{J}_{\cdot, E} \sqrt{n} \lasso 
    + o_p (1) \\
    &= M ^ {\prime j} \sqrt{n} \mlejprime + N ^ {\prime j} \sqrt{n} \mlenjprime + T'\sqrt{n} |\lasso| + o_p (1).
\end{aligned}
\]
Based on the Karush-Kuhn-Tucker (KKT) conditions of stationary that
\[
	\sqrt{n} X ^ \top \nabla \widehat Q_{n;h} (X_E \lasso; Y) + \D - \sqrt{n}\rand = 0_p,
\]
we conclude that
\[
    T' \sqrt{n} |\lasso| + \D  = \sqrt{n} \rand - M ^ {\prime j} \sqrt{n} \mlejprime - N ^ {\prime j} \sqrt{n} \mlenjprime + \Delta_2'.
\]
\end{proof}

\section{Proofs for results in Section \ref{sec:pivot}}
\label{appdx:A}

\subsection{Supporting results}
\label{subsec:auxiliary_results_pivot}

\begin{lem}
\label{lem:joint_density}
For a fixed set $E$ and a fixed set of signs $\sgnr$, the joint density of the variables $\sqrt{n} \mlej$, $\sqrt{n} \mlenj$, $\sqrt{n} \lasso$ and $\Z$ is equal to:
\[
\begin{aligned}
	&\phi \left(\sqrt{n} \mlej; \sqrt{n} \truestarj, \varmlej\right)
	\phi \left(\sqrt{n} \mlenj; \sqrt{n} \truenj, \Lambda_2 \Lambda_2 ^ \top\right) \\
	&\quad\quad \times  \phi \left(T \sqrt{n} |\lasso| + \D + M ^ j \sqrt{n} \mlej + N ^ j \sqrt{n} \mlenj; 0_p, \varrand\right)
	\times |\cJ|,
\end{aligned}
\]
where 
\[
    \cJ = 
    \begin{pmatrix}
         \mathbb{S}_E T & 0_{q, q'} \\
         \mathbb{S}_{E'} T & \lambda I_{q', q'}
    \end{pmatrix}.
\]
\end{lem}

\begin{proof}
Because $\Delta_2=0$, note that
\[
	\Pi \begin{pmatrix} \sqrt{n} |\lasso| \\ \Z \end{pmatrix}
	= \sqrt{n} \rand - M ^ j \sqrt{n} \mlej - N ^ j \sqrt{n} \mlenj -  \lambda \begin{pmatrix} \sgnr \\ 0_{q'} \end{pmatrix},
\]
where 
\[
	\Pi = \begin{bmatrix}
		\mathbb{S}_E T & 0_{q, q'} \\
		\mathbb{S}_{E'} T & \lambda I_{q', q'}
	\end{bmatrix} \in \RR ^ {p \times p}.
\]
For $(x_1, x_2, x_3, x_4) \in \RR \times \RR ^ {p-1} \times \RR ^ q \times \RR ^ {q'}$, define
\begin{equation}
\label{eq:defn_g}
	g(x_1, x_2, x_3, x_4)
	= \begin{pmatrix}
		x_1 \\ 
		x_2 \\ 
		\mathbb{S}_E T x_3 + \mathbb{S}_E M ^ j x_1 + \mathbb{S}_E N ^ j x_2 + \lambda \sgnr \\
		\mathbb{S}_{E'} T x_3 + \lambda x_4 + \mathbb{S}_{E'} x_1 + \mathbb{S}_{E'} x_2 
	\end{pmatrix},
\end{equation}
and observe that
$$
g\left(\begin{pmatrix}
		\sqrt{n} \mlej \\
		\sqrt{n} \mlenj \\
		\sqrt{n} |\lasso| \\
		\Z
	\end{pmatrix}\right)
	=
	\begin{pmatrix}
		\sqrt{n} \mlej \\
		\sqrt{n} \mlenj \\
		\sqrt{n} \rand 
	\end{pmatrix}.
$$ 
for a fixed set $E$ and a fixed set of signs $\sgnr$.
Suppose that $p_n(\cdot)$ is the joint density of the variables 
\[
	\left(\sqrt{n} \mlej, \sqrt{n} \mlenj, \sqrt{n} \omega_n\right) \in \RR \times \RR ^ {p-1} \times \RR ^ p.
\]
By applying the change of variables
\[ \begin{pmatrix}
		\sqrt{n} \mlej \\
		\sqrt{n} \mlenj \\
		\sqrt{n} \rand 
	\end{pmatrix}	
	\stackrel{\left(g\right)^{-1}}{\longrightarrow}
	\begin{pmatrix}
		\sqrt{n} \mlej \\
		\sqrt{n} \mlenj \\
		\sqrt{n} |\lasso| \\
		\Z
	\end{pmatrix},
\]
we obtain the density for the variables on the right-hand side. 
It follows that this density is equal to
\begin{equation}
\label{eq:jocobian}
	|\cJ_{g}| \, p_n \left(g \left(\sqrt{n} \mlej, \sqrt{n} \mlenj, \sqrt{n} |\lasso|, \Z\right)\right),
\end{equation}
where $\cJ_{g}$ is the Jacobian matrix of the map $g$, which is equal to
\[
\begin{pmatrix}
    1 & 0_{1, p-1} & 0_{1, q} & 0_{1, q'} \\
    0_{p-1, 1} & I_{p-1,p-1} & 0_{p-1,q} & 0_{p-1,q'}\\
    \mathbb{S}_E M ^ j & \mathbb{S}_E N ^ j & \mathbb{S}_E T & 0_{q, q'} \\
    \mathbb{S}_{E'} & \mathbb{S}_{E'} & \mathbb{S}_{E'} T & \lambda I_{q', q'}
\end{pmatrix}.
\]
Due to the lower triangular structure of this matrix, note that $|\cJ_{g}|=|\cJ|$ for $\cJ$ defined in the Lemma.

Since $\sqrt{n} \cZ_n \stackrel{d}{=} \cN (0_p, I_{p, p})$ and $\sqrt{n} \omega_n$ is independent of both estimators $\mlej$ and $\mlenj$, we have that
\[
\begin{aligned}
	p_n \left(\sqrt{n} \mlej, \sqrt{n} \mlenj, \sqrt{n} \rand\right)
	=\,\, &\phi \left(\sqrt{n} \mlej; \sqrt{n} \truestarj, \varmlej\right)
	   \phi \left(\sqrt{n} \mlenj; \sqrt{n} \truenj, \Lambda_2 \Lambda_2 ^ \top\right) \\
	  &\times \phi \left(\sqrt{n} \rand; 0_p, \varrand\right).
\end{aligned}
\]
Using the definition of $g$ in \eqref{eq:defn_g}, the joint density in \eqref{eq:jocobian} simplifies as
\[
\begin{aligned}
	&\phi \left(\sqrt{n} \mlej; \sqrt{n} \truestarj, \varmlej\right)
	\phi \left(\sqrt{n} \mlenj; \sqrt{n} \truenj, \Lambda_2 \Lambda_2 ^ \top\right) \\
	&\quad\quad \times \phi \left(\Pi \begin{pmatrix} \sqrt{n} \lasso \\ \Z \end{pmatrix} + \begin{pmatrix} \lambda \sgnr \\ 0_{q'} \end{pmatrix} + M ^ j \sqrt{n} \mlej + N ^ j \sqrt{n} \mlenj; 0_p, \varrand\right)
	\times|\cJ| \\
	&= \phi \left(\sqrt{n} \mlej; \sqrt{n} \truestarj, \varmlej\right)
	\phi \left(\sqrt{n} \mlenj; \sqrt{n} \truenj, \Lambda_2 \Lambda_2 ^ \top\right) \\
	&\quad\quad \times  \phi \left(T \sqrt{n} |\lasso| + \D + M ^ j \sqrt{n} \mlej + N ^ j \sqrt{n} \mlenj; 0_p, \varrand\right)
	\times |\cJ|,
\end{aligned}
\]
completing our proof. 
\end{proof}


\begin{lem}
\label{lem:pivot_equal}
The pivot in Proposition \ref{prop:pivot_exact_distr} matches the pivot provided by Theorem 1 in \cite{panigrahi2023exact} for the $\ell_1$-penalized least squares regression.
\end{lem}

\begin{proof}
Observe that 
\[
    \phi \big(x; \sqrt{n} \truestarj, \varmlej\big)
	W_0 \big(x, \sqrt{n} \mlenj\big)
\]
is proportional to 
\[
\begin{aligned}
    &\int_{I_1 ^ j} ^ {I_2 ^ j}\exp \left\{- \frac{\big(x - \sqrt{n} \truestarj\big) ^ 2}{2\varmlej} 
        - \frac{\big(Q ^ j t + M ^ j x + N ^ j \sqrt{n} \mlenj + P ^ j\big) ^ \top \!\varrand ^ {-1} \big(Q ^ j t + M ^ j x + N ^ j \sqrt{n} \mlenj + P ^ j\big)}{2}\right\} dt \\
    &\propto \int_{I_1 ^ j} ^ {I_2 ^ j} \exp \Bigg\{- \frac{x ^ 2 - 2 \sqrt{n} \truestarj x}{2\varmlej} - \frac{M ^ {j \top} \varrand ^ {-1} M ^ j x ^ 2 + 2 M ^ {j\top} \varrand ^ {-1} \big(N ^ j \sqrt{n} \mlenj + P ^ j\big) x}{2} \\
    &\quad\quad\quad\quad 
    - \frac{Q ^ {j\top} \varrand ^ {-1} Q ^ j t ^ 2 + 2 Q ^ j \varrand ^ {-1} \big(M ^ j x + N ^ j \sqrt{n} \mlenj + P ^ j\big) t}{2}\Bigg\} dt \\
    &\propto \int_{I_1 ^ j} ^ {I_2 ^ j} \exp \Bigg\{- \frac{1}{2} \left(\frac{1}{\varmlej} + M ^ {j \top} \varrand ^ {-1} M ^ j\right) x ^ 2 - \left[\frac{\sqrt{n} \truestarj}{\varmlej} - M ^ {j\top} \varrand ^ {-1} \big(N ^ j \sqrt{n} \mlenj + P ^ j\big)\right] x \\
    &\quad\quad\quad\quad
    - \frac{1}{2} Q ^ {j\top} \varrand ^ {-1} Q ^ j t ^ 2 - Q ^ j \varrand ^ {-1} \big(M ^ j x + N ^ j \sqrt{n} \mlenj + P ^ j\big) t \Bigg\} dt\\
    & \propto \int_{I_1 ^ j} ^ {I_2 ^ j} \phi \big(\vartheta_j ^ {-1} \big(x - \nu_j \sqrt{n} \truestarj  - \varphi_j\big)\big) \phi \big(\kappa ^ {-1}\big(t - \kappa_j ^ 2 x - \delta_j\big)\big) dt
\end{aligned}
\]
where
\begin{align*}
\begin{gathered}
	\vartheta_j ^ 2 = \bigg(\frac{1}{\varmlej} + M ^ {j\top} \varrand ^ {-1} M ^ j \bigg) ^ {-1},\ \nu_j = \frac{\vartheta_j ^ 2}{\varmlej},\ \varphi_j = - \vartheta_j ^ 2  M ^ {j\top} \varrand ^ {-1} \big(N ^ j \sqrt{n} \mlenj + \D\big), \\
	\kappa ^ 2 = \Ej ^ \top \Psi \Ej,
    \quad\delta_j = - \Ej ^ \top \Psi R ^ \top \varrand ^ {-1} \big(N ^ j \sqrt{n} \mlenj + \D\big).
\end{gathered}	
\end{align*}
Therefore, the pivot in Proposition \ref{prop:pivot_exact_distr} can be rewritten as  
\[
	\dfrac{\bigintss_{-\infty} ^ {\sqrt{n} \mlej} \phi \big(\vartheta_j ^ {-1} \big(x - \nu_j \sqrt{n} \truestarj  - \varphi_j\big)\big) \big\{\Phi\big(\kappa ^ {-1}\big(I^j_1 - \kappa ^ 2 x - \delta_j\big)\big) - \Phi\big(\kappa ^ {-1}\big(I^j_2 - \kappa ^ 2 x - \delta_j\big)\big)\big\} dx}
	{\bigintss_{-\infty} ^ {\infty} \phi \big(\vartheta_j ^ {-1} \big(x - \nu_j \sqrt{n} \truestarj - \varphi_j\big)\big) \big\{\Phi\big(\kappa ^ {-1}\big(I^j_1 - \kappa ^ 2 x - \delta_j\big)\big) - \Phi\big(\kappa ^ {-1}\big(I^j_2 - \kappa ^ 2 x - \delta_j\big)\big)\big\} dx},
\]
which matches the expression of the exact pivot provided in \cite{panigrahi2023exact}.
\end{proof}

\begin{prop}
\label{prop:pivot_exact_distr}
Consider $W_0(\cdot)$ as defined in \eqref{eq:defn_W0}.
The conditional density 
$$\sqrt{n} \mlej\bigg\lvert \left\{\D = \dr, \Vj= \vjr, \sqrt{n}\mlenj = \sqrt{n}\mlenjr\right\}$$
is equal to
\[
	\dfrac{
	\phi \big(\sqrt{n} \mlej; \sqrt{n} \truestarj, \varmlej\big)
	W_0 \big(\sqrt{n} \mlej, \sqrt{n} \mlenjr\big)}
	{
	\bigintsss_{-\infty} ^ \infty \phi \big(x; \sqrt{n} \truestarj, \varmlej\big)
	W_0 \big(x, \sqrt{n} \mlenjr\big) d x
	}.
\]
\end{prop}

\begin{proof}
Based on the result in Lemma \ref{lem:joint_density}, it is easy to see that the density for
\[
	\left(\sqrt{n} \mlej, \sqrt{n} \mlenj, \Uj, \Vj, \Z \right)
\]
is proportional to
\begin{equation}
\label{eq:joint_density_UV}
\begin{aligned}
	&\phi \left(\sqrt{n} \mlej; \sqrt{n} \truej, \varmlej\right)
	 \phi \left(\sqrt{n} \mlenj; \sqrt{n} \truenj, \Lambda_2 \Lambda_2 ^ \top\right) \\
	&\quad\quad \times  \phi \left(\frac{T \Psi \Ej}{\Ej ^ \top \Psi \Ej} \Uj + T \Vj + \D + M ^ j \sqrt{n} \mlej + N ^ j \sqrt{n} \mlenj; 0_p, \varrand\right)
\end{aligned}
\end{equation}
for fixed $E$ and $\sgnr$.
This follows by decomposing $\lasso$ as 
\[
	\sqrt{n} |\lasso|
	= \frac{\Psi \Ej}{\Ej ^ \top \Psi \Ej} \Uj + \Vj.
\]
The density in \eqref{eq:joint_density_UV} when conditioned on the event in Proposition \ref{prop:conditional_event} is therefore proportional to
\[
\begin{aligned}
	&\phi \left(\sqrt{n} \mlej; \sqrt{n} \truej, \varmlej\right)
	\phi \left(\sqrt{n} \mlenj; \sqrt{n} \truenj, \Lambda_2 \Lambda_2 ^ \top\right) \\
	&\times \phi \left(\frac{T \Psi \Ej}{\Ej ^ \top \Psi \Ej} \Uj + T \vjr + \dr + M ^ j \sqrt{n} \mlej + N ^ j \sqrt{n} \mlenj; 0_p, \varrand\right) 
	\boldone_{\left\{I_1 \leq \Uj \leq I_2\right\}}.
\end{aligned}
\]
By integrating out or marginalizing over $\Uj$, we prove our claim.
\end{proof}


\subsection{Main results}
\subsubsection{Proof of Proposition \ref{prop:diff_target}}
\label{suppl:diff_target}
\begin{proof}
Define
\[
    Q_{h'}(b) = \EE_{\FF_n} \left[\widehat Q_{n; h'}(X_E b; Y)\right].
\]

Using the definition of $\trueprime$, we have that
\[
\begin{aligned}
    Q_{h'}(\trueprime) - Q_{h'}(\truestar) - \nabla Q_{h'}(\truestar) ^ \top (\trueprime - \truestar)
    &\leq - \nabla Q_{h'} (\truestar) ^ \top (\trueprime - \truestar) \\
    &\leq \big\|\nabla Q_{h'}(\truestar)\big\| \big\|\trueprime - \truestar\big\|,
\end{aligned}
\]
where the last step uses the H\"older’s inequality. 
Let 
$\mathcal{K}_h(u)=\mathcal{K}(u / h)$, where $$\mathcal{K}(u)=\int_{-\infty}^u K(v) d v,\ \text{ for } u \in \mathbb{R}.$$
Also, define $\varepsilon = y - x_E ^ \top \truestar$.
We note that $\nabla Q_{h'}(\truestar) = \EE \big[(\cK_{h'} (- \varepsilon)-\tau) x_E\big]$.

By applying an integration by parts, we have that
\[
\begin{aligned}
    \EE [\cK_{h'} (- \varepsilon) | x_E]
    &= \int_{-\infty}^{\infty} \mathcal{K} (-t/{h'}) d F_{\varepsilon | x_E}(t)
     = -\frac{1}{h'} \int_{-\infty}^{\infty} K(-t/{h'}) F_{\varepsilon | x_E}(t) d t \\
    &= \int_{-\infty}^{\infty} K(u) F_{\varepsilon | x_E}(-h' u) d u
     = \tau + \int_{-\infty}^{\infty} K(u) \int_0^{-h' u}\left\{f_{\varepsilon | x_E}(t)-f_{\varepsilon | x_E}(0)\right\} dt du.
\end{aligned}
\]
From the above display, it follows that 
\[
    |\EE [\cK_{h'} (- \varepsilon) | x_E] - \tau| 
    \leq 0.5 \kappa l_{0, E} {h'} ^ 2.
\]
Consequently,
\begin{equation}
\label{eq:upper1}
\begin{aligned}
    Q_{h'}(\trueprime) - Q_{h'}(\truestar) - \nabla Q_{h'}(\truestar) ^ \top (\trueprime - \truestar)
    &\leq 0.5 \kappa l_{0, E} {h'} ^ 2 \times \EE \big[\|x_E\|\big] \times \big\|\trueprime - \truestar\big\| \\
    &\leq 0.5 \kappa \sigma l_{0, E} {h'} ^ 2 \times \big\|\trueprime - \truestar\big\|.
\end{aligned}
\end{equation}
This gives us an upper bound on the expression in the left-hand side of \eqref{eq:upper1}.

By applying a Taylor series expansion, we have that
\begin{equation}
\label{eq:lower1}
\begin{aligned}
    Q_{h'}(\trueprime) - Q_{h'}(\truestar) &- \nabla Q_{h'}(\truestar) ^ \top (\trueprime - \truestar) \\
    &= (\trueprime - \truestar) ^ \top \int_0 ^ 1 \nabla ^ 2 Q_{h'}(t\trueprime + (1 - t) \truestar)) dt (\trueprime - \truestar),
\end{aligned}
\end{equation}
where $\nabla ^ 2 Q_{h'}(b) = \EE [K_{h'}(y-x_E ^ \top b) x_E x_E ^ \top]$.
Note that
\[
\begin{aligned}
    \EE [K_{h'}(y-x_E ^ \top b) | x_E]
    &= \int_{-\infty}^{\infty} K_{h'} (u-x_E ^ \top b+x_E ^ \top \truestar) f_{\varepsilon | x_E}(u) du \\
    &= \int_{-\infty}^{\infty} K_{h'}(v) f_{\varepsilon | x_E}(x_E ^ \top b-x_E ^ \top \truestar+ v) dv \\
    &= \int_{-\infty}^{\infty} K_{h'}(v) \left\{f_{\varepsilon | x_E}(v) + R_{h'}(b-\truestar)\right\} dv,
\end{aligned}
\]
where the last step uses the Lipschitz continuity of $f_{\varepsilon | x_E}(\cdot)$, with $R_{h'}(b-\truestar)$ satisfying $|R_{h'}(b-\truestar)| \leq l_{0, E} |x_E ^ \top b-x_E ^ \top \truestar|$. 
Then, we have that
\[
    \EE [K_{h'}(y-x_E ^ \top b) | x_E] = \EE [K_{h'}(\varepsilon ) | x_E] + \bar R_{h'}(b-\truestar)
\]
where
$$|\bar R_{h'}(b-\truestar)| \leq l_{0, E} |x_E ^ \top b-x_E ^ \top \truestar|.$$
Letting $b=t\trueprime + (1 - t) \truestar$, we have
\[
\begin{aligned}
    \int_0 ^ 1 \nabla ^ 2 Q_{h'}(t\trueprime + (1 - t) \truestar)) dt 
    &= \int_0 ^ 1 \EE \left[K_{h'}(y- t x_E ^ \top \trueprime + (1 - t) x_E ^ \top \truestar) x_E x_E ^ \top\right] dt \\
    &= \EE \left[K_{h'}(\varepsilon ) x_E x_E ^ \top\right] + \int_0 ^ 1 \EE \left[\bar R_{h'}(t\trueprime - t\truestar) x_E x_E ^ \top\right] dt.
\end{aligned}
\]
Combining with display \eqref{eq:lower1}, we note that
\begin{equation}
\label{eq:lower2}
\begin{aligned}
    (\trueprime - \truestar) ^ \top &\int_0 ^ 1 \nabla ^ 2 Q_{h'}(t\trueprime + (1 - t) \truestar)) dt (\trueprime - \truestar) \\
    &\geq m_{0, E} \big\|\trueprime - \truestar\big\| ^ 2 - 0.5 l_{0, E} \EE \big[\big|x_E ^ \top b-x_E ^ \top \truestar\big| ^ 3\big] \\
    &\geq m_{0, E} \big\|\trueprime - \truestar\big\| ^ 2 - 3 \sigma ^ 3 l_{0, E} \big\|\trueprime - \truestar\big\| ^ 3.
\end{aligned}
\end{equation}

Using the bounds in \eqref{eq:upper1} and \eqref{eq:lower2}, we find that 
\[
    m_{0, E} \big\|\trueprime - \truestar\big\| - 3 \sigma ^ 3 l_{0, E} \big\|\trueprime - \truestar\big\| ^ 2
    \leq 0.5 \kappa \sigma l_{0, E} {h'} ^ 2.
\]
Solving this quadratic function yields
\[
    \big\|\trueprime - \truestar\big\|
    \leq \frac{6 \kappa \sigma^4 l_{0, E}^2 {h'} ^ 2}{m_{0, E} + \Delta ^ {1/2}}
    \quad\text{or}\quad
    \big\|\trueprime - \truestar\big\|
    \geq \frac{m_{0, E} + \Delta ^ {1/2}}{6 \sigma ^ 3 l_{0, E}},
\]
where $\Delta = m_{0, E}^2 - 6 \kappa \sigma^4 l_{0, E}^2 {h'} ^ 2 > 0$. It remains to rule out the second bound.

To this end, suppose that 
\[
    \big\|\trueprime - \truestar\big\|
    \geq \frac{m_{0, E} + \Delta ^ {1/2}}{6 \sigma ^ 3 l_{0, E}} 
    > \left(\frac{\kappa}{6 \sigma ^ 2}\right) ^ {1/2} {h'} := r_0.
\]
Then there exists some $\eta \in (0, 1)$ such that $\widetilde{b}_n ^ E = \eta \trueprime + (1 - \eta) \truestar$ such that $ \|\widetilde{b}_n ^ E - \truestar\| = \eta \|\trueprime - \truestar\| = r_0$. Similar to displays \eqref{eq:upper1} and \eqref{eq:lower2}, we can show that
\[
    m_{0, E} \big\|\widetilde{b}_n ^ E - \truestar\big\| ^ 2 - 3 \sigma ^ 3 l_{0, E} \big\|\widetilde{b}_n ^ E - \truestar\big\| ^ 3
    \leq 0.5 \kappa \sigma l_{0, E} {h'} ^ 2 \big\|\widetilde{b}_n ^ E - \truestar\big\|.
\]
Canceling out the common factor $\|\widetilde{b}_n ^ E - \truestar\|$ from both sides, we obtain
\[
    \big\|\widetilde{b}_n ^ E - \truestar\big\|
    \leq \frac{0.5 \kappa \sigma l_{0, E} {h'} ^ 2}{m_{0, E} - 3 \sigma ^ 3 l_{0, E} (\kappa /6 \sigma ^ 2) ^ {1/2} h'} 
    < \frac{0.5 \kappa \sigma l_{0, E} {h'} ^ 2}{3 \sigma ^ 3 l_{0, E} (\kappa /6 \sigma ^ 2) ^ {1/2} h'} 
    = \big\|\widetilde{b}_n ^ E - \truestar\big\|,
\]
which leads to a contradiction. 
Thus, we have that
\[
     \big\|\trueprime - \truestar\big\|
    \leq \frac{6 \kappa \sigma^4 l_{0, E}^2 {h'} ^ 2}{m_{0, E} + \Delta ^ {1/2}}
    < \frac{6 \kappa \sigma^4 l_{0, E}^2 {h'} ^ 2}{m_{0, E}}.
\]
\end{proof}

\subsubsection{Proof of Proposition \ref{prop:conditional_event}}

\begin{proof}
We note that
\[	
	\left\{\D = \dr, \Vj= \vjr\right\}
    = \left\{- \text{\normalfont Diag}(\sgnr) \lasso < 0, \Z= \zr, \Vj= \vjr\right\}.
\]
Direct algebra leads us to the following equation:
\[
\begin{aligned}
	- \sqrt{n} \text{\normalfont Diag}(\sgnr) \lasso
	&= - \frac{\Psi \Ej}{\Ej ^ \top \Psi \Ej} \sqrt{n} \Ej ^ \top |\lasso|
	  - \sqrt{n} \bigg(I_{q, q} - \frac{\Psi \Ej}{\Ej ^ \top \Psi \Ej} \Ej ^ \top\bigg) |\lasso| \\
	&= - \frac{\Psi \Ej}{\Ej ^ \top \Psi \Ej} \Uj - \Vj.
\end{aligned}
\]
Then, we observe that the linear constraints $- \text{\normalfont Diag}(\sgnr) \lasso < 0$ are equivalent to
\[
	-\min_{k: \mathbb{S}_k  \Psi \Ej > 0}  \frac{\Ej ^ \top \Psi \Ej}{\mathbb{S}_k  \Psi \Ej} \mathbb{S}_k \Vj
	\leq \Uj 
	\leq -\max_{k: \mathbb{S}_k  \Psi \Ej < 0} \frac{\Ej ^ \top \Psi \Ej}{\mathbb{S}_k  \Psi \Ej} \mathbb{S}_k \Vj.
\]
Therefore,
\[
	\left\{- \text{\normalfont Diag}(\sgnr) \lasso < 0, \Z= \zr, \Vj= \vjr\right\}
	= \left\{I^j_1 \leq \Uj \leq I^j_2, \Z = \zr, \Vj = \vjr\right\},
\]
where 
\[
\begin{aligned}
	I^j_1 = -\min_{k: \mathbb{S}_k  \Psi \Ej > 0}  \frac{\Ej ^ \top \Psi \Ej}{\mathbb{S}_k  \Psi \Ej} \mathbb{S}_k \vjr,
    \quad
    I^j_2 = -\max_{k: \mathbb{S}_k  \Psi \Ej < 0} \frac{\Ej ^ \top \Psi \Ej}{\mathbb{S}_k  \Psi \Ej} \mathbb{S}_k \vjr.
\end{aligned}
\]  
Then the conclusion follows immediately.
\end{proof}

\subsubsection{Proof of Proposition \ref{prop:ls_exact}}

\begin{proof}
In the least squares setting, note that
\[
    \mlej = \SS_j \left(X_E ^ \top X_E\right) ^ {-1} X_E ^ \top Y,
\]
\[
    \mlenj = \begin{pmatrix} \mathbb{S}_{[E]\setminus j}\left(\mle - \frac{1}{\varmlej}\covmlej\mlej\right) \\ X_{E'} ^\top (X_E\mle - Y) \end{pmatrix},
\]
from which it follows that:
\begin{equation}
\label{eq:ls_equal1}
	\sqrt{n} \begin{pmatrix}
		\mlej -  \truestarj \\ 
		\mlenj - \truenj
	\end{pmatrix} 
	= \sqrt{n} \begin{pmatrix}
		\SS_j \left(X_E ^ \top X_E\right) ^ {-1} X_E ^ \top \left(Y - X_E \truestar\right)\\ 
		\mathbb{S}_{[E] \setminus j}\left(\mle - \truestarj - \frac{1}{\varmlej}\covmlej\mlej + \frac{1}{\varmlej}\covmlej\truestarj\right) \\ X_{E'} ^\top (X_E \mle- Y) - \EE_{\FF_n} \left[X_{E'} ^\top (X_E \true - Y)\right]
	\end{pmatrix}.
\end{equation}	
Additionally, we observe that
\[
\begin{gathered}
	\Lambda_1 = \begin{bmatrix} - \SS_j \left(X_E ^ \top X_E\right) ^ {-1} & 0 ^ \top_{p-q} \end{bmatrix} H ^ {1/2} \in \RR ^ {1\times p},\\
	\Lambda_2 =\begin{bmatrix} \Lambda_{2,1} \\ \Lambda_{2,2}\end{bmatrix}
		= \begin{bmatrix}  
		 \SS_{[E] \setminus j}\left(\frac{1}{\varmlej} \covmlej \SS_j \left(X_E ^ \top X_E\right) ^ {-1} - \left(X_E ^ \top X_E\right) ^ {-1}\right) & 0_{q-1, p-q} \\
		- X_{E'} ^ \top X_E \left(X_E ^ \top X_E\right) ^ {-1} & I_{p-q, p-q} \end{bmatrix} 
		H ^ {1/2},
\end{gathered}
\]
and
\[	
	\cZ = H ^ {- 1/2} \left(\begin{pmatrix} X_E ^ \top \left(X_E \truestar - Y\right) \\ X_{E'} ^ \top \left(X_E \truestar - Y\right) \end{pmatrix} - \begin{pmatrix} 0_q \\ \EE_{\FF_n} \left[X_{E'} ^ \top \left(X_E \truestar - Y\right)\right] \end{pmatrix}\right).
\]
After multiplying the matrices above and replacing $\left(X_E ^ \top X_E\right) ^ {-1} X_E ^ \top Y$ with $\mle$, we obtain:
\begin{equation}
\label{eq:ls_equal2}
	\begin{pmatrix} \Lambda_1 \\ \Lambda_2 \end{pmatrix} \sqrt{n} \cZ
	= \sqrt{n} \begin{pmatrix}
		- \SS_j \left(X_E ^ \top X_E\right) ^ {-1} X_E ^ \top \left(X_E \truestar - Y\right) \\ 
		\SS_{[E] \setminus j}\left(\frac{1}{\varmlej} \covmlej \SS_j \truestar - \frac{1}{\varmlej} \covmlej \SS_j \mle - \truestar + \mle\right) \\ 
		X_{E'} ^ \top X_E \mle - X_{E'} ^ \top Y - \EE_{\FF_n} \left[X_{E'} ^ \top \left(X_E \truestar - Y\right)\right]
	\end{pmatrix}.
\end{equation}
From the equations in \eqref{eq:ls_equal1} and \eqref{eq:ls_equal2}, we conclude that
\[
	\sqrt{n} \begin{pmatrix}
		\mlej -  \truestarj \\ 
		\mlenj - \truenj
	\end{pmatrix}
	= \begin{pmatrix} \Lambda_1 \\ \Lambda_2 \end{pmatrix} \sqrt{n} \cZ.
\]

In the same setting, we note that
\begin{align*}
\begin{gathered}
	M ^ j = - \frac{1}{\varmlej} \begin{bmatrix} H_{E, E} \\ H_{E', E} \end{bmatrix} (X_E ^ \top X_E) ^ {-1} \mathbb{S}_j ^ \top, 
	\ \ 
	N ^ j = \begin{bmatrix} - X_E ^ \top X_E \mathbb{S}_{[E]\setminus j} ^ \top & 0_{q, 1} \\ - X_{E'} ^ \top X_E \mathbb{S}_{[E]\setminus j} ^ \top & I_{q', q'}\end{bmatrix}, 
	\ \ 
	T = X ^ \top \! X_E \,\text{\normalfont Diag}(\sgnr).
\end{gathered}	
\end{align*}
Therefore, it holds that 
$$T \sqrt{n} |\lasso| + M ^ j \sqrt{n} \mlej + N ^ j \sqrt{n} \mlenj$$ 
is equal to
\begin{equation}
\label{eq:ls_equal3}
\begin{aligned}
	 & X ^ \top \! X_E \sqrt{n} \lasso
	 - \sqrt{n} \begin{pmatrix} \frac{1}{\varmlej} X_E ^ \top X_E \covmlej \mlej \\ \frac{1}{\varmlej} X_{E'} ^ \top X_E \covmlej \mlej\end{pmatrix} \\
	 &\quad\quad\quad+
	 	\sqrt{n} \begin{pmatrix}
	 		- X_E ^ \top X_E \mathbb{S}_{[E]\setminus j} ^ \top \mathbb{S}_{[E]\setminus j}\left(\mle - \frac{1}{\varmlej}\covmlej\mlej\right) \\
	 		- X_{E'} ^ \top X_E \mathbb{S}_{[E]\setminus j} ^ \top \mathbb{S}_{[E]\setminus j}\left(\mle - \frac{1}{\varmlej}\covmlej\mlej\right) + X_{E'} ^\top (X_E\mle - Y)
	 	\end{pmatrix} \\
	 &=  X ^ \top \! X_E \sqrt{n} \lasso 
	 - \sqrt{n} \begin{pmatrix}
	 	- X_E ^ \top X_E \mathbb{S}_{[E]\setminus j} ^ \top \mathbb{S}_{[E]\setminus j} \mle - X_E ^ \top X_E \mathbb{S}_{j} ^ \top \mathbb{S}_{j} \mle \\
	 	- X_{E'} ^ \top X_E \mathbb{S}_{[E]\setminus j} ^ \top \mathbb{S}_{[E]\setminus j} \mle - X_{E'} ^ \top X_E \mathbb{S}_{j} ^ \top \mathbb{S}_{j} \mle + X_{E'} ^\top (X_E\mle - Y) 
	 \end{pmatrix} \\
	 &= X ^ \top \! X_E \sqrt{n} \lasso - \sqrt{n} X ^ \top Y.
\end{aligned}	 
\end{equation}
In least squares regression, the Karush-Kuhn-Tucker (KKT) condition of stationary satisfies that
\[
	\sqrt{n} X ^ \top \left(X_E \lasso - Y\right) + \D - \sqrt{n}\rand = 0_p,
\]
Plugging the equation in \eqref{eq:ls_equal3}, we derive
\[
	T \sqrt{n} |\lasso| + M ^ j \sqrt{n} \mlej + N ^ j \sqrt{n} \mlenj + \D - \sqrt{n}\rand = 0_p,
\]
which completes our proof.

\end{proof}

\subsubsection{Proof of Corollary \ref{cor:pivot_cdf}}

\begin{proof}
As shown in Proposition \ref{prop:pivot_exact_distr}, the conditional density of
$$\sqrt{n} \mlej\bigg\lvert \left\{\D = \dr, \Vj= \vjr, \mlenj = \mlenjr\right\}$$
is equal to
\[
	\dfrac{
	\phi \big(\sqrt{n} \mlej; \sqrt{n} \truestarj, \varmlej\big)
	W_0 \big(\sqrt{n} \mlej, \sqrt{n} \mlenjr\big)}
	{
	\bigintsss_{-\infty} ^ \infty \phi \big(x; \sqrt{n} \truestarj, \varmlej\big)
	W_0 \big(x, \sqrt{n} \mlenjr\big) d x
	}.
\]
A probability integral transform based on the conditional density yields the variable
\[
	\dfrac{
	\bigintsss_{-\infty} ^ {\sqrt{n} \mlej} \phi \big(x; \sqrt{n} \truestarj, \varmlej\big)
	W_0 \big(x, \sqrt{n} \mlenj\big) dx}
	{
	\bigintsss_{-\infty} ^ \infty \phi \big(x; \sqrt{n} \truestarj, \varmlej\big)
	W_0 \big(x, \sqrt{n} \mlenj\big) d x
	}.
\]
Using the properties of a probability integral transform, we have that this variable is distributed as a $\operatorname{Unif} (0, 1)$ random variable conditional on the event in Proposition \ref{prop:conditional_event}.
\end{proof}

\section{Supporting results for Section \ref{sec:theory}}
\label{appdx:B}

\subsection{Moment and tail bounds}
\label{appendix:moment}

\begin{prop}
\label{lem:an_moment_bound}
Let 
$\mathcal{K}_h(u)=\mathcal{K}(u / h)$, where $\mathcal{K}(u)=\int_{-\infty}^u K(v) d v$, for $u \in \mathbb{R}$.
Let $e_{i, n} \in \RR ^ p$ be defined by
\[
    e_{i, n} = H ^ {- 1/2} \left\{\left(\cK_h \left(y_i - x_i ^ \top \true\right) - \tau\right) x_i - \EE_{\FF_n} \left[\left(\cK_h \left(y_i - x_i ^ \top \true\right) - \tau\right) x_i\right]\right\}.
\]
For any $n \in \mathbb{N}$ and $\mathbb{F}_n \in \cF_n$ under Assumption \ref{aspt:moment_bound}, we have
\begin{flalign*}
&\begin{aligned}
    \text{(i)}&\;\; \EE_{\FF_n} \left[\exp \left(\lambda \left\|e_{i, n}\right\|\right)\right]
    \leq \exp \left(4 \sigma ^ 2 \lambda ^ 2 \left\|H ^ {-1}\right\|\right) \text{ for all } \lambda \in \RR; \\[6pt]
    \text{(ii)}&\;\; \big(\EE_{\FF_n} [\|e_{i, n}\| ^ \gamma]\big)^{\frac{1}{\gamma}} \leq 2 \sigma \sqrt{\gamma} \left\|H ^ {-1}\right\| \text{ for all } \gamma \in \NN_+,
\end{aligned}&&
\end{flalign*}
where $\sigma \in \RR^{+}$ is a constant, and $\|\cdot\|$ denotes the operator norm for matrices and the $\ell_2$-norm for vectors. 
\end{prop}
\begin{proof}
First, note that $\cK_h(v) - \tau \leq 2$ for any $v \in \RR$, and $\EE_{\FF_n} [\exp(\lambda u ^ {\top} x_i)] \leq \exp (\sigma^2 \lambda^2)$ for all $\lambda \in \RR$ and $u \in \RR ^ p$ such that $\|u\| = 1$.
Therefore, we have that 
\begin{equation}
\label{eq:en_moment_bound}
    \EE_{\FF_n} \left[\exp \left(\lambda u ^ \top e_{i, n}\right)\right]
    \leq \exp \left(4 \sigma ^ 2 \lambda ^ 2 u ^ \top H ^ {-1} u \right)
    \leq \exp \left(4 \sigma ^ 2 \lambda ^ 2 \left\|H ^ {-1}\right\|\right)
\end{equation}
for all $\lambda \in \RR$ and $u \in \RR ^ p$ such that $\|u\| = 1$, and conclude that $e_{i, n}$ is  a sub-Gaussian random variable.
From here, it follows that
\[
    \EE_{\FF_n} \left[\exp \left(\lambda \left\|e_{i, n}\right\|\right)\right]
    \leq \exp \left(4 \sigma ^ 2 \lambda ^ 2 \left\|H ^ {-1}\right\|\right).
\]
The bound on the moments of $\left\|e_{i, n}\right\|$ directly follows by using properties of a sub-Gaussian random variable (see Proposition 2.5.2 in \cite{Ver18}). 
\end{proof}

\begin{prop}
\label{lem:Zn_moment_bound}
Consider $\Upsilon_n \in \RR ^ p$ as defined in Proposition \ref{prop:beta_gamma_expression}.
For any $\xi > 0$ and sufficiently large $n$ and $\mathbb{F}_n \in \cF_n$, under the conditions in Proposition \ref{lem:an_moment_bound}, we have that
\begin{flalign*}
&\begin{aligned}
    \text{(i)}&\;\; \sup_{\mathbb{F}_n \in \cF_n}\EE_{\FF_n} \left[\exp \left(\lambda u ^ \top \Upsilon_n \right)\right] \leq \exp \left(4 \sigma ^ 2 \lambda ^ 2 \left\|H ^ {-1}\right\|\right)\; \text{ for any } \lambda \in \mathbb{R} \text{ and unit vector } u; \\[7pt]
    \text{(ii)}&\;\; \sup_{\mathbb{F}_n \in \cF_n} \PP_{\FF_n} \left(\left\|\Upsilon_n\right\| \ge \xi\right)
    \leq 2 \exp \left(p \log 5 - \frac{\xi ^ 2}{16 \sigma ^ 2 \left\|H ^ {-1}\right\|}\right).
   \end{aligned}&&
\end{flalign*}
\end{prop}
\begin{proof}

As per the definition of $\Upsilon_n$, we have
\[
    \Upsilon_n = \frac{1}{\sqrt{n}} \sum_{i = 1} ^ n e_{i, n}.
\]
Applying Proposition \ref{lem:an_moment_bound}, for $\lambda \in \mathbb{R}$ and $u \in \mathbb{R}^p$ such that $\|u\|=1$, we observe that for any $\mathbb{F}_n \in \cF_n$:
\[
\begin{aligned}
	\EE_{\FF_n} \left[\exp \left(\lambda u ^ \top \Upsilon_n\right)\right]
	&= \EE_{\FF_n} \left[\exp \left(\frac{\lambda}{\sqrt{n}} \sum_{i = 1} ^ n u ^ \top e_{i, n}\right)\right] \\
    &= \prod_{i = 1} ^ n \EE_{\FF_n} \left[\exp \left(\frac{\lambda}{\sqrt{n}} u ^ \top e_{i, n}\right)\right] \\
    &\leq \prod_{i = 1} ^ n \exp \left(\frac{4 \sigma ^ 2 \lambda ^ 2}{n} \left\|H ^ {-1}\right\|\right) \\
    &\leq \exp \left(4 \sigma ^ 2 \lambda ^ 2 \left\|H ^ {-1}\right\|\right).
\end{aligned}
\]
This prove the assertion in (i).

 We apply the Markov's inequality and note that for $n\geq n_1$ and $\mathbb{F}_n \in \cF_n$, we can bound the right tail as 
\[
\begin{aligned}
    \PP_{\FF_n} \left(u ^ \top \Upsilon_n \ge \xi\right) 
    &= \PP_{\FF_n} \left(\exp \left(\lambda u ^ \top \Upsilon_n\right) \ge \exp \left(\lambda \xi\right)\right) \\
    &\leq \exp \left(- \lambda \xi\right) \times  \EE_{\FF_n} \left[\exp \left(\lambda u ^ \top \Upsilon_n\right)\right] \\
    &\leq \exp \left(- \lambda \xi + 4 \sigma ^ 2 \lambda ^ 2 \left\|H ^ {-1}\right\|\right) \\
    &\leq \exp \left(- \frac{\xi ^ 2}{16 \sigma ^ 2 \left\|H ^ {-1}\right\|}\right).
\end{aligned}
\]
The same conclusion applies to the left tail $\PP_{\FF_n} \left(u ^ \top \Upsilon_n \leq - \xi\right)$ and thus we have
\begin{equation}
\label{eq:markov}
    \PP_{\FF_n} \left(|u ^ \top \Upsilon_n| \ge \xi\right)
    \leq 2 \exp \left(- \frac{\xi ^ 2}{16 \sigma ^ 2 \left\|H ^ {-1}\right\|}\right)
\end{equation}
for any $\FF_n \in \cF_n$.

Lastly, to conclude our proof, let $\mathbb{U} ^ p = \{u \in \RR ^ p: \|u\| = 1\}$ denote the unit Euclidean sphere.
Also, let $\mathbb{G} ^ p$ denote a $1/2$-net of $\mathbb{U} ^ p$, which means that
\[
    \forall\, u \in \mathbb{U} ^ p, \exists\, v \in \mathbb{G} ^ p: \|u - v\| \leq \frac{1}{2}.
\]
Note that
\[
    \left\|\Upsilon_n\right\|
    = \max_{u \in \RR ^ p, \|u\| = 1} \left|u ^ \top \Upsilon_n\right|
    = \max_{u \in \RR ^ p, \|u\| = 1, v \in \mathbb{G} ^ p} \left|(u - v) ^ \top \Upsilon_n + v ^ \top \Upsilon_n\right| 
    \leq \frac{1}{2} \left\|\Upsilon_n\right\| + \max_{v \in \mathbb{G} ^ p} \left|v ^ \top \Upsilon_n\right|,
\]
which implies:
\[
    \left\|\Upsilon_n\right\| \leq 2 \max_{v \in \mathbb{G} ^ p} \left|v ^ \top \Upsilon_n\right|.
\]
Additionally, according to Corollary 4.2.13 in \cite{Ver18}, the cardinality $|\mathbb{G} ^ p|$ is bounded by $5 ^ p$.
Coupled with the probability bound in \eqref{eq:markov}, we obtain that for any $\xi > 0$ and sufficiently large $n$,
\[
\begin{aligned}
   \sup_{\mathbb{F}_n \in \cF_n}  \PP_{\FF_n} \left(\left\|\Upsilon_n\right\| \ge \xi\right)
    &= \sup_{\mathbb{F}_n \in \cF_n}\sum_{v \in \mathbb{G} ^ p} \PP_{\FF_n} \left(\left|v ^ \top \Upsilon_n\right| \ge \xi\right)\\
    &\leq 2 \times 5 ^ p \exp \left(- \frac{\xi ^ 2}{{16 \sigma ^ 2 \left\|H ^ {-1}\right\|}}\right)\\
    &= 2 \exp \left(p \log 5 - \frac{\xi ^ 2}{{16 \sigma ^ 2 \left\|H ^ {-1}\right\|}}\right).
\end{aligned}
\]
\end{proof}

\subsection{Conditional expectation}
\label{appendix:condexp}

\begin{prop}
\label{cor:conditional_expect_n}
Consider $\cZ$, defined as in Proposition \ref{prop:ls_exact}.
It holds that
$$\EE_{\cN} \left[\arb \circ \cP ^ {j \cdot E} (\cZ) \Big\lvert \event\right]= \dfrac{\EE_{\cN} \left[\arb \circ \cP ^ {j \cdot E} (\cZ) \times W_0 \left(\cZ\right)\right]}
	{\EE_{\cN} \left[W_0 \left(\cZ\right)\right]}.$$ 	
\end{prop}

\begin{proof}
As shown in the proof of Proposition \ref{prop:pivot_exact_distr}, the conditional density of
\[
    \sqrt{n} \mlej, \sqrt{n} \mlenj \Big\lvert \event
\]
is equal to
\begin{equation}
\label{eq:conditional_ratio}
    \dfrac{
	\phi \big(\sqrt{n} \mlej; \sqrt{n} \truestarj, \varmlej\big)
    \phi \big(\sqrt{n} \mlenj; \sqrt{n} \truenj, \Lambda_2 \Lambda_2 ^ \top\big)
	W_0 \big(\sqrt{n} \mlej, \sqrt{n} \mlenj\big)}
	{
	\bigintsss_{-\infty} ^ \infty \bigintsss_{-\infty} ^ \infty \phi \big(x_1; \sqrt{n} \truestarj, \varmlej\big)
    \phi \big(x_2; \sqrt{n} \truenj, \Lambda_2 \Lambda_2 ^ \top\big)
	W_0 \big(x_1, x_2\big) d x_1 d x_2
	}.    
\end{equation}
Using the definition of $\cZ$:
\[
    \sqrt{n} \begin{pmatrix}
		\mlej -  \truestarj \\ 
		\mlenj - \truenj
	\end{pmatrix}
	= \begin{pmatrix} \Lambda_1 \\ \Lambda_2 \end{pmatrix} \cZ,
\]
it follows that the conditional density of $\cZ$ at the point $z$ is equal to
$$
\dfrac{\phi(z; 0_p, I_{p,p}) W_0(z)}{\int \phi(z'; 0_p, I_{p,p}) W_0(z') dz'}= \left(\EE_{\cN} \left[W_0 \left(\cZ\right)\right]\right)^{-1} \phi(z; 0_p, I_{p,p}) W_0(z).
$$
This means that the ratio between the conditional and unconditional densities of $\cZ$, at the point $z$, is equal to
$$
R(z) = \left(\EE_{\cN} \left[W_0 \left(\cZ\right)\right]\right)^{-1} W_0(z),
$$
and therefore 
\begin{equation*}
    \begin{aligned}
\EE_{\cN} \left[\arb \circ \cP ^ {j \cdot E} (\cZ) \Big\lvert \event\right] 
&= \EE_{\cN} \left[\arb \circ \cP ^ {j \cdot E} (\cZ) R(\cZ) \right]\\
&= \dfrac{\EE_{\cN} \left[\arb \circ \cP ^ {j \cdot E} (\cZ) \times W_0 \left(\cZ\right)\right]}
	{\EE_{\cN} \left[W_0 \left(\cZ\right)\right]}.
 \end{aligned}
\end{equation*}
\end{proof}


\begin{prop}
\label{prop:conditional_expect}
Under the condition in Assumption \ref{aspt:randomization}, we have that
\begin{equation*}
    \lim_n \sup_{\FF_n} \left|\EE_{\FF_n} \left[\arb \circ \cP ^ {j \cdot E} (\Zn) \Big\lvert \event\right] - \dfrac{\EE_{\FF_n} \left[\arb \circ \cP ^ {j \cdot E} (\Zn) \times W_0 \left(\Zn\right)\right]}
	{\EE_{\FF_n} \left[W_0 \left(\Zn\right)\right]}\right| = 0.
\end{equation*}
\end{prop}
\begin{proof}
Following the same proof strategy as Proposition \ref{cor:conditional_expect_n}, we have that
\[
    \EE_{\FF_n} \left[\arb \circ \cP ^ {j \cdot E} (\Zn) \Big\lvert \event\right]
    = \dfrac{\EE_{\FF_n} \left[\arb \circ \cP ^ {j \cdot E} (\Zn) \times W_{0, n} \left(\Zn\right)\right]}
	{\EE_{\FF_n} \left[W_{0, n} \left(\Zn\right)\right]}.
\]
Using this display, we note that our assertion is equivalent to
\begin{equation}
\label{eq:prop_assertion2}
    \lim_n \sup_{\FF_n} \left|\dfrac{\EE_{\FF_n} \left[\arb \circ \cP ^ {j \cdot E} (\Zn) \times W_{0, n} \left(\Zn\right)\right]}
	{\EE_{\FF_n} \left[W_{0, n} \left(\Zn\right)\right]}
    - \dfrac{\EE_{\FF_n} \left[\arb \circ \cP ^ {j \cdot E} (\Zn) \times W_0 \left(\Zn\right)\right]}
	{\EE_{\FF_n} \left[W_0 \left(\Zn\right)\right]}\right| = 0,
\end{equation}
and hereafter we prove the claim in the above display.

Now, we can use the triangle inequality to bound the difference on the left-hand side of \eqref{eq:prop_assertion2} by
\[
    \cB_1 + \cB_2,
\]
where
\begin{gather*}
    \cB_1 = \lim_n \sup_{\FF_n} \left|\dfrac{\EE_{\FF_n} \left[\arb \circ \cP ^ {j \cdot E} (\Zn) \times W_{0, n} \left(\Zn\right)\right]}
	{\EE_{\FF_n} \left[W_{0, n} \left(\Zn\right)\right]}
    - \dfrac{\EE_{\FF_n} \left[\arb \circ \cP ^ {j \cdot E} (\Zn) \times W_{0, n} \left(\Zn\right)\right]}
	{\EE_{\FF_n} \left[W_0 \left(\Zn\right)\right]}\right|, \\
    \cB_2 = \lim_n \sup_{\FF_n} \left|\dfrac{\EE_{\FF_n} \left[\arb \circ \cP ^ {j \cdot E} (\Zn) \times W_{0, n} \left(\Zn\right)\right]}
	{\EE_{\FF_n} \left[W_0 \left(\Zn\right)\right]}
    - \dfrac{\EE_{\FF_n} \left[\arb \circ \cP ^ {j \cdot E} (\Zn) \times W_0 \left(\Zn\right)\right]}
	{\EE_{\FF_n} \left[W_0 \left(\Zn\right)\right]}\right|.
\end{gather*}
Observe that
\[
\begin{aligned}
	\cB_1 
	&\leq \lim_n \sup_{\FF_n} \int |\arb \circ \cP ^ {j \cdot E} (\Zn)| \left|\frac{W_{0, n} \left(\Zn\right)}{\EE_{\FF_n} \left[W_{0, n} \left(\Zn\right)\right]} - \frac{W_{0, n} \left(\Zn\right)}{\EE_{\FF_n} \left[W_0 \left(\Zn\right)\right]}\right| d\,\FF_n (\Zn) \\
	&\leq \lim_n \sup_{\FF_n} \underset{\arb \in \mathbb{C} ^ 3(\RR, \RR)}{\sup} |\arb| \times \left|\frac{\EE_{\FF_n} \left[W_{0, n} \left(\Zn\right) - W_0 \left(\Zn\right)\right]}{\EE_{\FF_n} \left[W_0 \left(\Zn\right)\right]}\right|,
\end{aligned}
\]
and that
\[
\begin{aligned}
	\cB_2 
    &= \lim_n \sup_{\FF_n} \left|\dfrac{\EE_{\FF_n} \left[\arb \circ \cP ^ {j \cdot E} (\Zn) \times W_{0, n} \left(\Zn\right)\right] - \EE_{\FF_n} \left[\arb \circ \cP ^ {j \cdot E} (\Zn) \times W_0 \left(\Zn\right)\right]}
	{\EE_{\FF_n} \left[W_0 \left(\Zn\right)\right]}\right| \\
    &\leq \lim_n \sup_{\FF_n} \underset{\arb \in \mathbb{C} ^ 3(\RR, \RR)}{\sup} |\arb| \times \left|\frac{\EE_{\FF_n} \left[W_{0, n} \left(\Zn\right) - W_0 \left(\Zn\right)\right]}{\EE_{\FF_n} \left[W_0 \left(\Zn\right)\right]}\right|.
\end{aligned}
\]
The condition stated in Assumption \ref{aspt:randomization} implies that
\[
\begin{aligned}
    \lim_n \sup_{\FF_n} \left|\frac{\EE_{\FF_n} \left[W_{0, n} \left(\Zn\right) - W_0 \left(\Zn\right)\right]}{\EE_{\FF_n} \left[W_0 \left(\Zn\right)\right]}\right| =0,
\end{aligned}
\]
which proves \eqref{eq:prop_assertion2}. 
\end{proof}

\subsection{Behavior of pivot and weight function}
\label{appendix:smoothness}

In Propositions \ref{prop:pivot_bound}, \ref{prop:W0} and \ref{prop:An_upper_bound},  we denote the partial derivative of a multivariate function $f: \RR ^ p \rightarrow \RR$, evaluated at $x = (x_1, x_2, \ldots, x_p) ^ \top \in \RR ^ p$, by
\[ 
    \partial_{i_1, \ldots, i_m} ^ m f(x) = \frac{\partial ^ m f(x)}{\partial x_{i_1} \ldots \partial x_{i_m}}.
\]

\begin{prop}
\label{prop:pivot_bound}
The third order partial derivatives of our pivot satisfy
\[
	\left|\partial_{i_1, i_2, i_3} ^ 3 \cP ^ {j\cdot E}(z)\right| 
    \leq \sum_{l = 0} ^ 3 C_l \big\|\mn ^ j z + \tail\big\| ^ l,
\]
where $C_0, \ldots, C_3$ are constants that do not depend on $n$.
\end{prop}

\begin{proof}
We begin by computing the first derivative of our pivot. 
Recall that we have
\[
    \cP ^ {j\cdot E}(z)
    = \dfrac{
	\bigintsss_{-\infty} ^ {\Lambda_1 z + \sqrt{n} \truestarj} \phi \left(x; \sqrt{n} \truestarj, \varmlej\right)
	W_0 \left(x, \Lambda_2 z + \sqrt{n} \truenj\right) dx}
	{
	\bigintsss_{-\infty} ^ \infty \phi \left(x; \sqrt{n} \truestarj, \varmlej\right)
	W_0 \left(x, \Lambda_2 z + \sqrt{n} \truenj\right) dx
	},
\]
and 
\[
    W_0 (x_1, x_2)
    = \int_{I^j_1} ^ {I^j_2} \phi \left(Q ^ j t + M ^ j x_1 + N ^ j x_2 + P ^ j; 0_p, \varrand\right) dt.
\]
Using the Leibniz integral rule, we note that
\begin{equation}
\label{eq:Leibniz}
\begin{aligned}
    \partial \cP ^ {j\cdot E}(z)
    = &\dfrac{
	\bigintsss_{-\infty} ^ {\Lambda_1 z + \sqrt{n} \truestarj} \phi \left(x; \sqrt{n} \truestarj, \varmlej\right)
	\partial W_0 \left(x, \Lambda_2 z + \sqrt{n} \truenj\right) dx}
	{
	\bigintsss_{-\infty} ^ \infty \phi \left(x; \sqrt{n} \truestarj, \varmlej\right)
	W_0 \left(x, \Lambda_2 z + \sqrt{n} \truenj\right) dx
	} \\
	&+ \dfrac{
	\phi \left(\Lambda_1 z + \sqrt{n} \truestarj; \sqrt{n} \truestarj, \varmlej\right)
	W_0 \left(\Lambda_1 z + \sqrt{n} \truestarj, \Lambda_2 z + \sqrt{n} \truenj\right) \Lambda_1}
	{
	\bigintsss_{-\infty} ^ \infty \phi \left(x; \sqrt{n} \truestarj, \varmlej\right)
	W_0 \left(x, \Lambda_2 z + \sqrt{n} \truenj\right) dx
	}\\
    &- \cP ^ {j\cdot E}(z) \times \dfrac{
	\bigintsss_{-\infty} ^ \infty \phi \left(x; \sqrt{n} \truestarj, \varmlej\right)
	\partial W_0 \left(x, \Lambda_2 z + \sqrt{n} \truenj\right) dx
	}
	{
	\bigintsss_{-\infty} ^ \infty \phi \left(x; \sqrt{n} \truestarj, \varmlej\right)
	W_0 \left(x, \Lambda_2 z + \sqrt{n} \truenj\right) dx
	}.
\end{aligned}
\end{equation}
It follows directly from the definition of $W_0$ that
\[
\begin{aligned}
	&\frac{\partial W_0 \left(x, \Lambda_2 z + \sqrt{n} \truenj\right)}{\partial z} 
	= \frac{\partial W_0 \left(x, \Lambda_2 z + \sqrt{n} \truenj\right)}{\partial \left(x, \Lambda_2 z + \sqrt{n} \truenj \right)}  \times 
    \frac{\partial \left(x, \Lambda_2 z + \sqrt{n} \truenj \right)}{\partial z} \\
    &\quad\quad = \int_{I^j_1} ^ {I^j_2} -2 \left(Q ^ j t + P ^ j + M ^ j x + N ^ j \Lambda_2 z + N ^ j \sqrt{n} \truenj\right) ^ \top \varrand ^ {-1}
    \begin{bmatrix}
        M ^ j & N ^ j
    \end{bmatrix}
    \begin{bmatrix}
        0_{1, p} \\
        \Lambda_2
    \end{bmatrix} \\
    &\quad\quad\quad\quad \times \phi \left(Q ^ j t + P ^ j + M ^ j x + N ^ j \Lambda_2 z + N ^ j \sqrt{n} \truenj; 0_p, \varrand\right) dt \\
    &\quad\quad \leq \big\|\mn z + \tail\big\| \times W_0 \left(x, \Lambda_2 z + \sqrt{n} \truenj\right).
\end{aligned}
\]
Thus, the first term on the right-hand side of \eqref{eq:Leibniz} is bounded as
\[
\begin{aligned}
	\dfrac{
	\bigintsss_{-\infty} ^ {\Lambda_1 z + \sqrt{n} \truestarj} \phi \left(x; \sqrt{n} \truestarj, \varmlej\right)
	\partial W_0 \left(x, \Lambda_2 z + \sqrt{n} \truenj\right) dx}
	{
	\bigintsss_{-\infty} ^ \infty \phi \left(x; \sqrt{n} \truestarj, \varmlej\right)
	W_0 \left(x, \Lambda_2 z + \sqrt{n} \truenj\right) dx
	}
	\leq \big\|\mn z + \tail\big\| \times \cP ^ {j\cdot E}(z).
\end{aligned}
\]
The second and third terms in the right hand side of \eqref{eq:Leibniz} can be bounded similarly.
This leads us to claim that
\[
\begin{aligned}
   \left|\partial_{i_1} \cP ^ {j\cdot E}(z)\right|
    &\leq C_0 \big\|\mn z + \tail\big\|,
\end{aligned}
\]
where $C_0$ is a constant.
A similar strategy can be applied to the higher order derivatives of our pivot, which we omit.
\end{proof}

\begin{prop}
\label{prop:W0}
Define
\[
    \Theta = \left(\varrand ^ {-1} - \frac{\varrand ^ {-1} Q ^ j Q ^ {j\top} \varrand ^ {-1}}{Q ^ {j\top} \varrand ^ {-1} Q ^ j}\right).
\]
We note that $\Theta \succeq 0$ and that
\[
    \left|\partial_{i_1, i_2, i_3} ^ 3 W_0 \left(z\right)\right|	
	\leq \sum_{l = 0} ^ 3 C_l' \big\|\mn ^ j z + \tail\big\| ^ l \times \operatorname{Exp} \left(\mn ^ j z + \tail, \Theta\right),
\]
where $C_0', \ldots, C_3'$ are constants that do not depend on $n$.
\end{prop}

\begin{proof}
For any vector $\xi \in \RR ^ p$, we have:
\[
    \xi ^ \top \Theta  \xi 
    = \xi ^ \top \varrand ^ {-1} \xi - \xi ^ \top \frac{\varrand ^ {-1} Q ^ j Q ^ {j\top} \varrand ^ {-1}}{Q ^ {j\top} \varrand ^ {-1} Q ^ j} \xi
    = \frac{1}{Q ^ {j\top} \varrand ^ {-1} Q ^ j} \left[\big(\xi ^ \top \varrand ^ {-1} \xi\big) \big(Q ^ {j\top} \varrand ^ {-1} Q ^ j\big) - \big(\xi ^ \top \varrand ^ {-1} Q ^ j\big) ^ 2\right].
\]
Let $\bar \xi = \varrand ^ {-1/2} \xi$ and $\bar Q ^ j = \varrand ^ {-1/2} Q ^ j$. 
By the Cauchy-Schwarz inequality, it follows that
\[
    \big(\xi ^ \top \varrand ^ {-1} \xi\big) \big(Q ^ {j\top} \varrand ^ {-1} Q ^ j\big) - \big(\xi ^ \top \varrand ^ {-1} Q ^ j\big) ^ 2
    = \big(\bar \xi ^ \top \bar \xi\big) \big(\bar Q ^ {j\top} \bar Q ^ j\big) -  \big(\bar \xi ^ \top \bar Q ^ j\big) ^ 2
    \ge 0.
\]
This implies that
\[
    \xi ^ \top \Theta  \xi \ge 0
\]
for an arbitrary vector $\xi$. 
Therefore, we have proved that $\Theta$ is a positive semidefinite matrix.

To prove the second part of our claim, we observe that
\[
\begin{aligned}
	W_0 (z) \!
	&\propto \!\! \int_{I_1 ^ j} ^ {\infty} \!
		\exp \left\{- \frac{1}{2} (Q ^ j t + \mn ^ j z + \tail) ^ \top \varrand ^ {-1} (Q ^ j t + \mn ^ j z + \tail)\right\} dt \\
  	&\propto \!\! \int_{I_1 ^ j} ^ {\infty} \!
		\exp \left\{- \frac{1}{2} (\mn ^ j z + \tail) ^ {\!\top} \Theta (\mn ^ j z + \tail)\right\}\\
    &\;\;\;\;\;\;\;\;\;\;\;\;\;\;\;\;\;\;\;\times \exp \left\{- \frac{Q ^ {j\top} \varrand ^ {-1} Q ^ j}{2} \bigg[t + \frac{Q ^ {j\top} \varrand ^ {-1} (\mn ^ j z + \tail)}{Q ^ {j\top} \varrand ^ {-1} Q ^ j}\bigg] ^ {\!2}\right\} dt.	
\end{aligned}
\]
Let
 \[
	\cL (z) = \frac{Q ^ {j\top} \varrand ^ {-1} (\mn ^ j z + \tail)}{Q ^ {j\top} \varrand ^ {-1} Q ^ j}.
\]
Using the representation in the previous part, we note that
\[
	W_0 (z) \propto \operatorname{Exp} \left(\mn ^ j z + \tail, \Theta\right) 
	\times \int_{I_1 ^ j + \cL(z)} ^ {\infty}
		\exp \left\{- \frac{Q ^ {j\top} \varrand ^ {-1} Q ^ j}{2} \tilde{t} ^ 2\right\} d\tilde{t},
\]
through a change of variables $\tilde{t} = t + \cL (z)$ in the integral.
An application of the Leibniz integral rule to the above display yields:
\[
	\left|\partial_{i_1, i_2, i_3} ^ 3 W_0 \left(z\right)\right|	 
	\leq \sum_{l = 0} ^ 3 C_l' \big\|\mn ^ j z + \tail\big\| ^ l \times \operatorname{Exp} \left(\mn ^ j z + \tail, \Theta\right).
\]
\end{proof}



\subsection{A Stein bound}

\begin{lem}
\label{prop:stein_bound}
Rewrite the variable $\Upsilon_n$ from Proposition \ref{prop:beta_gamma_expression} as $ \Upsilon_n = \sum_{i=1} ^ n \frac{e_{i, n}}{\sqrt{n}}$
with $e_{i, n}$ defined in Proposition \ref{lem:an_moment_bound}.
Define $\Upsilon_{i, n} = \frac{e_{i, n}}{\sqrt{n}}$ and $\Upsilon_n[-i] = \Upsilon_n - \Upsilon_{i, n}$.
Furthermore, let $e_{i, n} ^ *$ be an independent copy of $e_{i, n}$ and
\[
	\W = \Upsilon_n[-1] + \frac{\alpha}{\sqrt{n}} e_{1, n} + \frac{\kappa}{\sqrt{n}} e_{1, n} ^ *.
\]
Then, we have that
\[
\begin{aligned}
    &\big|\EE_{\FF_n} \big[G \left(\Upsilon_n\right)\big] - \EE_{\cN} \big[G \left(\cZ\right)\big]\big|\\
    &\lesssim \frac{1}{\sqrt{n}} \sum_{\substack{\lambda, \gamma \in \NN: \\ \lambda+\gamma \leq 3}} \sum_{i_1, i_2, i_3 \in [p]} 
	\EE_{\FF_n}
		\bigg[\|e_{1, n}\| ^ \lambda \|e_{1, n} ^ *\| ^ \gamma 
		\sup_{\alpha, \kappa \in[0,1]} \int_0 ^ 1 \frac{\sqrt{t}}{2} \, \EE_{\cN} \left[\left|\partial_{i_1, i_2, i_3} ^ 3 G \left(\sqrt{t} \W + \sqrt{1-t} \cZ\right)\right|\right] dt\bigg],
\end{aligned}
\]
where $G: \RR ^ p \rightarrow \RR$ is a Lebesgue-almost surely three times differentiable mapping and $\EE_{\cN} [|G(\cZ)|] < \infty$.
\end{lem}

\begin{proof}
The Stein bound in Lemma 2 from \cite{panigrahi2018carving} yields:
\[
\begin{aligned}
	\Big|\EE_{\FF_n} \big[G \left(\Upsilon_n\right)\!\big] - \EE_{\cN} \big[G \left(\cZ\right)\!\big]\Big| 
    \lesssim \frac{1}{\sqrt{n}}\sum_{\substack{\lambda, \gamma \in \NN: \\ \lambda+\gamma \leq 3}} \sum_{i_1, i_2, i_3}
	\EE_{\FF_n}
		\bigg[\|e_{1, n}\| ^ \lambda \|e_{1, n} ^ *\| ^ \gamma 
		\sup_{\alpha, \kappa \in[0,1]} \Big|\partial_{i_1, i_2, i_3} ^ 3 \cS_G \left(\W\right)\Big|\bigg].
\end{aligned}
\]
where 
\[
	\cS_G (z)
	= \int_0 ^ 1 \frac{1}{2 t}
	\left(\EE_{\cN} \big[G \big(\sqrt{t} z + \sqrt{1-t} \cZ\big)\big]
	- \EE_{\cN} \big[G (\cZ) \big]\right) d t.
\]
It follows directly that
\[
\begin{aligned}
	\Big|\partial_{i_1, i_2, i_3} ^ 3 \cS_G \left(\W\right) [i_1, i_2, i_3]\Big|
	= \int_0 ^ 1 \frac{\sqrt{t}}{2} \, \EE_{\cN} \left[\left|\partial_{i_1, i_2, i_3} ^ 3 G \left(\sqrt{t} \W + \sqrt{1-t} \cZ\right)\right|\right] dt.
\end{aligned}
\]
which gives us the claimed bound on the difference in the two expectations.
\end{proof}


\section{Proofs for results in Section \ref{sec:theory}}

\subsection{Proof of Theorem \ref{thm:relative_diff}}


To prove Theorem \ref{thm:relative_diff}, we first state a useful proposition and then turn to our proof.
\begin{prop}
\label{cor:tail}
Let $\widetilde{D}_n$ be an increasing sequence of sets in $\mathbb{R}^p$ such that 
$$\lim_n \sup_{\FF_n\in \cF_n} \mathbb{P}_{\mathbb{F}_n}[ \Upsilon_n \in \widetilde{D}_n^c] = 0.$$
Then, it holds that
\begin{gather*}
   \lim_n \sup_{\FF_n\in \cF_n} \frac{\EE_{\FF_n} [W_0 (\Upsilon_n) \boldone_{\widetilde{D}_n ^ c}(\Upsilon_n)]}{\EE_{\FF_n} [W_0 (\Upsilon_n)]} = 0.
 \end{gather*}
\end{prop}

\begin{proof}
Note, for any $\varepsilon > 0$, there exist $n_0$ such that for all $n \ge n_0$,
\[
	\EE_{\FF_n} \left[\boldone_{\widetilde{D}_n ^ c}(\Upsilon_n)\right] < \varepsilon.
\]
for all $\FF_n \in \cF_n$.
This implies that
\[
	\EE_{\FF_n} \left[W_0 (\Upsilon_n) \left(\boldone_{\widetilde{D}_n ^ c}(\Upsilon_n) - \varepsilon\right)\right]
	\leq  \sup_{z} W_0 (z) \times  \EE_{\FF_n} \left[\boldone_{\widetilde{D}_n ^ c}(\Upsilon_n) - \varepsilon\right] 
	< 0
\]
for all $n \ge n_0$ and $\FF_n\in \cF_n$.

Thus, we have shown that for any $\varepsilon > 0$, there exists a $n_0$ such that for all $n \ge n_0$ and $\FF_n\in \cF_n$,
\[
	\EE_{\FF_n} \left[W_0 (\Upsilon_n) \boldone_{\widetilde{D}_n ^ c}(\Upsilon_n)\right]
	< \varepsilon \EE_{\FF_n} \left[W_0 (\Upsilon_n)\right].
\]
Equivalently,
\[
	\lim_n \sup_{\FF_n \in \cF_n} \frac{\EE_{\FF_n} [W_0 (\Upsilon_n) \boldone_{\widetilde{D}_n ^ c}(\Upsilon_n)]}{\EE_{\FF_n} [W_0 (\Upsilon_n)]} = 0.
\]
\end{proof}


\begin{prop}
\label{prop:diff_Up}
Under Assumption \ref{aspt:error}, it holds that
\begin{equation}
\label{eq:diff_goal}
    \lim_{n} \sup_{\FF_n \in \cF_n} \Bigg|\dfrac{\EE_{\FF_n} \left[\arb \circ \pivot \times W_0 \left(\cZ_n\right)\right]}
	{\EE_{\FF_n} \left[W_0 \left(\cZ_n\right)\right]} - \dfrac{\EE_{\FF_n} \left[\arb \circ \pivotup \times W_0 \left(\Upsilon_n\right)\right]}
	{\EE_{\FF_n} \left[W_0 \left(\Upsilon_n\right)\right]}\Bigg|
    = 0.
\end{equation}
\end{prop}

\begin{proof}
Using the triangle inequality, we bound
\[
    \Bigg|\dfrac{\EE_{\FF_n} \left[\arb \circ \pivot \times W_0 \left(\cZ_n\right)\right]}
	{\EE_{\FF_n} \left[W_0 \left(\cZ_n\right)\right]} - \dfrac{\EE_{\FF_n} \left[\arb \circ \pivotup \times W_0 \left(\Upsilon_n\right)\right]}
	{\EE_{\FF_n} \left[W_0 \left(\Upsilon_n\right)\right]}\Bigg|
\]
by
\[
    \cE_1 + \cE_2,
\]
where
\[
\begin{aligned}
    \cE_1 = \Bigg|\dfrac{\EE_{\FF_n} \left[\arb \circ \pivot \times W_0 \left(\cZ_n\right)\right]}
	{\EE_{\FF_n} \left[W_0 \left(\cZ_n\right)\right]} - \dfrac{\EE_{\FF_n} \left[\arb \circ \pivot \times W_0 \left(\cZ_n\right)\right]}
	{\EE_{\FF_n} \left[W_0 \left(\Upsilon_n\right)\right]}\Bigg|, \\
    \cE_2 = \Bigg|\dfrac{\EE_{\FF_n} \left[\arb \circ \pivot \times W_0 \left(\cZ_n\right)\right]}
	{\EE_{\FF_n} \left[W_0 \left(\Upsilon_n\right)\right]} - \dfrac{\EE_{\FF_n} \left[\arb \circ \pivotup \times W_0 \left(\Upsilon_n\right)\right]}
	{\EE_{\FF_n} \left[W_0 \left(\Upsilon_n\right)\right]}\Bigg|.
\end{aligned}
\]

Observe that
\[
\begin{aligned}
    \cE_1 &\leq \Big|\EE_{\FF_n} \left[\arb \circ \pivot \times W_0 \left(\cZ_n\right)\right]\Big| \times
    \Bigg|\dfrac{1}{\EE_{\FF_n} \left[W_0 \left(\cZ_n\right)\right]} 
    - \dfrac{1}{\EE_{\FF_n} \left[W_0 \left(\Upsilon_n\right)\right]}\Bigg| \\
    &\leq \underset{\arb \in \mathbb{C} ^ 3(\RR, \RR)}{\sup} |\arb| \times \Bigg|\dfrac{\EE_{\FF_n} \left[W_0 \left(\cZ_n\right)\right]}{\EE_{\FF_n} \left[W_0 \left(\cZ_n\right)\right]} 
    - \dfrac{\EE_{\FF_n} \left[W_0 \left(\cZ_n\right)\right]}{\EE_{\FF_n} \left[W_0 \left(\Upsilon_n\right)\right]}\Bigg| \\
    &= \underset{\arb \in \mathbb{C} ^ 3(\RR, \RR)}{\sup} |\arb| \times \Bigg|\dfrac{\EE_{\FF_n} \left[W_0 \left(\Upsilon_n\right)\right] - \EE_{\FF_n} \left[W_0 \left(\cZ_n\right)\right]}{\EE_{\FF_n} \left[W_0 \left(\Upsilon_n\right)\right]}\Bigg|,
\end{aligned}
\]
and that
\[
\begin{aligned}
    \cE_2 = \Bigg|\dfrac{\EE_{\FF_n} \left[\arb \circ \pivot \times W_0 \left(\cZ_n\right)\right] - \EE_{\FF_n} \left[\arb \circ \pivotup \times W_0 \left(\Upsilon_n\right)\right]}
	{\EE_{\FF_n} \left[W_0 \left(\Upsilon_n\right)\right]}\Bigg|.
\end{aligned}
\]

In order to prove \eqref{eq:diff_goal}, it suffices to show that
\begin{equation} 
\label{eq:diff1}
    \lim_{n} \sup_{\FF_n \in \cF_n} \Bigg|\dfrac{\EE_{\FF_n} [W_0(\Upsilon_n + \widetilde \Delta_1) - W_0(\Upsilon_n)]}{\EE_{\FF_n} [W_0 (\Upsilon_n)]}\Bigg| = 0,
\end{equation}
where $\cZ_n = \Upsilon_n + \widetilde \Delta_1$.
Similar proof strategy applies to show that
\[
    \lim_{n} \sup_{\FF_n \in \cF_n} \Bigg|\dfrac{\EE_{\FF_n} \big[\arb \circ \cP ^ {j \cdot E} (\Upsilon_n + \widetilde \Delta_1) \times W_0 (\Upsilon_n + \widetilde \Delta_1)\big]\!\! - \EE_{\FF_n} \left[\arb \circ \pivotup \times W_0 \left(\Upsilon_n\right)\right]}
	{\EE_{\FF_n} \left[W_0 \left(\Upsilon_n\right)\right]}\Bigg| = 0.
\]
Note that, for any $\varepsilon > 0$, there exist $n_0$ such that for all $n \ge n_0$,
\[
	\EE_{\FF_n} \left[\left|\frac{W_0(\Upsilon_n + \widetilde \Delta_1)}{W_0(\Upsilon_n)} - 1\right|\right] < \varepsilon.
\]
for all $\FF_n \in \cF_n$. This implies that
\[
\begin{aligned}
	&\EE_{\FF_n} \left[W_0 (\Upsilon_n) \times \left\{\left|\frac{W_0(\Upsilon_n + \widetilde \Delta_1)}{W_0(\Upsilon_n)} - 1\right| - \varepsilon\right\}\right] \\
	&\quad\leq \sup_{z} W_0 (z) \times  \EE_{\FF_n} \left[\left|\frac{W_0(\Upsilon_n + \widetilde \Delta_1)}{W_0(\Upsilon_n)} - 1\right| - \varepsilon\right] 
	< 0
\end{aligned}
\]
for all $n \ge n_0$ and $\FF_n\in \cF_n$. Thus, we have shown that for any $\varepsilon > 0$, there exists a $n_0$ such that for all $n \ge n_0$ and $\FF_n\in \cF_n$,
\[
	\EE_{\FF_n} \left[W_0 (\Upsilon_n) \times \left|\frac{W_0(\Upsilon_n + \widetilde \Delta_1)}{W_0(\Upsilon_n)} - 1\right|\right]
	< \varepsilon \EE_{\FF_n} \left[W_0 (\Upsilon_n)\right].
\]
This is equivalent to \eqref{eq:diff1}. Our conclusion thus follows.

\end{proof}

\begin{proof}
In order to prove Theorem \ref{thm:main_weak_convergence}, it suffices to show 
\[
\begin{aligned}
	\lim_{n} \sup_{\FF_n \in \cF_n} \Big|\EE_{\FF_n} &\left[\arb \circ \pivot \Big\lvert \event\right] \\
	&- \EE_\cN \left[\arb \circ \pivotz \Big\lvert \event\right]\Big| = 0
\end{aligned}
\]
for any arbitrary function $\arb \in \mathbb{C} ^ 3 (\RR, \RR)$ with bounded derivatives up to the third order.
Furthermore, the claims made in Propositions \ref{cor:conditional_expect_n} and \ref{prop:conditional_expect} imply that
it is sufficient to prove:
\[
\begin{aligned}
	\lim_{n} \sup_{\FF_n \in \cF_n} \Bigg|\dfrac{\EE_{\FF_n} \left[\arb \circ \pivot \times W_0 \left(\cZ_n\right)\right]}
	{\EE_{\FF_n} \left[W_0 \left(\cZ_n\right)\right]} - \dfrac{\EE_{\cN} \left[\arb \circ \pivotz \times W_0 \left(\cZ\right)\right]}
	{\EE_{\cN} \left[W_0 \left(\cZ\right)\right]}\Bigg| = 0.
\end{aligned}
\]

\indent To obtain the stated sufficient conditions, we bound 
$$
\Bigg|\dfrac{\EE_{\FF_n} \left[\arb \circ \pivot \times W_0 \left(\cZ_n\right)\right]}
	{\EE_{\FF_n} \left[W_0 \left(\cZ_n\right)\right]} - \dfrac{\EE_{\cN} \left[\arb \circ \pivotz \times W_0 \left(\cZ\right)\right]}
	{\EE_{\cN} \left[W_0 \left(\cZ\right)\right]}\Bigg|,
$$
using the triangle inequality with
\[
	\cB_1 + \cB_2 + \cB_3 + \cB_4,
\]	
where 
\[
\begin{aligned}
    &\cB_1 = \Bigg|\dfrac{\EE_{\FF_n} \left[\arb \circ \pivot \times W_0 \left(\cZ_n\right)\right]}
	{\EE_{\FF_n} \left[W_0 \left(\cZ_n\right)\right]} - \dfrac{\EE_{\FF_n} \left[\arb \circ \pivotup \times W_0 \left(\Upsilon_n\right)\right]}
	{\EE_{\FF_n} \left[W_0 \left(\Upsilon_n\right)\right]}\Bigg|, \\
	&\cB_2 = \left|\dfrac{\EE_{\FF_n} \left[\arb \circ \pivotup \times W_0 \left(\Upsilon_n\right)\boldone_{\widetilde{D}_n}(\Upsilon_n)\right]}
	{\EE_{\FF_n} \left[W_0 \left(\Upsilon_n\right)\right]}
	-  \dfrac{\EE_{\FF_n} \left[\arb \circ \pivotup \times W_0 \left(\Upsilon_n\right)\boldone_{\widetilde{D}_n}(\Upsilon_n)\right]}
	{\EE_{\cN} \left[W_0 \left(\cZ\right)\right]}\right|, \\
	&\cB_3 = \left|\dfrac{\EE_{\FF_n} \left[\arb \circ \pivotup \times W_0 \left(\Upsilon_n\right)\boldone_{\widetilde{D}_n}(\Upsilon_n)\right]}
	{\EE_{\cN} \left[W_0 \left(\cZ\right)\right]}
	-  \dfrac{\EE_{\cN} \left[\arb \circ \pivotz \times W_0 \left(\cZ\right)\boldone_{\widetilde{D}_n}(\cZ)\right]}
	{\EE_{\cN} \left[W_0 \left(\cZ\right)\right]}\right|, \\
	&\cB_4 = \left|\dfrac{\EE_{\FF_n} \left[\arb \circ \pivotup \times W_0 \left(\Upsilon_n\right)\boldone_{\widetilde{D}_n ^ c}(\Upsilon_n)\right]}
	{\EE_{\FF_n} \left[W_0 \left(\Upsilon_n\right)\right]}
	-  \dfrac{\EE_{\cN} \left[\arb \circ \pivotz \times W_0 \left(\cZ\right)\boldone_{\widetilde{D}_n ^ c}(\cZ)\right]}
	{\EE_{\cN} \left[W_0 \left(\cZ\right)\right]}\right|.
\end{aligned}
\]

Proposition \ref{prop:diff_Up} shows that
\[
    \lim_{n} \sup_{\FF_n \in \cF_n} \cB_1 = 0.
\]
\indent Observe that
\[
\begin{aligned}
	\cB_2
	&\leq \left|\EE_{\FF_n} \left[\arb \circ \pivotup \times W_0 \left(\Upsilon_n\right)\boldone_{\widetilde{D}_n}(\Upsilon_n)\right] \right| \times \left|\frac{1}{\EE_{\FF_n} \left[W_0 \left(\Upsilon_n\right)\right]} - \frac{1}{\EE_{\cN} \left[W_0 \left(\cZ\right)\right]}\right|  \\
	&\leq \underset{\arb \in \mathbb{C} ^ 3(\RR, \RR)}{\sup} |\arb| \times \left|
	\frac{\EE_{\FF_n} \left[W_0 \left(\Upsilon_n\right)\boldone_{\widetilde{D}_n}(\Upsilon_n)\right]}{\EE_{\FF_n} \left[W_0 \left(\Upsilon_n\right)\right]}
	- \frac{\EE_{\FF_n} \left[W_0 \left(\Upsilon_n\right)\boldone_{\widetilde{D}_n}(\Upsilon_n)\right]}{\EE_{\cN} \left[W_0 \left(\cZ\right)\right]}
	\right|.
\end{aligned}
\]
By applying the triangle inequality once again, we further obtain the bound:
\[
\begin{aligned}
	\cB_2
	&\leq \underset{\arb \in \mathbb{C} ^ 3(\RR, \RR)}{\sup} |\arb| \times \bigg\{\left|
	\frac{\EE_{\FF_n} \left[W_0 \left(\Upsilon_n\right)\boldone_{\widetilde{D}_n}(\Upsilon_n)\right]}{\EE_{\FF_n} \left[W_0 \left(\Upsilon_n\right)\right]}
	- \frac{\EE_{\cN} \left[W_0 \left(\cZ\right)\boldone_{\widetilde{D}_n}(\cZ)\right]}{\EE_{\cN} \left[W_0 \left(\cZ\right)\right]}
	\right| \\
	&\quad\quad\quad\quad\quad\quad\quad\quad\quad\quad\quad\quad+ \left|
	\frac{\EE_{\cN} \left[W_0 \left(\cZ\right)\boldone_{\widetilde{D}_n}(\cZ)\right]}{\EE_{\cN} \left[W_0 \left(\cZ\right)\right]}
	- \frac{\EE_{\FF_n} \left[W_0 \left(\Upsilon_n\right)\boldone_{\widetilde{D}_n}(\Upsilon_n)\right]}{\EE_{\cN} \left[W_0 \left(\Upsilon_n\right)\right]}
	\right|\bigg\} \\
	&= \underset{\arb \in \mathbb{C} ^ 3(\RR, \RR)}{\sup} |\arb| \times \Bigg\{\left|
	\frac{\EE_{\FF_n} \left[W_0 \left(\Upsilon_n\right)\boldone_{\widetilde{D}_n ^ c}(\Upsilon_n)\right]}{\EE_{\FF_n} \left[W_0 \left(\Upsilon_n\right)\right]}
	- \frac{\EE_{\cN} \left[W_0 \left(\cZ\right)\boldone_{\widetilde{D}_n ^ c}(\cZ)\right]}{\EE_{\cN} \left[W_0 \left(\cZ\right)\right]}
	\right| \\
	&\quad\quad\quad\quad\quad\quad\quad\quad\quad\quad\quad\quad+ \left|
	\frac{\EE_{\cN} \left[W_0 \left(\cZ\right)\boldone_{\widetilde{D}_n}(\cZ)\right]}{\EE_{\cN} \left[W_0 \left(\cZ\right)\right]}
	- \frac{\EE_{\FF_n} \left[W_0 \left(\Upsilon_n\right)\boldone_{\widetilde{D}_n}(\Upsilon_n)\right]}{\EE_{\cN} \left[W_0 \left(\Upsilon_n\right)\right]}
	\right|\Bigg\},\\
	&\leq \underset{\arb \in \mathbb{C} ^ 3(\RR, \RR)}{\sup} |\arb| \times \Bigg\{
	2\sup_{\FF_n \in \cF_n}\frac{\EE_{\FF_n} \left[W_0 \left(\Upsilon_n\right)\boldone_{\widetilde{D}_n ^ c}(\Upsilon_n)\right]}{\EE_{\FF_n} \left[W_0 \left(\Upsilon_n\right)\right]}
	 + \RD ^ {(1)}\Bigg\}.
\end{aligned}
\]

It is easy to see that $\cB_3$ is equal to $\RD ^ {(2)}$, and that
\[
	\cB_4
	\leq \underset{\arb \in \mathbb{C} ^ 3(\RR, \RR)}{\sup} |\arb| \times 2
	\sup_{\FF_n \in \cF_n}\frac{\EE_{\FF_n} \left[W_0 \left(\Upsilon_n\right)\boldone_{\widetilde{D}_n ^ c}(\Upsilon_n)\right]}{\EE_{\FF_n} \left[W_0 \left(\Upsilon_n\right)\right]}.
\]
Thus, we conclude that
\[
\begin{aligned}
    &\lim_{n} \sup_{\FF_n \in \cF_n} \Bigg|\dfrac{\EE_{\FF_n} \left[\arb \circ \pivotup\times W_0 \left(\Upsilon_n\right)\right]} 
	{\EE_{\FF_n} \left[W_0 \left(\Upsilon_n\right)\right]}  - \dfrac{\EE_{\cN} \left[\arb \circ \pivotz \times W_0 \left(\cZ\right)\right]}
	{\EE_{\cN} \left[W_0 \left(\cZ\right)\right]}\Bigg| \\
    &\leq \underset{\arb \in \mathbb{C} ^ 3(\RR, \RR)}{\sup} |\arb| \times \lim_{n} \sup_{\FF_n \in \cF_n} \RD ^ {(1)} + \lim_{n} \sup_{\FF_n \in \cF_n} \RD ^ {(2)} \\
    &\quad\quad+ \underset{\arb \in \mathbb{C} ^ 3(\RR, \RR)}{\sup} |\arb|  \times  4\lim_{n} \sup_{\FF_n \in \cF_n} \frac{\EE_{\FF_n} \left[W_0 \left(\Upsilon_n\right)\boldone_{\widetilde{D}_n ^ c}(\Upsilon_n)\right]}{\EE_{\FF_n} \left[W_0 \left(\Upsilon_n\right)\right]}.
\end{aligned}
\]
Proposition \ref{cor:tail} yields
$$\lim_{n} \sup_{\FF_n \in \cF_n} \frac{\EE_{\FF_n} \left[W_0 \left(\Upsilon_n\right)\boldone_{\widetilde{D}_n ^ c}(\Upsilon_n)\right]}{\EE_{\FF_n} \left[W_0 \left(\Upsilon_n\right)\right]}=0,$$
thereby proving our result.
\end{proof}

\subsection{Proof of Theorem \ref{thm:rd_limit}}

For ease of presentation, we fix some notations that will be used throughout the proof.

Note that there are three scenarios: 1. $- \infty < I_1 ^ j < \infty$ and $I_2 ^ j = \infty$; 2. $I_1 ^ j = - \infty$ and $- \infty < I_2 ^ j < \infty$; 3. $- \infty < I_1 ^ j < I_2 ^ j < \infty$. 
The proof strategy employed for the first scenario applies to the second and third scenario in a similar manner. 
Therefore, we focus on the first scenario hereafter.

For our sequence of parameters, note that 
$$L_n +Q ^ j I_1 ^ j =  \begin{bmatrix} M ^ j & N ^ j \end{bmatrix} r_n \beta + P ^ j + Q ^ j I_1 ^ j = O(r_n).$$
In the remaining proof, with a slight abuse of notation, we will write 
\begin{equation}
    L_n +Q ^ j I_1 ^ j = - r_n \bar b,
    \label{reparam}
\end{equation}
where $\bar b \in \RR ^ p$ is a fixed vector that does not depend on $n$.

Lastly, for a vector $x \in \RR ^ p$ and a positive semidefinite matrix $\Sigma \in \RR ^ {p \times p}$, we use $\operatorname{Exp}(x, \Sigma)$ as a shorthand to denote the function $\exp \left(-\frac{1}{2} x^{\top} \Sigma x\right)$.

\begin{proof}
We divide our proof into two different cases.

\begin{enumerate}
\item Case I: $r_n \leq C$. In the first case, our conclusion follows from Theorem \ref{thm:bound_case} in \ref{CaseI}.
\item Case II: $r_n\to \infty$ and $r_n = o(n ^ {1/6})$. In the second case, our conclusion follows from Theorem \ref{thm:bound_case}. Details are provided in \ref{CaseII}.
\end{enumerate}
\end{proof}

\subsubsection{Case I: $r_n \leq C$}
\label{CaseI}
\begin{thm}
\label{thm:bound_case}
Suppose that the conditions in Assumptions \ref{aspt:moment_bound} and \ref{aspt:randomization} are met with $r_n \leq C$ for a constant $C$.
Then, it holds that
$$\lim_{n} \sup_{\FF_n \in \cF_n} \RD ^ {(1)} = 0, \quad \lim_{n} \sup_{\FF_n \in \cF_n} \RD ^ {(2)}  = 0.$$
\end{thm}
\begin{proof}
First, we prove that
\begin{flalign}
&\begin{aligned}
\text{(i)}&\;\; \sup_{\FF_n \in \cF_n}\left|\EE_{\FF_n} \left[W_0 \left(\Upsilon_n\right)\right] - \EE_{\cN} \left[W_0 \left(\cZ\right)\right]\right| \le \frac{1}{\sqrt{n}}; \\
\text{(ii)}&\;\; \sup_{\FF_n \in \cF_n}\left|\EE_{\FF_n} \left[\arb \circ \pivotup \times W_0 \left(\Upsilon_n\right)\right] - \EE_{\cN} \left[\arb \circ \pivotz \times W_0 \left(\cZ\right)\right]\right| \le \frac{1}{\sqrt{n}}.
\end{aligned}&&
\label{first:part}
\end{flalign}

For the sake of brevity, we define the real-valued functions
\begin{equation*}
\begin{aligned}
    & G^{(1)}\left(z\right) = W_0\left(z\right), \ \  G^{(2)}\left(z\right) = H \circ \cP ^ {j \cdot E}\left(z\right) \times W_0\left(z\right).
\end{aligned}
\end{equation*}
We then apply the Stein bound in Proposition \ref{prop:stein_bound} with $G= G^{(l)}$ for $l\in \{1,2\}$, and use the definitions of $\W$, $\Upsilon_n[-1]$, $e_{1, n}$, and $e_{1, n}^*$ provided in this result.
Simplifying the Stein bound, we obtain 
\begin{equation}
\begin{aligned}
& \left|\EE_{\FF_n} \left[G ^ {(l)} (\Upsilon_n)\right] - \EE_{\cN} \left[G ^ {(l)}(\cZ)\right]\right|\\
&\lesssim \frac{1}{\sqrt{n}} \sum_{\substack{\lambda, \gamma \in \NN: \\ \lambda+\gamma \leq 3}} \sum_{i_1, i_2, i_3} 
	\EE_{\FF_n}
		\bigg[\|e_{1, n}\| ^ \lambda \|e_{1, n} ^ *\| ^ \gamma 
		\sup_{\alpha, \kappa \in[0,1]} \int_0 ^ 1 \frac{\sqrt{t}}{2} \, \EE_{\cN} \bigg[\sum_{l = 0} ^ 3 \sqrt{t} \big\|\W\| ^ l + \sqrt{1-t} \|\cZ\| ^ l\bigg] dt\bigg],
\end{aligned}
\label{bdd:1}
\end{equation} 
by using the behavior of the pivot and weight function in Propositions \ref{prop:pivot_bound} and \ref{prop:W0}, which yields:
\[
\begin{aligned}
 	\left|\partial_{i_1, i_2, i_3} ^ 3 G ^ {(l)} \left(\sqrt{t} \W + \sqrt{1-t} z\right)\right| 
    &\lesssim \sum_{l = 0} ^ 3 \big\|\mn ^ j \sqrt{t} \W + \mn ^ j \sqrt{1-t} z + \tail\big\| ^ l \\
    &\lesssim \sum_{l = 0} ^ 3 \sqrt{t} \big\|\W\| ^ l + \sqrt{1-t} \|z\| ^ l.
\end{aligned}
\]

We conclude the proof by noting that the bound on the right-hand side of \eqref{bdd:1} satisfies:
\begin{equation*}
\begin{aligned}
& \frac{1}{\sqrt{n}} \sum_{\substack{\lambda, \gamma \in \NN: \\ \lambda+\gamma \leq 3}} \sum_{i_1, i_2, i_3} 
	\EE_{\FF_n}
		\bigg[\|e_{1, n}\| ^ \lambda \|e_{1, n} ^ *\| ^ \gamma 
		\sup_{\alpha, \kappa \in[0,1]} \int_0 ^ 1 \frac{\sqrt{t}}{2} \, \EE_{\cN} \bigg[\sum_{l = 0} ^ 3 \sqrt{t} \big\|\W\| ^ l + \sqrt{1-t} \|\cZ\| ^ l\bigg] dt\bigg]\\
& \lesssim	\frac{1}{\sqrt{n}} \sum_{\lambda, \gamma \in \NN: \lambda+\gamma \leq 3} 
	\EE_{\FF_n}
		\bigg[\|e_{1, n}\| ^ \lambda \|e_{1, n} ^ *\| ^ \gamma 
		\sup_{\alpha, \kappa \in[0,1]} \big\|\W\big\| ^ 3\bigg]\\
&\lesssim \frac{1}{\sqrt{n}} \sup_n \sup_{\FF_n \in \cF_n} \EE_{\FF_n} \left[\left\|e_{1, n}\right\| ^ 6\right],		
\end{aligned}
\end{equation*}
where the final display uses the independence of the variables $\Upsilon_n[-1]$, $e_{1, n}$, and $e_{1, n}^*$.
At last, using Proposition \ref{lem:an_moment_bound}, we note that  
$$\sup_n \sup_{\FF_n \in \cF_n} \EE_{\FF_n} \left[\left\|e_{1, n}\right\| ^ 6\right]<\infty,$$
which proves the bounds in \eqref{first:part}.

In the next step, we prove that
\begin{equation}
\EE_{\cN} \left[W_0 \left(\cZ\right)\right] \ge C
\label{second:part}
\end{equation}
To do so, observe that
\[
    \EE_{\cN} \left[W_0\left(\cZ\right)\right]
    \propto \EE_{\cN} \left[\int_{I_1 ^ j} ^ {\infty} 
		\exp \left\{- \frac{1}{2} (Q ^ j t + \mn ^ j z + \tail) ^ \top \varrand ^ {-1} (Q ^ j t + \mn ^ j z + \tail)\right\} dt\right].
\]
Choose $C_0$, a positive constant, such that $\mathbb{P}_{\cN}[\cZ \in  \cS_{C_0}]\geq \frac{1}{2}$ for $\cS_{C_0}= [-C_0 \cdot 1_p, C_0 \cdot 1_p]$.
Using the parameterization in \eqref{reparam},
we have that
\[
    \EE_{\cN} \left[W_0\left(\cZ\right)\right]
    \geq  \EE_{\cN} \left[\int_{I_1 ^ j} ^ {\infty} \!
		\exp \left\{- \frac{1}{2} (Q ^ j t +  \mn ^ j z - r_n \bar b - Q ^ j I_1 ^ j) ^ \top \varrand ^ {-1} (Q ^ j t + \mn ^ j z - r_n \bar b- Q ^ j I_1 ^ j)\right\} dt \times \boldone_{\cZ \in \cS} \right].
\]
For $z\in \cS_{C_0}$ and $r_n \leq C$, we have that
$$
\exp \left\{- \frac{1}{2} (Q ^ j t +  \mn ^ j z - r_n \bar b - Q ^ j I_1 ^ j) ^ \top \varrand ^ {-1} (Q ^ j t + \mn ^ j z - r_n \bar b- Q ^ j I_1 ^ j)\right\} \geq C_1,
$$ 
for a positive constant $C_1$.

Based on the above-stated observation, we have 
\[
    \EE_{\cN} \left[W_0\left(\cZ\right)\right]
    \ge C_1 \times \mathbb{P}_{\cN}[\cZ \in  \cS_{C_0}]
    \ge C_1 \times \frac{1}{2},
\]
which proves \eqref{second:part}.

In the definition of $\RD ^ {(1)}$ and $\RD ^ {(2)}$, we fix  $\widetilde{D}_n=\mathbb{R}^p$.
Our proof is complete by combining the results in \eqref{first:part} and \eqref{second:part}.
\end{proof}

\subsubsection{Case II: $r_n \rightarrow \infty$ and $r_n = o_p (n ^ {1/6})$}
\label{CaseII}

To analyze this case, we start by presenting some useful results.
 
 \begin{prop}
\label{prop:An_upper_bound}
There exist Lebesgue-almost everywhere differentiable functions $\cA_n, \bar \cA_n: \RR ^ p \rightarrow \RR$ such that the following assertions hold.
\begin{enumerate}
\item For $Q ^ {j\top} \varrand ^ {-1} \bar b < 0$ and $m \in \{ 0, 1, 2, 3\}$, it holds that 
$$W_0 (z) = \operatorname{Exp} \big(\mn ^ j z - r_n \bar b, \varrand ^ {-1}\big) \!\times\! \cA_n(z), \text{ and  } \underset{\FF_n \in \cF_n} {\sup}\underset{z \in \cD_n}{\sup} |r_n| \left|\partial_{i_1, \ldots i_m} ^ m \cA_n (z) \right| < \infty,$$ 
where $\cD_n = [-c_0 r_n \cdot 1_p, c_0 r_n \cdot 1_p]$ and $c_0$ is a positive constant such that
\begin{equation*}
\label{eq:c_range}
     c_0 < \frac{1}{2} \frac{|Q ^ {j\top} \varrand ^ {-1} \bar b|}{\sqrt{p} \|R^{j \top}\varrand ^ {-1}Q ^ {j}  \|}.
\end{equation*}
\item For $Q ^ {j\top} \varrand ^ {-1} \bar b > 0$ and $m \in \{ 0, 1, 2, 3\}$, it holds that 
$$W_0 (z) = \operatorname{Exp} \big(\mn ^ j z  - r_n \bar b, \Theta\big) \!\times\! \bar \cA_n (z) \text{ and } \sup_{\FF_n \in \cF_n} \sup_{z\in \mathbb{R}^p}\left|\partial_{i_1, \ldots i_m} ^ m \bar \cA_n (z) \right| < \infty.$$
\end{enumerate}
\end{prop}

\begin{proof}
We consider two cases depending on the sign of the scalar $Q ^ {j\top} \varrand ^ {-1} \bar b$.
First, let us consider the case when $Q ^ {j\top} \varrand ^ {-1} \bar b < 0$.
Observe that
\begin{equation}
\label{eq:decomp_W0}
\begin{aligned}
    W_0 \left(z\right)
    &\propto \operatorname{Exp} \left(\mn ^ j z + \tail, \Theta\right) \times \int_{I_1 ^ j} ^ {\infty} \exp \left\{- \frac{Q ^ {j\top} \varrand ^ {-1} Q ^ j}{2} \bigg[t + \frac{Q ^ {j\top} \varrand ^ {-1} (\mn ^ j z + \tail)}{Q ^ {j\top} \varrand ^ {-1} Q ^ j}\bigg] ^ 2\right\} dt \\
    &\propto \operatorname{Exp} \left(\mn ^ j z + \tail + Q ^ j I_1 ^ j, \Theta\right) \times \int_{I_1 ^ j} ^ {\infty} \exp \left\{- \frac{Q ^ {j\top} \varrand ^ {-1} Q ^ j}{2} \bigg[t + \frac{Q ^ {j\top} \varrand ^ {-1} (\mn ^ j z + \tail)}{Q ^ {j\top} \varrand ^ {-1} Q ^ j}\bigg] ^ 2\right\} dt.
\end{aligned}
\end{equation}
We write the weight function $W_0 \left(z\right)$ as
\[
\begin{aligned}
    W_0 \left(z\right) 
    = \operatorname{Exp} \left(\mn ^ j z + \tail + Q ^ j I_1 ^ j, \varrand ^ {-1}\right) 
    \times \cA_n(z),
\end{aligned}
\]
where 
\[
\begin{aligned}
	\cA_n (z) = C &\exp \left\{\frac{1}{2} \left(\mn ^ j z + \tail + Q ^ j I_1 ^ j\right) ^ \top \frac{\varrand ^ {-1} Q ^ j Q ^ {j\top} \varrand ^ {-1}}{Q ^ {j\top} \varrand ^ {-1} Q ^ j}\left(\mn ^ j z + \tail + Q ^ j I_1 ^ j\right)\right\} \\
	&\quad\quad \times 
	\int_{I_1 ^ j} ^ {\infty}
		\exp \left\{- \frac{Q ^ {j\top} \varrand ^ {-1} Q ^ j}{2} \bigg[t + \frac{Q ^ {j\top} \varrand ^ {-1} (\mn ^ j z + \tail)}{Q ^ {j\top} \varrand ^ {-1} Q ^ j}\bigg] ^ 2\right\} dt,
\end{aligned}
\]
for a constant $C$.

By substituting $\tilde{t} = t - I_1 ^ j$ in the integral involved in the expression of $\cA_n(z)$, we obtain 
\[
\begin{aligned}
    \cA_n (z) &\propto \exp \left\{\frac{1}{2} \left(\mn ^ j z - r_n \bar b\right) ^ \top \frac{\varrand ^ {-1} Q ^ j Q ^ {j\top} \varrand ^ {-1}}{Q ^ {j\top} \varrand ^ {-1} Q ^ j}\left(\mn ^ j z - r_n \bar b\right)\right\} \\
	&\quad\quad \times 
	\int_{0} ^ {\infty} 
		\exp \left\{-\frac{1}{2} \left(\mn ^ j z - r_n \bar b + Q ^ j \tilde{t}\right) ^ \top \frac{\varrand ^ {-1} Q ^ j Q ^ {j\top} \varrand ^ {-1}}{Q ^ {j\top} \varrand ^ {-1} Q ^ j}\left(\mn ^ j z - r_n \bar b + Q ^ j \tilde{t}\right)\right\} d \tilde{t}.
\end{aligned}
\]
A direct algebraic simplification of the expression for $\cA_n(z)$ leads us to:
\begin{equation}
\label{eq:rewrite_An}
\begin{aligned}
	\cA_n (z) 
	&\propto \int_{0} ^ {\infty}
		\exp \left\{
		- \frac{1}{2} \tilde{t} Q ^ {j\top} \varrand ^ {-1} Q ^ j \tilde{t}\right\} \times
		\exp \left\{- \tilde{t} Q ^ {j\top} \varrand ^ {-1} \left(\mn ^ j z - r_n \bar b\right)\right\} d\tilde{t} \\
	&\leq \int_0 ^ {\infty} \exp \left\{- \tilde{t} Q ^ {j\top} \varrand ^ {-1} \left(\mn ^ j z - r_n \bar b\right) \right\} d\tilde{t}.
\end{aligned}
\end{equation}
For $z\in \cD_n$, we have
\[
\begin{aligned}
    \cA_n (z) \lesssim \frac{1}{Q ^ {j\top} \varrand ^ {-1} \left(\mn ^ j z - r_n \bar b\right)}
     \lesssim \frac{1}{r_n |Q ^ {j\top} \varrand ^ {-1} \bar b|}.
\end{aligned}
\]
This proves the assertion for $m = 0$. \\
\indent From the display in \eqref{eq:rewrite_An}, we note that
\[
\begin{aligned}
	\partial ^ 1 \cA_n (z)
	\propto \int_{0} ^ {\infty}
		&- \tilde{t} Q ^ {j\top} \varrand ^ {-1} \mn ^ j
		\exp \left\{
		- \frac{1}{2} \tilde{t} Q ^ {j\top} \varrand ^ {-1} Q ^ j \tilde{t}\right\} 
        \times \exp \left\{- \tilde{t} Q ^ {j\top} \varrand ^ {-1} \left(\mn ^ j z - r_n \bar b\right) \right\} d\tilde{t}.
\end{aligned}
\]
This leads us to observe that
\[
\begin{aligned}
	\left|\partial_{i_1} ^ 1 \cA_n (z)\right|
	&\lesssim \int_{0} ^ {\infty} - \tilde{t} \exp \left\{- \tilde{t} Q ^ {j\top} \varrand ^ {-1} \left(\mn ^ j z - r_n \bar b\right) \right\} d\tilde{t} \\
    &=\frac{1}{\left[Q ^ {j\top} \varrand ^ {-1} \left(\mn ^ j z - r_n \bar b\right)\right] ^ 2}.
\end{aligned}
\]
As a result, our assertion follows for $m = 1$ whenever $z\in \cD_n$. A similar strategy is applied to obtain the conclusions for  $\big|\partial_{i_1, i_2} ^ 2 \cA_n (z)\big|$ and $\big|\partial_{i_1, i_2, i_3} ^ 3 \cA_n (z)\big|$.

\medskip
Now, we consider the case when $Q ^ {j\top} \varrand ^ {-1} \bar b > 0$. 
We return to the display in \eqref{eq:decomp_W0} and observe that
\[
\begin{aligned}
	W_0 (z) = \operatorname{Exp} \left(\mn ^ j z + \tail + Q ^ j I_1 ^ j, \Theta\right) \times \bar \cA_n(z),
\end{aligned}
\]
where
\[
    \bar \cA_n(z)
    = \bar{C} \int_{I_1 ^ j} ^ {\infty} \exp \left\{- \frac{Q ^ {j\top} \varrand ^ {-1} Q ^ j}{2} \bigg[t + \frac{Q ^ {j\top} \varrand ^ {-1} (\mn ^ j z + \tail)}{Q ^ {j\top} \varrand ^ {-1} Q ^ j}\bigg] ^ 2\right\} dt,
\]
for a constant $\bar{C}$.

\indent The conclusions for $\big|\partial_{i_1} ^ 1 \bar \cA_n (z)\big|$, $\big|\partial_{i_1, i_2} ^ 2 \bar \cA_n (z)\big|$ and $\big|\partial_{i_1, i_2, i_3} ^ 3 \bar \cA_n (z)\big|$ now follow directly after we apply the Leibniz integral rule.
\end{proof}


\begin{prop}
\label{prop:denominator_expatation}
Suppose that $Q ^ {j\top} \varrand ^ {-1} \bar b < 0$.
It holds that
\[
\begin{aligned}
\EE_{\cN} \left[W_0 \left(\cZ\right)\right] &\propto\operatorname{Exp} \left(Q ^ j I_1 ^ j + \tail, (\varrand + \mn ^ j \mn ^ {j\top}) ^ {-1} - \frac{(\varrand + \mn ^ j \mn ^ {j\top}) ^ {-1} Q ^ j Q ^ {j\top} (\varrand + \mn ^ j \mn ^ {j\top}) ^ {-1} }{  Q ^ {j\top} (\varrand + \mn ^ j \mn ^ {j\top}) ^ {-1} Q ^ j}\right) \\
&\;\;\;\;\;\;\;\;\;\;\;\;\;\;\;\;\;\;\;\;\;\;\;\;\;\;\;\;\;\;\;\;\;\;\;\;\;\;\;\;\;\;\;\;\;\;\;\;\;\;\;\;\;\;\;\;\;\;\;\;\;\;\;\;\;\;\;\;\;\;\;\;\;\;\;\;\;\;\;\;\;\;\;\;\;\;\;\;\times \PP(I_1 ^ j \leq \rv \leq \infty),
\end{aligned}
\]
where $\rv$ is distributed as
\[
	\cN \left(- \left[Q ^ {j\top} (\varrand + \mn ^ j \mn ^ {j\top}) ^ {-1} Q ^ j\right] ^ {-1} {Q ^ {j\top} (\varrand + \mn ^ j \mn ^ {j\top}) ^ {-1} } \tail, \left[Q ^ {j\top} (\varrand + \mn ^ j \mn ^ {j\top}) ^ {-1} Q ^ j\right]^{-1}\right).
\]
\end{prop}
\begin{proof}
We change the order of integration to write
\[
\begin{aligned}
    \EE_{\cN} \left[W_0 \left(\cZ\right)\right]
    &= \EE_{\cN} \left[\int_{I_1 ^ j} ^ {\infty} \phi \big(Q ^ j t + \mn ^ j \cZ + \tail; 0_p, \varrand\big) dt\right] \\
    &= \int_{I_1 ^ j} ^ {\infty} \EE_{\cN} \Big[\phi \big(Q ^ j t + \mn ^ j \cZ + \tail; 0_p, \varrand\big)\Big] dt.  
\end{aligned}
\]

We simplify the expectation in the integrand as:
\begin{equation}
\label{eq:integrand}
\begin{aligned}
	&\EE_{\cN} \Big[\phi \big(Q ^ j t + \mn ^ j \cZ + \tail; 0_p, \varrand\big)\Big] \\
	&= \int_{\RR ^ p} 
		\phi \left(Q ^ j t + \mn ^ j \cZ + \tail; 0_p, \varrand\right) \phi(\cZ; 0_p, I_{p, p}) d\cZ \\
	&\propto \int_{\RR ^ p} 
		\exp \left\{- \frac{1}{2} \big(Q ^ j t + \mn ^ j \cZ + \tail\big) ^ \top \varrand ^ {-1} \big(Q ^ j t + \mn ^ j \cZ + \tail\big) - \frac{1}{2} \cZ ^ \top \cZ\right\} d\cZ \\
 	&\propto \exp \left\{- \frac{1}{2} \big(Q ^ j t + \tail\big) ^ \top \left[\varrand ^ {-1} - \varrand ^ {-1} \mn ^ j \big(I_{p, p} + \mn ^ {j\top} \varrand ^ {-1} \mn ^ j\big) ^ {-1} \mn ^ {j\top} \varrand ^ {-1}\right] \big(Q ^ j t + \tail\big)\right\} \\
	&\quad\quad\times 
		\int_{\RR ^ p} \exp \bigg\{- \frac{1}{2} \big(\cZ - \xi_n\big) ^ \top \big(I_{p, p} + \mn ^ {j\top} \varrand ^ {-1} \mn ^ j\big) \big(\cZ - \xi_n\big)\bigg\} d\cZ
\end{aligned}
\end{equation}
where 
\[
    \xi_n = - \big(I_{p, p} + \mn ^ {j\top} \varrand ^ {-1} \mn ^ j\big) ^ {-1} \mn ^ {j\top} \varrand ^ {-1} \big(Q ^ j t + \tail\big).
\]
Using the Woodbury matrix identity, we have
\[
    \varrand ^ {-1} - \varrand ^ {-1} \mn ^ j \big(I_{p, p} + \mn ^ {j\top} \varrand ^ {-1} \mn ^ j\big) ^ {-1} \mn ^ {j\top} \varrand ^ {-1}
    = \big(\varrand + \mn ^ j \mn ^ {j\top}\big) ^ {-1},
\]
which implies that
$$
    \EE_{\cN} \Big[\phi \big(Q ^ j t + \mn ^ j \cZ + \tail; 0_p, \varrand\big)\Big]
    \propto \exp \left\{- \frac{1}{2} \big(Q ^ j t + \tail\big) ^ \top \big(\varrand + \mn ^ j \mn ^ {j\top}\big) ^ {-1} \big(Q ^ j t + \tail\big)\right\}.
$$

Plugging the simplified integrand into our integral, we note that
\[
\begin{aligned}
	\EE_{\cN} \left[W_0\left(\cZ\right)\right] 
	&\propto  \int_{I_1 ^ j} ^ {\infty}
		\exp \Big\{- \frac{1}{2} t ^ 2 Q ^ {j\top} (\varrand + \mn ^ j \mn ^ {j\top}) ^ {-1} Q ^ j - t Q ^ {j\top} (\varrand + \mn ^ j \mn ^ {j\top}) ^ {-1} \tail\\
	&\;\;\;\;\;\;\;\;\;\;\;\;\;\; \;\;\;\;\;\;\; \;\;\;\;\;\;\; \;\;\;\;\;\;\; \;\;\;\;\;\;\; \;\;\;\;\;\;\; \;\;\;\;\;\;\; \;\;\;\;\;\;\; \;\;\;\;\;\;\; - \frac{1}{2} \tail ^ \top (\varrand + \mn ^ j \mn ^ {j\top}) ^ {-1} \tail\Big\} dt \\
	&\propto \int_{I_1 ^ j} ^ {\infty} 
		\exp \left\{- \frac{1}{2} Q ^ {j\top} (\varrand + \mn ^ j \mn ^ {j\top}) ^ {-1} Q ^ j \bigg(t + \frac{Q ^ {j\top} (\varrand + \mn ^ j \mn ^ {j\top}) ^ {-1}}{Q ^ {j\top} (\varrand + \mn ^ j \mn ^ {j\top}) ^ {-1} Q ^ j}\tail\bigg) ^ 2\right\} \\
	& \times \exp \left\{- \frac{1}{2} \tail ^ \top \bigg[(\varrand + \mn ^ j \mn ^ {j\top}) ^ {-1} - \frac{(\varrand + \mn ^ j \mn ^ {j\top}) ^ {-1} Q ^ j Q ^ {j\top} (\varrand + \mn ^ j \mn ^ {j\top}) ^ {-1} }{Q ^ {j\top} (\varrand + \mn ^ j \mn ^ {j\top}) ^ {-1} Q ^ j}\bigg] \tail \right\} dt.
\end{aligned}
\]

Noting that
\[
\begin{aligned}
    &\operatorname{Exp} \left(\tail, (\varrand + \mn ^ j \mn ^ {j\top}) ^ {-1} - \frac{(\varrand + \mn ^ j \mn ^ {j\top}) ^ {-1} Q ^ j Q ^ {j\top} (\varrand + \mn ^ j \mn ^ {j\top}) ^ {-1} }{Q ^ {j\top} (\varrand + \mn ^ j \mn ^ {j\top}) ^ {-1} Q ^ j}\right) \\
    &= \operatorname{Exp} \left(Q ^ j I_1 ^ j + \tail, (\varrand + \mn ^ j \mn ^ {j\top}) ^ {-1} - \frac{(\varrand + \mn ^ j \mn ^ {j\top}) ^ {-1} Q ^ j Q ^ {j\top} (\varrand + \mn ^ j \mn ^ {j\top}) ^ {-1} }{Q ^ {j\top} (\varrand + \mn ^ j \mn ^ {j\top}) ^ {-1} Q ^ j}\right),
\end{aligned}
\]
we conclude that $\EE_{\cN} \left[W_0 \left(\cZ\right)\right]$ is proportional to 
\[
    \operatorname{Exp} \left(Q ^ j I_1 ^ j + \tail, (\varrand + \mn ^ j \mn ^ {j\top}) ^ {-1} - \frac{(\varrand + \mn ^ j \mn ^ {j\top}) ^ {-1} Q ^ j Q ^ {j\top} (\varrand + \mn ^ j \mn ^ {j\top}) ^ {-1} }{Q ^ {j\top} (\varrand + \mn ^ j \mn ^ {j\top}) ^ {-1} Q ^ j}\right) 
    \times \PP(I_1 ^ j \leq \rv \leq \infty).
\]
\end{proof}


\begin{prop}
\label{prop:denominator}
For sufficiently large $n$, it holds that:
\begin{flalign*}
&\begin{aligned}
    \text{(i)}&\;\; \EE_{\cN} \left[W_0 \left(\cZ\right)\right] 
	\gtrsim \operatorname{Exp} \left(- r_n \bar b, \left[\Theta ^ {-1} + \mn ^ j \mn ^ {j \top}\right] ^ {-1}\right) \text{ for } \; Q ^ {j\top} \varrand ^ {-1} \bar b < 0;\\
    \text{(ii)}&\;\; \EE_{\cN} \left[W_0 \left(\cZ\right)\right] 
	\gtrsim r_n ^ {-1}
	\operatorname{Exp} \left(- r_n \bar b, \left[\varrand + \mn ^ j \mn ^ {j\top}\right]^{-1}\right)) \text{ for } \; Q ^ {j\top} \varrand ^ {-1} \bar b > 0.
   \end{aligned}&&
\end{flalign*}
\end{prop}
\begin{proof}
We begin with the proof of the assertion in (i). 
We have that
\[
    W_0 \left(\cZ\right)
    \propto \operatorname{Exp} \left(\mn ^ j \cZ - r_n \bar b, \Theta\right) \times \int_{0} ^ {\infty} \exp \left\{- \frac{Q ^ {j\top} \varrand ^ {-1} Q ^ j}{2} \bigg[t + \frac{Q ^ {j\top} \varrand ^ {-1} (\mn ^ j \cZ - r_n \bar b)}{Q ^ {j\top} \varrand ^ {-1} Q ^ j}\bigg] ^ 2\right\} dt.
\]
Define the set $\cS_n = \left\{\cZ: Q ^ {j\top} \varrand ^ {-1} \left(\mn ^ j \cZ - r_n \bar b\right) < 0\right\}$.
Since
\[
     \int_{0} ^ {\infty} \exp \left\{- \frac{Q ^ {j\top} \varrand ^ {-1} Q ^ j}{2} \bigg[t + \frac{Q ^ {j\top} \varrand ^ {-1} (\mn ^ j \cZ - r_n \bar b)}{Q ^ {j\top} \varrand ^ {-1} Q ^ j}\bigg] ^ 2\right\} dt
    \gtrsim \frac{1}{2} 
    \quad\text{for } \cZ \in \cS_n,
\]
it holds that
\[
    W_0 \left(\cZ\right)
    \gtrsim \operatorname{Exp} \left(\mn ^ j \cZ - r_n \bar b, \Theta\right) \times \boldone_{\cZ \in \cS_n}.
\]
Then, we conclude that 
\[
\begin{aligned}
    \EE_\cN \left[W_0 \left(\cZ\right)\right] 
    &\gtrsim \EE_\cN \left[\operatorname{Exp} \left(\mn ^ j \cZ - r_n \bar b, \Theta\right) \boldone_{\cZ \in \cS_n}\right] \\
    &\gtrsim \EE_\cN \left[\operatorname{Exp} \left(\mn ^ j \cZ - r_n \bar b, \Theta\right)\right] 
    - \EE_\cN \left[\operatorname{Exp} \left(\mn ^ j \cZ - r_n \bar b, \Theta\right) \boldone_{\cZ \in \cS_n ^ c}\right] \\
    &\gtrsim \frac{1}{2} \EE_\cN \left[\operatorname{Exp} \left(\mn ^ j \cZ - r_n \bar b, \Theta\right)\right].
\end{aligned}
\]
Integrating with respect to $\cZ$ and applying the Woodbury matrix identity, we note that
\[
    \EE_\cN \left[\operatorname{Exp} \left(\mn ^ j \cZ - r_n \bar b, \Theta\right)\right]
    = \operatorname{Exp} \left(- r_n \bar b, \left[\Theta ^ {-1} + \mn ^ j \mn ^ {j \top}\right] ^ {-1}\right).
\]
Our conclusion thus follows.

For the assertion in (ii), note that
\begin{equation}
\label{eq:T_prob}
\begin{aligned}
    \PP(I_1 ^ j \leq \rv \leq \infty)
    &\propto \int_{I_1 ^ j} ^ {\infty} 
		\exp \left\{- \frac{1}{2} Q ^ {j\top} (\varrand + \mn ^ j \mn ^ {j\top}) ^ {-1} Q ^ j \bigg(t - I_1 ^ j + \frac{Q ^ {j\top} (\varrand + \mn ^ j \mn ^ {j\top}) ^ {-1} (Q ^ j I_1 ^ j + \tail)}{Q ^ {j\top} (\varrand + \mn ^ j \mn ^ {j\top}) ^ {-1} Q ^ j}\bigg) ^ 2\right\} dt \\
    &\propto \int_0 ^ {\infty} 
		\exp \left\{- \frac{1}{2} Q ^ {j\top} (\varrand + \mn ^ j \mn ^ {j\top}) ^ {-1} Q ^ j \bigg(\tilde t + \frac{Q ^ {j\top} (\varrand + \mn ^ j \mn ^ {j\top}) ^ {-1} (Q ^ j I_1 ^ j + \tail)}{Q ^ {j\top} (\varrand + \mn ^ j \mn ^ {j\top}) ^ {-1} Q ^ j}\bigg) ^ 2\right\} d \tilde t,
\end{aligned}
\end{equation}
through a change of variable $\tilde t = t - I_1 ^ j$.

When $Q ^ {j\top} (\varrand + \mn ^ j \mn ^ {j\top}) ^ {-1} (Q ^ j I_1 ^ j + \tail) > 0$, we apply the Mill's ratio bound to note that
\[
\begin{aligned}
	\PP(I_1 ^ j \leq \rv \leq \infty)
	&\ge\! \left(\frac{Q ^ {j\top} (\varrand + \mn ^ j \mn ^ {j\top}) ^ {-1} (Q ^ j I_1 ^ j + \tail)}{Q ^ {j\top} (\varrand + \mn ^ j \mn ^ {j\top}) ^ {-1} Q ^ j}\right) ^ {\!\!-1}
    \!\!\!\!\!\times\!\! \left[1 - \left(\frac{Q ^ {j\top} (\varrand + \mn ^ j \mn ^ {j\top}) ^ {-1} (Q ^ j I_1 ^ j + \tail)}{Q ^ {j\top} (\varrand + \mn ^ j \mn ^ {j\top}) ^ {-1} Q ^ j}\right) ^ {\!\!-2}\right] \\
    &\quad\quad \times \operatorname{Exp} \left(Q ^ j I_1 ^ j + \tail, \frac{(\varrand + \mn ^ j \mn ^ {j\top}) ^ {-1} Q ^ j Q ^ {j\top} (\varrand + \mn ^ j \mn ^ {j\top}) ^ {-1} }{Q ^ {j\top} (\varrand + \mn ^ j \mn ^ {j\top}) ^ {-1} Q ^ j}\right).
\end{aligned}
\]	
This yields:
\[
\begin{aligned}
	\EE_\cN \left[W_0 \left(\cZ\right)\right] 
	&\ge \left(\frac{Q ^ {j\top} (\varrand + \mn ^ j \mn ^ {j\top}) ^ {-1} (Q ^ j I_1 ^ j + \tail)}{Q ^ {j\top} (\varrand + \mn ^ j \mn ^ {j\top}) ^ {-1} Q ^ j}\right) ^ {-1}
    \!\!\!\! \times\!\! \left[1 - \left(\frac{Q ^ {j\top} (\varrand + \mn ^ j \mn ^ {j\top}) ^ {-1} (Q ^ j I_1 ^ j + \tail)}{Q ^ {j\top} (\varrand + \mn ^ j \mn ^ {j\top}) ^ {-1} Q ^ j}\right) ^ {\!-2}\right] \\
	&\quad\quad\times \operatorname{Exp} \left(Q ^ j I_1 ^ j + \tail, \frac{(\varrand + \mn ^ j \mn ^ {j\top}) ^ {-1} Q ^ j Q ^ {j\top} (\varrand + \mn ^ j \mn ^ {j\top}) ^ {-1} }{Q ^ {j\top} (\varrand + \mn ^ j \mn ^ {j\top}) ^ {-1} Q ^ j}\right) \\
	&\quad\quad\times \operatorname{Exp} \left(Q ^ j I_1 ^ j + \tail, (\varrand + \mn ^ j \mn ^ {j\top}) ^ {-1} - \frac{(\varrand + \mn ^ j \mn ^ {j\top}) ^ {-1} Q ^ j Q ^ {j\top} (\varrand + \mn ^ j \mn ^ {j\top}) ^ {-1} }{Q ^ {j\top} (\varrand + \mn ^ j \mn ^ {j\top}) ^ {-1} Q ^ j}\right) \\
	&= \left(\frac{Q ^ {j\top} (\varrand + \mn ^ j \mn ^ {j\top}) ^ {-1} (Q ^ j I_1 ^ j + \tail)}{Q ^ {j\top} (\varrand + \mn ^ j \mn ^ {j\top}) ^ {-1} Q ^ j}\right) ^ {-1}
    \!\!\!\! \times\!\! \left[1 - \left(\frac{Q ^ {j\top} (\varrand + \mn ^ j \mn ^ {j\top}) ^ {-1} (Q ^ j I_1 ^ j + \tail)}{Q ^ {j\top} (\varrand + \mn ^ j \mn ^ {j\top}) ^ {-1} Q ^ j}\right) ^ {\!-2}\right] \\
    &\quad\quad\times \operatorname{Exp} \left(Q ^ j I_1 ^ j + \tail, \left[\varrand + \mn ^ j \mn ^ {j\top}\right]^{-1}\right).
\end{aligned}
\]
Using the parameterization $Q ^ j I_1 ^ j + \tail = - r_n \bar b$, we conclude that 
\[
	\EE_\cN \left[W_0 \left(\cZ\right)\right] 
	\gtrsim r_n ^ {-1} \operatorname{Exp} \left(- r_n \bar b, \left[\varrand + \mn ^ j \mn ^ {j\top}\right]^{-1}\right).
\]

When $Q ^ {j\top} (\varrand + \mn ^ j \mn ^ {j\top}) ^ {-1} (Q ^ j I_1 ^ j + \tail) < 0$, we note that
\[
    \PP(I_1 ^ j \leq \rv \leq \infty) \ge \frac{1}{2}
\]
in \eqref{eq:T_prob}.
Therefore, for sufficiently large $n$, we have
\[
\begin{aligned}
	\EE_\cN \left[W_0 \left(\cZ\right)\right] 
	&\gtrsim \frac{1}{2} \operatorname{Exp} \left(- r_n \bar b, (\varrand + \mn ^ j \mn ^ {j\top}) ^ {-1} - \frac{(\varrand + \mn ^ j \mn ^ {j\top}) ^ {-1} Q ^ j Q ^ {j\top} (\varrand + \mn ^ j \mn ^ {j\top}) ^ {-1} }{Q ^ {j\top} (\varrand + \mn ^ j \mn ^ {j\top}) ^ {-1} Q ^ j}\right) \\
	&\gtrsim \operatorname{Exp} \left(- r_n \bar b, \left(\varrand + \mn ^ j \mn ^ {j\top}\right)^{-1}\right).
\end{aligned}
\]
\end{proof}

\begin{lem}
\label{lem:LDP}
Suppose that
\[
    \Xi_t(z) = \begin{cases}
        \frac{1}{2} (\sqrt{t} \mn ^ j z - \bar b) ^ \top [\varrand + \mn ^ j \mn ^ {j\top} (1 - t)] ^ {-1} (\sqrt{t} \mn ^ j z - \bar b), & \;\text{ if } t \in (0, 1], \\
        0, &\; \text{ if }  t = 0.
    \end{cases}
\]
For sufficiently large $n$ and $\mathbb{F}_n \in \cF_n$ under Assumptions \ref{aspt:moment_bound} and \ref{aspt:error}, we have that
\[
	r_n ^ {-2} \log \EE_{\FF_n} \left[\exp \left(-r_n ^ 2 \Xi_t \left(\frac{\Upsilon_n}{r_n}\right)\right) \boldone_{\cR_t}(r_n^{-1} \Upsilon_n)\right]
	\leq - \inf_{z \in \cR_t} \left(\frac{1}{2} z ^ \top z + \Xi_t(z)\right),
\]
where $\cR_t = [-c_1 \cdot 1_p, c_1 \cdot 1_p]$ for $c_1 > 0$ and $t \in (0, 1]$, and $\cR_0$ is the complement of $[-c_1 \cdot 1_p, c_1 \cdot 1_p]$.
\end{lem}

\begin{proof}
Due to the assertion (i) in Proposition \ref{lem:Zn_moment_bound}, we have that $\Upsilon_n$ has a finite moment generating function in the neighborhood of the origin. 
It ensures that $\Upsilon_n$ satisfies a large deviation principle with rate function $I(z) = \|z\|^2/2$.
An application of Varadhan's large deviation lemma yields that
\[
\begin{aligned}
    r_n ^ {-2} \log \EE_{\FF_n} \left[\exp \left(-r_n ^ 2 \Xi_t \left(\frac{\Upsilon_n}{r_n}\right)\right) \boldone_{\cR_t}(r_n^{-1} \Upsilon_n)\right] 
    &\leq \sup_{z \in \cR_t} \left(- \frac{\|z\| ^ 2}{2} - \Xi_t (z)\right)\\
    &= - \inf_{z \in \cR_t} \left(\frac{1}{2} z ^ \top z + \Xi_t(z)\right)
\end{aligned}
\]
for sufficiently large $n$ and $\mathbb{F}_n \in \cF_n$.
\end{proof}

\begin{prop}
\label{lem:use_bound}
Let $\W$ be as defined in Proposition \ref{prop:stein_bound}.
Under Assumptions \ref{aspt:moment_bound} and \ref{aspt:error}, we have
\[
\begin{aligned} 
	\sup_n\sup_{\FF_n \in \cF_n} \dfrac{\EE_{\FF_n} \bigg[\|e_{1, n}\| ^ \lambda \|e_{1, n} ^ *\| ^ \gamma \sup_{\alpha, \kappa \in[0,1]} \bigg|\bigintsss_0 ^ 1 \sqrt{t} \operatorname{Exp} \left(\mn ^ j \sqrt{t} \W - r_n \bar b, \left[\varrand + \mn ^ j \mn ^ {j\top} (1 - t)\right] ^ {-1} \right)dt\bigg|\bigg]}
	{\operatorname{Exp} \left(- r_n \bar b, \left[\varrand + \mn ^ j \mn ^ {j\top}\right] ^ {-1}\right)}
	< \infty,
\end{aligned}
\]	
for $\lambda, \gamma \in \NN$ such that $\lambda+\gamma \leq 3$.  
\end{prop}

\begin{proof}
Denote by $\varrho_{\max}$ the largest eigenvalue of $(\varrand + \mn ^ j \mn ^ {j\top})$. 
Let
\[
    \Pi(t) = \Big\{t \mn ^ {j\top} \big[\varrand + \mn ^ j \mn ^ {j\top} (1 - t)\big] ^ {-1} \mn + I_{p, p} \Big\} ^ {-1} \mn ^ {j\top} \big[\varrand + \mn ^ j \mn ^ {j\top} (1 - t)\big] ^ {-1},
\]
and let $\Pi_k(t)$ be the $k ^ {\text{th}}$ row of $\Pi(t)$, and let $\|\Pi_k(t)\|_{\max} = \max_{k \in [p]} \left\|\Pi_k(t)\right\|_2$. 
Fix
\[
    c_1 > \max \bigg(\varrho_{\max} ^ {1 / 2} \times (\|\bar b\|+1), \sup _{t \in[0,1]}\|\Pi_k(t)\|_{\max } \times (\|\bar b\|+1)\bigg).
\]

Observe that
\[
\begin{aligned}
	&\EE_{\FF_n} \bigg[\|e_{1, n}\| ^ \lambda \|e_{1, n} ^ *\| ^ \gamma \sup_{\alpha, \kappa \in[0,1]} \bigg|\int_0 ^ 1 \sqrt{t} \operatorname{Exp} \Big(\mn ^ j \sqrt{t} \W - r_n \bar b, \left[\varrand + \mn ^ j \mn ^ {j\top} (1 - t)\right] ^ {-1}\Big) dt\bigg|\bigg] \\
	&\leq \EE_{\FF_n} \bigg[\|e_{1, n}\| ^ \lambda \|e_{1, n} ^ *\| ^ \gamma \sup_{\alpha, \kappa \in[0,1]} \bigg|\int_0 ^ 1 \sqrt{t} \operatorname{Exp} \Big(\mn ^ j \sqrt{t} \W - r_n \bar b, \left[\varrand + \mn ^ j \mn ^ {j\top} (1 - t)\right] ^ {-1}\Big) \boldone_{\cR_t}(r_n ^ {-1} \Upsilon_n[-1]) dt\bigg|\bigg] \\
	&\quad\quad + \EE_{\FF_n} \Big[\|e_{1, n}\| ^ \lambda \|e_{1, n} ^ *\| ^ \gamma \boldone_{\cR_0}(r_n ^ {-1} \Upsilon_n[-1]) \Big],
\end{aligned}
\]
where $\cR_t = \left[-c_1 \cdot 1_p, c_1 \cdot 1_p\right] \subseteq \RR ^ p$, for $t \in (0,1]$, and $\cR_0 = \cR_1 ^ c$.

Note that we can further write:
\[
\begin{aligned}
	&\EE_{\FF_n} \bigg[\|e_{1, n}\| ^ \lambda \|e_{1, n} ^ *\| ^ \gamma \sup_{\alpha, \kappa \in[0,1]} \bigg|\int_0 ^ 1 \sqrt{t} \operatorname{Exp} \Big(\mn ^ j \sqrt{t} \W - r_n \bar b, \left[\varrand + \mn ^ j \mn ^ {j\top} (1 - t)\right] ^ {-1}\Big) dt\bigg|\bigg] \\	
	&\leq \EE_{\FF_n} \Big[\|e_{1, n}\| ^ \lambda \|e_{1, n} ^ *\| ^ \gamma \exp(\chi \|e_{1, n}\|)\Big]\\
	&\quad\quad\quad\quad\quad\quad\quad \times \int_0 ^ 1 \sqrt{t} \EE_{\FF_n} \bigg[\operatorname{Exp} \left(\sqrt{t} \mn ^ j \Upsilon_n[-1] - r_n \bar b, \left[\varrand + \mn ^ j \mn ^ {j\top} (1 - t)\right] ^ {-1}\right) \boldone_{\cR_t}(r_n ^ {-1} \Upsilon_n[-1]) \bigg] dt  \\
	&\quad\quad + \EE_{\FF_n} \Big[\|e_{1, n}\| ^ \lambda \|e_{1, n} ^ *\| ^ \gamma\Big] \EE_{\FF_n} \Big[\boldone_{\cR_0}(r_n ^ {-1} \Upsilon_n[-1])\Big],
\end{aligned}
\]
for some positive constant $\chi$.

Define
\[
	\Xi_t(z) := \frac{1}{2} (\sqrt{t} \mn ^ j z - \bar b) ^ \top [\varrand + \mn ^ j \mn ^ {j\top} (1 - t)] ^ {-1} (\sqrt{t} \mn ^ j z - \bar b).
\]
Then, we have that
\begin{equation}
\begin{aligned}
	&\dfrac{\bigintsss_0 ^ 1  \sqrt{t} \EE_{\FF_n} \left[\operatorname{Exp} \left(\sqrt{t} \mn \Upsilon_n[-1] - r_n \bar \beta, \left[\varrand + \mn ^ j \mn ^ {j\top} (1 - t)\right] ^ {-1}\right) \boldone_{\cR_t}(r_n ^ {-1} \Upsilon_n[-1]) \right] dt}
    {{\operatorname{Exp} \left(- r_n \bar b, \left[\varrand + \mn ^ j \mn ^ {j\top}\right] ^ {-1}\right)}} \\
	&\quad\quad\quad\quad \leq \sup_n \sup_{\FF_n \in \cF_n} \sup_{t \in(0,1]}  
    \dfrac{\EE_{\FF_n} \Big[\exp \Big(- r_n^2 \Xi_t \big(r_n ^ {-1} \Upsilon_n[-1]\big)\Big) \boldone_{[-c_1 \cdot 1_p, c_1 \cdot 1_p]}(r_n ^ {-1} \Upsilon_n[-1])\Big] \bigintsss_0 ^ 1 \sqrt{t} dt}
    {{\operatorname{Exp} \left(- r_n \bar b, \left[\varrand + \mn ^ j \mn ^ {j\top}\right] ^ {-1}\right)}}
\end{aligned}
\label{bound:VarLDP}
\end{equation}
By applying Lemma \ref{lem:LDP}, we have 
\[
\begin{aligned}
	\sup_n \sup_{\FF_n \in \cF_n} \sup_{t \in(0,1]} 
	\dfrac{\EE_{\FF_n} \Big[\exp \Big\{- r_n ^ 2 \Xi_t \big(r_n ^ {-1} \Upsilon_n[-1]\big)\Big\} \boldone_{[-c_1 \cdot 1_p, c_1 \cdot 1_p]}(r_n ^ {-1} \Upsilon_n[-1])\Big]}{\operatorname{Exp} \left(- r_n \bar b, \left[\varrand + \mn ^ j \mn ^ {j\top}\right] ^ {-1}\right)} < \infty.
\end{aligned}
\]
Thus, we conclude that the display on the right-hand side of \eqref{bound:VarLDP} is bounded by a constant.

Using Lemma \ref{lem:LDP} once again, we have
\[
    \dfrac{\EE_{\FF_n} \left[\boldone_{\cR_0}(r_n ^ {-1} \Upsilon_n[-1])\right]}
    {\operatorname{Exp} \left(- r_n \bar b, \left[\varrand + \mn ^ j \mn ^ {j\top}\right] ^ {-1}\right)} 
    \leq \sup_n \sup_{\FF_n \in \cF_n} 
	\dfrac{\EE_{\FF_n} \left[\boldone_{\cR_0}(r_n ^ {-1} \Upsilon_n[-1])\right]}
    {\operatorname{Exp} \left(- r_n \bar b, \left[\varrand + \mn ^ j \mn ^ {j\top}\right] ^ {-1}\right)}
	< \infty,
\]
which follows from our choice of $c_1$.

Our proof is complete as
\[
	\sup_n \sup _{\FF_n \in \cF_n} \EE_{\FF_n} \Big[\|e_{1, n}\| ^ \lambda \exp \left(\chi\|e_{1, n}\|\right)\Big]<\infty, 
	\quad
	\sup _n \sup _{\FF_n \in \cF_n} \EE_{\FF_n} \Big[\left\|e_{1, n}^*\right\|^\gamma\Big]<\infty,
\]
based on the moment bounds in Proposition \ref{lem:an_moment_bound}.
\end{proof}


\begin{thm}
Suppose that the conditions in Assumptions  \ref{aspt:moment_bound}, \ref{aspt:randomization} and \ref{aspt:error} are met with $r_n\to \infty$ and $r_n= o(n ^ {1 / 6})$.
Then, it holds that
$$\lim_{n} \sup_{\FF_n \in \cF_n} \RD ^ {(1)} = 0, \quad \lim_{n} \sup_{\FF_n \in \cF_n} \RD ^ {(2)}  = 0.$$
\end{thm}

\begin{proof}
We focus on the case when $Q ^ {j\top} \varrand ^ {-1} \bar b > 0$. The proof strategy applies similarly to the other case when $Q ^ {j\top} \varrand ^ {-1} \bar b < 0$.

First, we prove that
\begin{equation}
\begin{aligned}
\text{(i)}&\;\; \sup_{\FF_n \in \cF_n} \left|\EE_{\FF_n} \left[W_0 \left(\Upsilon_n\right) \boldone_{\cD_n}(\Upsilon_n) \right] - \EE_{\cN} \left[W_0 \left(\cZ\right) \boldone_{\cD_n}(\cZ)\right]\right| \le \frac{r_n ^ 2}{\sqrt{n}} \operatorname{Exp} \Big(- r_n \bar b, \big[\varrand + \mn ^ j \mn ^ {j\top}\big] ^ {-1}\Big); \\[7pt]
\text{(ii)}&\;\; \sup_{\FF_n \in \cF_n} \left|\EE_{\FF_n} \left[\arb \circ \pivotup\times W_0 \left(\Upsilon_n\right) \boldone_{\cD_n}(\Upsilon_n)\right] - \EE_{\cN} \left[\arb \circ \pivotz \times W_0 \left(\cZ\right) \boldone_{\cD_n}(\cZ)\right]\right| \\
&\quad\quad\quad\quad\quad\quad\quad\quad\quad\quad\quad\quad\quad\quad\quad\quad\quad\quad\quad\quad \le \frac{r_n ^ 2}{\sqrt{n}} \operatorname{Exp} \Big(- r_n \bar b, \big[\varrand + \mn ^ j \mn ^ {j\top}\big] ^ {-1}\Big).
\end{aligned}
\label{first:part:rare}
\end{equation}

Consider the real-valued functions
\[
\begin{aligned}
& \bar G^{(1)} \left(z\right) 
    = \operatorname{Exp} \left(\mn ^ j z + \tail + Q ^ j I_1 ^ j, \varrand ^ {-1}\right) \times \cA_n (z) \boldone_{\cD_n}(z), \\
& \bar G^{(2)}\left(z\right) 
    = \arb \circ \cP ^ {j \cdot E} (z) \times \operatorname{Exp} \left(\mn ^ j z + \tail + Q ^ j I_1 ^ j, \varrand ^ {-1}\right) \times \cA_n (z) \boldone_{\cD_n}(z),
\end{aligned}
\]
where the function $\cA_n$ and the set $\cD_n$ are as defined in Proposition \ref{prop:An_upper_bound}.
On the set $\cD_n$, $\bar G^{(1)}(z)$ and $\bar G^{(2)}(z)$ are equal to $W_0(z)$ and $\arb \circ \cP ^ {j \cdot E} (z) \times W_0(z)$, respectively.

To establish \eqref{first:part:rare}, we use the Stein bound from Proposition \ref{prop:stein_bound} with $G= \bar G^{(l)}$ for $l\in \{1,2\}$, and use the definitions of $\W$, $\Upsilon_n[-1]$, $e_{1, n}$, and $e_{1, n}^*$ provided in this result.
Using the properties of our pivot, $W_0$ and $\cA_n$ as derived in Propositions \ref{prop:pivot_bound}, \ref{prop:W0}, and \ref{prop:An_upper_bound}, we note that
\[
\begin{aligned}
	\left|\partial_{i_1, i_2, i_3} ^ 3 \bar G ^ {(l)} \left(\sqrt{t} \W + \sqrt{1-t} z\right)\right| 
 	\lesssim &\sum_{l = 0} ^ 3 r_n ^ {-1} \big\|\mn ^ j \sqrt{t} \W + \mn ^ j \sqrt{1-t} z + \tail\big\| ^ l \\
    &\times \operatorname{Exp} \left(\mn ^ j \sqrt{t} \W + \mn ^ j \sqrt{1-t} z + \tail, \varrand ^ {-1}\right).
\end{aligned}
\]
After taking expectations, it holds that
\[
\begin{aligned}
    &\EE_{\cN} \left[\left|\partial_{i_1, i_2, i_3} ^ 3 \bar G ^ {(l)} \left(\sqrt{t} \W + \sqrt{1-t} \cZ\right)\right|\right] \\
    &\lesssim \sum_{\substack{\lambda, \gamma \in \NN: \\ \lambda+\gamma \leq 3}} r_n ^ {-1} \big\|\W\big\| ^ \lambda \big\|\tail\big\| ^ \gamma \operatorname{Exp} \left(\mn ^ j \sqrt{t} \W + \tail, \left[\varrand + \mn ^ j \mn ^ {j\top} (1 - t)\right] ^ {-1}\right) \\
    &\lesssim \sum_{\substack{\bar{\lambda}, \bar{\kappa}, \breve{\lambda}, \breve{\kappa} \in \NN: \\ \bar{\lambda} + \bar{\kappa}+\breve{\lambda}+\breve{\kappa} \leq 3}} \!\!\!\!
	r_n ^ {\breve{\lambda} - 1}
	\big\|\Upsilon_n[-1]\big\| ^ {\breve{\kappa}} 
	\left\|\frac{e_{1, n}}{\sqrt{n}}\right\| ^ {\bar{\lambda}}
	\left\|\frac{e_{1, n} ^ *}{\sqrt{n}}\right\| ^ {\bar{\kappa}}
	\operatorname{Exp} \left(\mn ^ j \sqrt{t} \W - r_n \bar b, \left[\varrand + \mn ^ j \mn ^ {j\top} (1 - t)\right] ^ {-1}\right). 
\end{aligned}
\]
By plugging this result into the bound derived in Proposition \ref{prop:stein_bound}, we conclude that
\begin{equation*}
\begin{aligned}
& \left|\EE_{\FF_n} \left[\bar G ^ {(l)} (\Upsilon_n)\right] - \EE_{\cN} \left[\bar G ^ {(l)}(\cZ)\right]\right|\\
&\lesssim \frac{1}{\sqrt{n}} \sum_{\substack{\lambda, \gamma \in \NN: \\ \lambda+\gamma \leq 3}} \sum_{i_1, i_2, i_3 \in [p]} 
	\EE_{\FF_n}
		\Bigg[\|e_{1, n}\| ^ \lambda \|e_{1, n} ^ *\| ^ \gamma 
		\sup_{\alpha, \kappa \in[0,1]} \int_0 ^ 1 \frac{\sqrt{t}}{2} \\
&	\times  \sum_{\substack{\bar{\lambda}, \bar{\kappa}, \breve{\lambda}, \breve{\kappa} \in \NN: \\ \bar{\lambda} + \bar{\kappa}+\breve{\lambda}+\breve{\kappa} \leq 3}} \!\!\!\!
	r_n ^ {\breve{\lambda} - 1}
	\big\|\Upsilon_n[-1]\big\| ^ {\breve{\kappa}} 
	\left\|\frac{e_{1, n}}{\sqrt{n}}\right\| ^ {\bar{\lambda}}
	\left\|\frac{e_{1, n} ^ *}{\sqrt{n}}\right\| ^ {\bar{\kappa}}
	\operatorname{Exp} \left(\mn ^ j \sqrt{t} \W - r_n \bar b, \left[\varrand + \mn ^ j \mn ^ {j\top} (1 - t)\right] ^ {-1}\right)	 dt\Bigg].
\end{aligned}
\end{equation*}

To complete our proof, we simplify the bound on the right-hand side when
$$\bar{\lambda}=\bar{\kappa}=\breve{\kappa}=0, \breve{\lambda}=3,$$
noting that the same approach can be used for different values of $\bar{\lambda}, \bar{\kappa}, \breve{\lambda}, \breve{\kappa}$.
Using Proposition \ref{lem:use_bound} in the last step gives us:
\[
   \sup_{\FF_n \in \cF_n}  \Big|\EE_{\FF_n} \big[\bar G ^ {(l)} \left(\Upsilon_n\right)\big] - \EE_{\cN} \big[\bar G ^ {(l)} \left(\cZ\right)\big]\Big| \\
	\lesssim \frac{r_n ^ 2}{\sqrt{n}}
    \operatorname{Exp} \left(- r_n \bar b, \left[\varrand + \mn ^ j \mn ^ {j\top}\right] ^ {-1}\right).
\]
This concludes our proof for \eqref{first:part:rare}.

It is immediate from Proposition \ref{prop:denominator} that 
\begin{equation}
\EE_{\cN} \left[W_0 \left(\cZ\right)\right] \ge r_n ^ {-1}
	\operatorname{Exp} \Big(- r_n \bar b, \big[\varrand + \mn ^ j \mn ^ {j\top}\big] ^ {-1}\Big).
	\label{second:part:rare}
\end{equation}	

In the definition of $\RD ^ {(1)}$ and $\RD ^ {(2)}$, we fix $\widetilde{D}_n=\cD_n$ as considered in Proposition \ref{prop:An_upper_bound}.
Using Proposition \ref{lem:Zn_moment_bound}, for this sequence of sets we have that
$$
\sup_{\FF_n\in \cF_n} \mathbb{P}_{\mathbb{F}_n}[ \Upsilon_n \in \widetilde{D}_n^c] \leq \sup_{\FF_n\in \cF_n} \mathbb{P}_{\mathbb{F}_n}\left[\left\|\frac{\Upsilon_n}{r_n}\right\| \geq c_0\right] \leq \exp \left(p \log 5 - \frac{r_n^2c_0^2}{16 \sigma ^ 2 \left\|H ^ {-1}\right\|}\right),
$$
thereby leading to 
$$
\lim_n\sup_{\FF_n\in \cF_n} \mathbb{P}_{\mathbb{F}_n}[ \Upsilon_n \in \widetilde{D}_n^c] =0.
$$
Our main claim is proven by combining the results in \eqref{first:part:rare} and \eqref{second:part:rare}.
\end{proof}

\section{Additional simulation experiments for Section \ref{sec:simulation}}
\label{appdx:simulation}

In this section, we provide additional simulation results for Section \ref{sec:simulation} with the bandwidths $h$ and $h'$ possibly different.

\subsection{Case I: $n = 800$ and $p = 200$.}
\label{appdx:simulation_previous}
We generate datasets with $n = 800$ i.i.d. observations and a set of $p = 200$ potential covariates from the following models.
\begin{enumerate}[label={Model \arabic*:}, leftmargin=*]
    \item $y_i = \beta_0 + x_i ^ \top \beta + \varepsilon_i$, where $\beta_0 = 0.2$, $\beta = (c, c, c, c, c, 0 \ldots, 0) ^ \top \in \RR ^ p$, $\varepsilon_i \sim \cN(0, 4)$ and $x_i \sim \cN(0, \Sigma)$ for an autoregressive design matrix with $\Sigma_{j, k} = 0.5 ^ {|j - k|}$, where $\varepsilon_i $ and $x_i$ are independent.
	\item $y_i = \beta_0 + x_{i, 1} ^ \top \beta + 1.5 x_{i, 2} \varepsilon_i$, where $\beta_0 = 0.2$, $\beta = (c, c, c, c, c, 0 \ldots, 0) ^ \top \in \RR ^ p$ and $\varepsilon_i \sim \cN(0, 4)$. In this model, $x = (x_{i, 2}, x_{i, 1} ^ \top) ^ \top \in \RR ^ p$, where the variable $x_{i, 2} \in \RR$ is drawn from $U(0 ,2)$ and the variables $x_{i, 1} \in \RR ^ {p - 1}$ are drawn independently from $x_{i, 2}$ and from $\cN(0, \Sigma)$ for an autoregressive design matrix with $\Sigma_{j, k} = 0.5 ^ {|j - k|}$.
	\item $y_i = \beta_0(u_i) + x_i ^ \top \beta(u_i)$ for $i=1, \ldots, n$, where $\beta_0(u_i) = 2 c u_i$ and $\beta(u_i) = (c u_i, c u_i, c, c, c, 0 \ldots, 0) ^ \top \in \RR ^ p$, $u_i \sim U(0, 1)$. 
	In this model, $x_{i, 1}, x_{i, 2} \in \RR$, the first and second elements of $x_i$, are drawn from $U(0, 2)$ and $x_{i, 3} \in \RR ^ {p - 2}$, the subvector of $x_i$ after removing $x_{i, 1}, x_{i, 2}$, is drawn from $\cN(0, \Sigma)$ for an autoregressive design matrix $\Sigma_{j, k} = 0.5 ^ {|j - k|}$.
\end{enumerate}

We draw selective inference in each of our models after estimating $F_{y|x} ^ {-1} (\tau)$, which is the $\tau$-th population conditional quantile of $y$ given $x$. 
Note that $F_{y|x} ^ {-1} (\tau) = \beta_0 (\tau) + x ^ \top \beta (\tau)$, where in 
\begin{equation*}
\begin{gathered}
\text{ Model 1: } \beta_0 (\tau) = 0.2 + \Phi ^ {-1} (\tau; 0, 4) \text{ and } \beta (\tau) = (c, c, c, c, c, 0 \ldots, 0) ^ \top, \\
\text{ Model 2: } \beta_0 (\tau) = 0.2 \text{ and } \beta (\tau) = (1.5 \Phi ^ {-1} (\tau; 0, 4), c, c, c, c, c, 0 \ldots, 0) ^ \top, \\ 
\text{ Model 3: } \beta_0 (\tau) = 2 c \tau \text{ and } \beta (\tau) = (c \tau, c \tau, c, c, c, 0 \ldots, 0) ^ \top. 
\end{gathered}
\end{equation*}

In our simulation, we set the quantile level at $\tau = 0.7$. 
The signal strength settings are categorized as ``Low'', ``Medium'',  and ``High'', depending on the value of $c$ from the set $\{0.1, 0.5, 1\}$.
We apply the randomized $\ell_1$-penalized SQR problem with tuning parameter $\lambda = 0.6 \sqrt{\log p / n}$.
To form the smoothed loss function, we use Gaussian kernels with bandwidths $h = \max \big\{0.05, \sqrt{\tau(1-\tau)} (\log (p) / n)^{1 / 4}\big\}$ and $h' = \{(q+\log n) / n\}^{2 / 5}$ for $q = |E|$ at the time of selection and inference, respectively.
Our proposed method is implemented with white noise $\omega_n$ drawn from $\cN(0_p, \Omega/n)$ with $\varrand = I_{p, p}$.
We report our findsing in Figure \ref{fig:coverage_old}, Figure \ref{fig:length_old} and Table \ref{table:F1scores_old}.
These results are based on $500$ independent Monte Carlo datasets for each pair of model and signal setting.

\begin{figure}[h]
	\centering
 	\includegraphics[scale = 0.62]{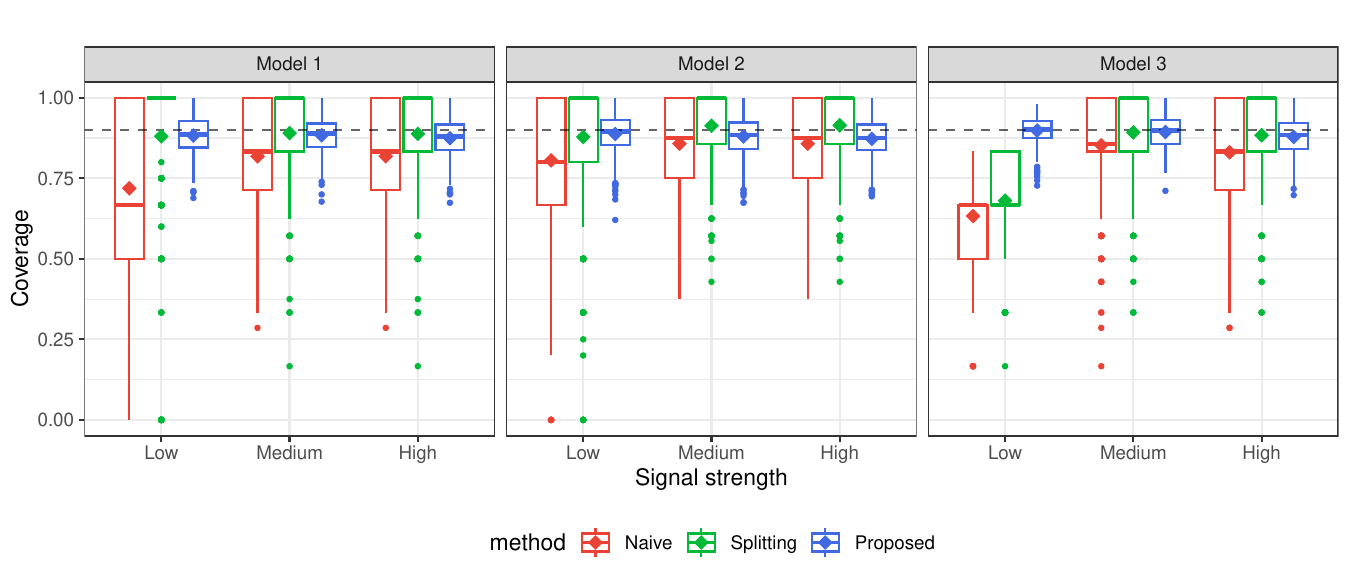}
	\caption{
    \small{Coverage rates of different methods across different models and signal settings. 
    The gray dashed line represents the prespecified target coverage rate at $0.9$, and the diamond marks highlight the averaged coverage rates over all replications.}}
	\label{fig:coverage_old}
\end{figure}

\begin{figure}[h]
	\centering
	\includegraphics[scale = 0.62]{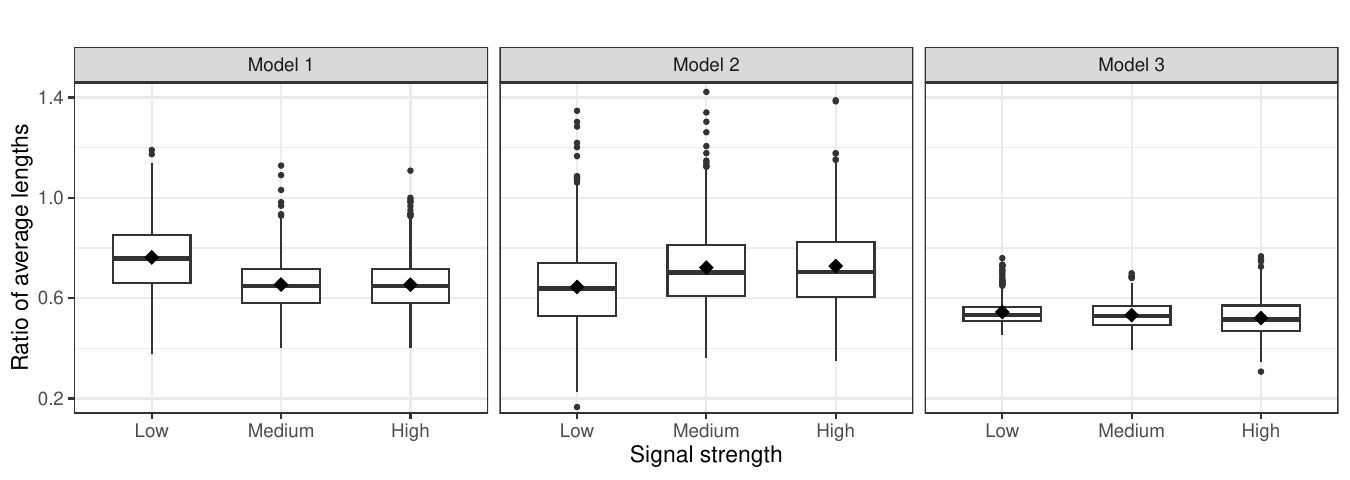}
	\caption{
    \small{The boxplots present the ratio of average interval lengths for the selected parameters between the ``Proposed'' method and the ``Splitting'' method across different models and signal strengths.}}
	\label{fig:length_old}
\end{figure}

\begin{table}[h]
\caption{\small{Accuracy based on F1 scores before and after applying ``Proposed'' in the three models and signal settings. }}
\vspace{10pt}
\footnotesize
\begin{tabular*}{\textwidth}{@{\extracolsep{\fill}}*{10}{c}}
\hline
\hline
    & & Model 1 & & & Model 2 & & & Model 3 & \\
     \cline{2 - 4} \cline{5 - 7} \cline{8 - 10} 
Signal Strength & Low & Medium & High & Low & Medium & High & Low & Medium & High \\ 
\hline
\hline
F1 score before inference
    & 0.10 & 0.19 & 0.22 & 0.12 & 0.19 & 0.23 & 0.22 & 0.21 & 0.21 \\
F1 score after inference
    & 0.27 & 0.63 & 0.70 & 0.30 & 0.59 & 0.67 & 0.74 & 0.70 & 0.67 \\
\hline
\hline
\end{tabular*}
\label{table:F1scores_old}
\end{table}

\subsection{Case II: $n = 800$ and $p = 1000$.}
\label{appdx:simulation_case1}
We generate datasets with $n$ i.i.d. observations and a set of $p$ potential covariates with $n = 800$ and $p = 1000$ from the following models.
\begin{enumerate}[label={Model \arabic*:}, leftmargin=*]
    \setcounter{enumi}{3}
    \item $y_i = \beta_0 + x_i ^ \top \beta + \varepsilon_i$, where $\beta_0 = 0.2$, $\beta = (c, c, c, c, c, 0 \ldots, 0) ^ \top \in \RR ^ p$, $\varepsilon_i \sim \cN(0, 1)$ and $x_i \sim \cN(0, \Sigma)$ for an autoregressive design matrix with $\Sigma_{j, k} = 0.5 ^ {|j - k|}$, where $\varepsilon_i $ and $x_i$ are independent.
	\item $y_i = \beta_0 + x_{i, 1} ^ {\top} \beta + 1.5 x_{i, 2} \varepsilon_i$, where $\beta_0 = 0.2$, $\beta = (c, c, c, c, c, 0 \ldots, 0) ^ \top \in \RR ^ p$ and $\varepsilon_i \sim \cN(0, 1)$. In this model, $x = (x_{i, 2}, x_{i, 1} ^ {\top}) ^ \top \in \RR ^ p$, where the variable $x_{i, 2} \in \RR$ is drawn from $U(0 ,2)$ and the variables $x_{i, 1} \in \RR ^ {p - 1}$ are drawn independently from $x_{i, 2}$ and from $\cN(0, \Sigma)$ for an autoregressive design matrix with $\Sigma_{j, k} = 0.5 ^ {|j - k|}$.
	\item $y_i = \beta_0(u_i) + x_i ^ \top \beta(u_i)$ for $i=1, \ldots, n$, where $\beta_0(u_i) = 2 c u_i$ and $\beta(u_i) = (c u_i, c u_i, c, c, c, 0 \ldots, 0) ^ \top \in \RR ^ p$, $u_i \sim U(0, 1)$. 
	In this model, $x_{i, 1}, x_{i, 2} \in \RR$, the first and second elements of $x_i$, are drawn from $U(0, 2)$ and $x_{i, 3} \in \RR ^ {p - 2}$, the subvector of $x_i$ after removing $x_{i, 1}, x_{i, 2}$, is drawn from $\cN(0, \Sigma)$ for an autoregressive design matrix $\Sigma_{j, k} = 0.5 ^ {|j - k|}$.
\end{enumerate}

We draw selective inferences in each of our models after estimating $F_{y|x} ^ {-1} (\tau)$, which is the $\tau$-th population conditional quantile of $y$ given $x$. 
Note that $F_{y|x} ^ {-1} (\tau) = \beta_0 (\tau) + x ^ \top \beta (\tau)$, where in 
\[
\begin{gathered}
\text{ Model 4: } \beta_0 (\tau) = 0.2 + \Phi ^ {-1} (\tau; 0, 1) \text{ and } \beta (\tau) = (c, c, c, c, c, 0 \ldots, 0) ^ \top, \\
\text{ Model 5: } \beta_0 (\tau) = 0.2 \text{ and } \beta (\tau) = (1.5 \Phi ^ {-1} (\tau; 0, 1), c, c, c, c, c, 0 \ldots, 0) ^ \top, \\ 
\text{ Model 6: } \beta_0 (\tau) = 2 c \tau \text{ and } \beta (\tau) = (c \tau, c \tau, c, c, c, 0 \ldots, 0) ^ \top.
\end{gathered}
\]

In our simulation, we set the quantile level at $\tau = 0.7$. The signal strength settings are categorized as ``Low'', ``Medium'',  and ``High'', depending on the value of $c$ from the set $\{0.1, 0.5, 1\}$.
We apply the randomized $\ell_1$-penalized SQR problem with tuning parameter $\lambda = 1.8 \sqrt{\log p / n}$.
To form the smoothed loss function, we use Gaussian kernels with bandwidths $h = \max \big\{0.05, \sqrt{\tau(1-\tau)} (\log (p) / n)^{1 / 4}\big\}$ and $h' = \{(q+\log n) / n\}^{2 / 5}$ for $q = |E|$ at the selection and inference stages, respectively.
Our proposed method is implemented with white noise $\omega$
drawn from $\cN(0_p, \Omega)$ with $\varrand = \frac{4}{n} I_{p, p}$.
We report our findsing in Figure \ref{fig:coverage_new}, Figure \ref{fig:length_new} and Table \ref{table:F1scores_new}.
These results are based on $500$ independent Monte Carlo datasets for each pair of model and signal setting.

\begin{figure}[h]
	\centering
	\includegraphics[scale = 0.62]{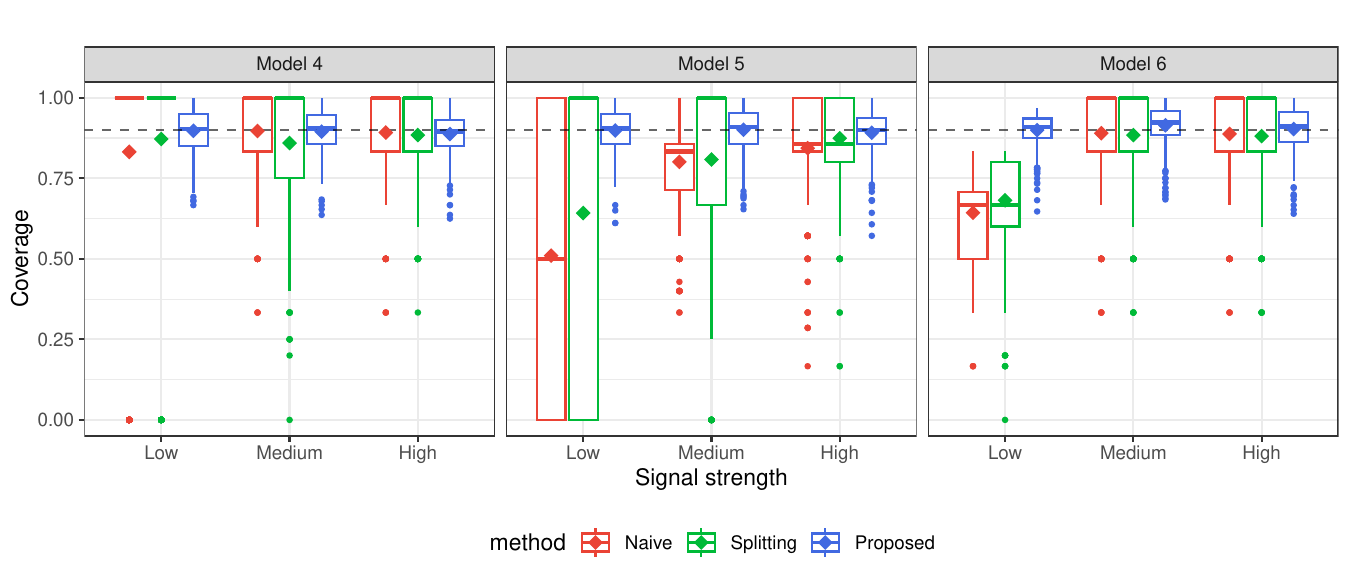}
	\caption{
    \small{Coverage rates of different methods across different models and signal settings. 
    The gray dashed line represents the prespecified target coverage rate at $0.9$, and the diamond marks highlight the averaged coverage rates over all replications. 
    }}
	\label{fig:coverage_new}
\end{figure}

\begin{figure}[h]
	\centering
	\includegraphics[scale = 0.62]{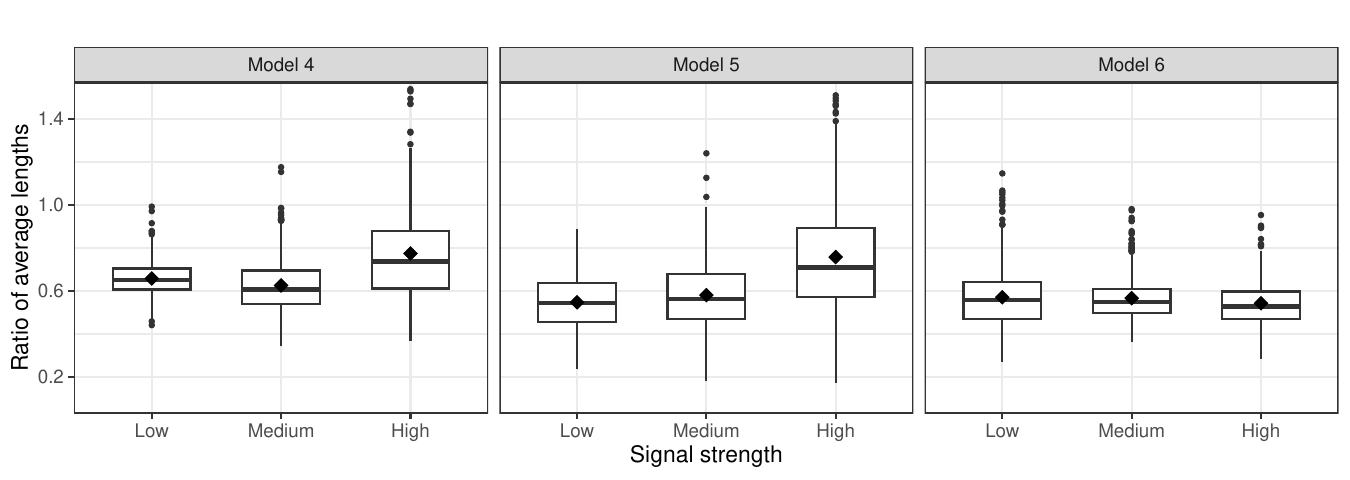}
	\caption{
    \small{The boxplots present the ratio of average interval lengths for the selected parameters between the ``Proposed'' method and the ``Splitting'' method across different models and signal strengths. 
    }}
	\label{fig:length_new}
\end{figure}
\begin{table}[h]
\caption{\small{Accuracy based on F1 scores before and after applying ``Proposed'' in the three models and signal settings.}}
\vspace{10pt}
\footnotesize
\begin{tabular*}{\textwidth}{@{\extracolsep{\fill}}*{10}{c}}
\hline
\hline
    & & Model 4 & & & Model 5 & & & Model 6 & \\
     \cline{2 - 4} \cline{5 - 7} \cline{8 - 10} 
Signal Strength & Low & Medium & High & Low & Medium & High & Low & Medium & High \\ 
\hline
\hline
F1 score before inference
    & 0.10 & 0.24 & 0.30 & 0.10 & 0.22 & 0.26 & 0.25 & 0.27 & 0.27 \\
F1 score after inference
    & 0.28 & 0.62 & 0.71 & 0.28 & 0.54 & 0.62 & 0.64 & 0.67 & 0.67 \\
\hline
\hline
\end{tabular*}
\label{table:F1scores_new}
\end{table}

\subsection{Case III: model misspecification.}
\label{appdx:simulation_case2}

We generate datasets with $n$ i.i.d. observations and a set of $p$ potential covariates with $n = 400$ and $p = 100$.
\begin{enumerate}[label={Model \arabic*:}, leftmargin=*]
    \setcounter{enumi}{6}
    \item $y_i ^ 3 = \beta_0 + x_i ^ \top \beta + \varepsilon_i$, where $\beta_0 = 0.2$, $\beta = (c, c, c, c, c, 0 \ldots, 0) ^ \top \in \RR ^ p$, $\varepsilon_i \sim \cN(0, 1)$ and $x_i \sim \cN(0, \Sigma)$ for an autoregressive design matrix with $\Sigma_{j, k} = 0.5 ^ {|j - k|}$, where $\varepsilon_i $ and $x_i$ are independent.
\end{enumerate}
We draw selective inferences in each of our models after estimating $F_{y|x} ^ {-1} (\tau) = \beta_0 (\tau) + x ^ \top \beta (\tau)$.
To compute an approximate of our target, we generate an independent dataset with $n ^ \prime = 100n$ and solve the following problem:
\[
\argmin_{b \in \RR ^ {|E|}} \frac{1}{n'} \sum_{i=1} ^ {n'} \int_{-\infty} ^ \infty \rho_\tau(u) K_{h'}(u + x_{i, E} ^ \top b - y_i) du,
\]
where $E$ is our selected variables.

In our simulation, we set the quantile level at $\tau = 0.7$. The signal strength settings are categorized as ``Low'', ``Medium'',  and ``High'', depending on the value of $c$ from the set $\{0.1, 0.5, 1\}$.
We apply the randomized $\ell_1$-penalized SQR problem with tuning parameter $\lambda = 0.6 \sqrt{\log p / n}$.
To form the smoothed loss function, we use Gaussian kernels with bandwidths $h = \max \big\{0.05, \sqrt{\tau(1-\tau)} (\log (p) / n)^{1 / 4}\big\}$ and $h' = \{(q+\log n) / n\}^{2 / 5}$ for $q = |E|$ at the selection and inference stages, respectively.
Our proposed method is implemented with white noise $\omega$
drawn from $\cN(0_p, \Omega)$ with $\varrand = \frac{1}{n} I_{p, p}$.
We report our findsing in Figure \ref{fig:coverage_new2}, Figure \ref{fig:length_new2} and Table \ref{table:F1scores_new2}.
These results are based on $500$ independent Monte Carlo datasets for each pair of model and signal setting.

\begin{figure}[h]
	\centering
	\includegraphics[scale = 0.7]{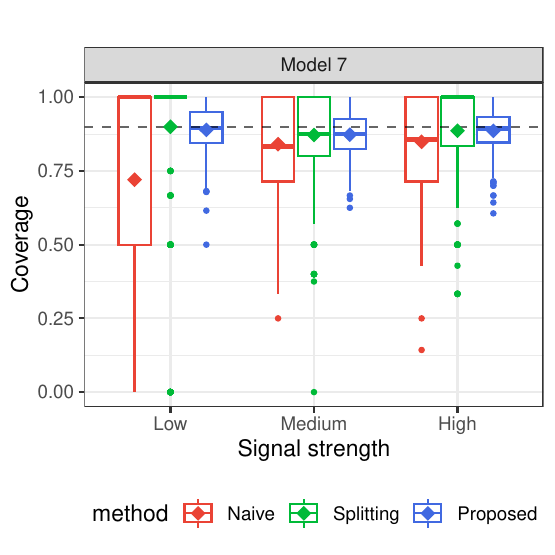}
	\caption{
    \small{Coverage rates of different methods across different models and signal settings. 
    The gray dashed line represents the prespecified target coverage rate at $0.9$, and the diamond marks highlight the averaged coverage rates over all replications. 
    }}
	\label{fig:coverage_new2}
\end{figure}

\begin{figure}[h]
	\centering
	\includegraphics[scale = 0.7]{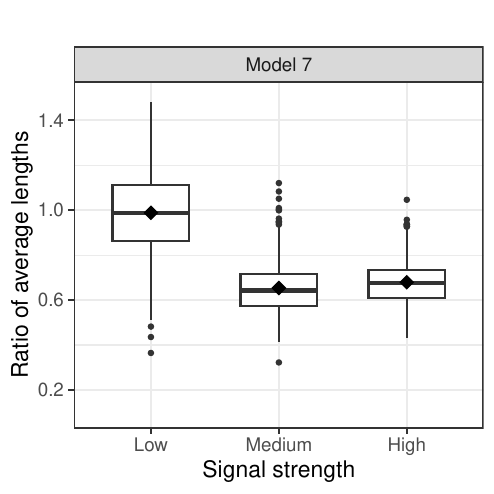}
	\caption{
    \small{The boxplots present the ratio of average interval lengths for the selected parameters between the ``Proposed'' method and the ``Splitting'' method across different models and signal strengths. 
    }}
	\label{fig:length_new2}
\end{figure}

\begin{table}[h]
\caption{\small{Accuracy based on F1 scores before and after applying ``Proposed'' in three signal settings.}}
\vspace{10pt}
\footnotesize
\centering
\begin{tabular}{cccc}
\hline
\hline
    & & Model 7 & \\
     \cline{2 - 4}  
Signal Strength & Low & Medium & High \\ 
\hline
\hline
F1 score before inference
    & 0.16 & 0.26 & 0.30  \\
F1 score after inference
    & 0.28 & 0.63 & 0.74  \\
\hline
\hline
\end{tabular}
\label{table:F1scores_new2}
\end{table}

\section{Additional information for Section \ref{sec:real}}
\label{appdx:data}

\subsection{Data processing and definition of variables}

The list of the variables and their description after our data was processed are provided in Table \ref{table:real_variable}.
Note that the variables ``Obstetric Estimate Edited'', ``Live Birth Order'', ``Delivery Weight'' were removed to avoid multicollinearity in our data. 
In order to ensure an adequate sample size for our replication analysis, we eliminated certain categorical variables from the dataset that had less than $150$ samples for one of their categories. 
These variables are: “Successful External Cephalic Version”, “Failed External Cephalic Version”, “Mother Transferred”, “Seizures”, “No Congenital Anomalies Checked”, “Mother’s Race/Hispanic Origin\_3”, “Mother’s Race/Hispanic Origin\_5”, and “Father’s Race/Hispanic Origin\_5”.

{\small
\begin{longtable}{ll}
\caption{Definition of variables} \\
\hline\hline
Variable                          & Description \\
\hline\hline
\endfirsthead

\multicolumn{2}{c}%
{{\tablename\ \thetable{} -- continued from previous page}} \\
\hline\hline
Variable                          & Description \\
\hline\hline
\endhead

\hline\hline
\endfoot

\hline\hline
\endlastfoot

Time of Birth                     & Time of Birth \\
Birth Day of Week                 & 1: Sunday, 2: Monday, 3: Tuesday, 4: Wednesday, \\
                                  & 5: Thursday, 6: Friday, 7: Saturday \\
Birth Place                       & 1: In hospital, 2: Not in hospital \\
Mother's Age                      & Mother's single years of age \\
Mother's Nativity                 & 1: Born in the U.S., 2: Born outside the U.S. \\
Paternity Acknowledged            & 1: Yes, 0: No \\
Marital Status                    & 1: Married, 2: Unmarried \\
Father's Age                      & Father's Combined Age \\
Prior Births Now Living           & Number of children still living from previous live births \\
Prior Births Now Dead             & Number of children dead from previous live births \\
Prior Other Terminations          & Number other terminations \\
Total Birth Order                 & Number of total birth order \\
Month Prenatal Care Began         & Month prenatal care began \\
Number of Prenatal Visits         & Number of prenatal visits \\
Cigarettes Before Pregnancy       & Number of cigarettes daily \\
Mother's Height                   & Height in inches \\
Pre-pregnancy Weight              & Weight in pounds \\
Cigarettes 1st Trimester          & Number of cigarettes daily \\
Cigarettes 2nd Trimester          & Number of cigarettes daily \\
Cigarettes 3rd Trimester          & Number of cigarettes daily \\
Weight Gain                       & Weight gain in pounds \\
Pre-pregnancy Diabetes            & 1: Yes, 0: No \\
Gestational Diabetes              & 1: Yes, 0: No \\
Pre-pregnancy Hypertension        & 1: Yes, 0: No \\
Gestational Hypertension          & 1: Yes, 0: No \\
Hypertension Eclampsia            & 1: Yes, 0: No \\
Previous Preterm Birth            & 1: Yes, 0: No \\
Infertility Treatment Used        & 1: Yes, 0: No \\
Number of Previous Cesareans      & Number of Previous Cesareans \\
No Infections Reported            & 1: True, 0: False \\
Induction of Labor                & 1: Yes, 0: No \\
Augmentation of Labor             & 1: Yes, 0: No \\
Steroids                          & 1: Yes, 0: No \\
Antibiotics                       & 1: Yes, 0: No \\
Chorioamnionitis                  & 1: Yes, 0: No \\
Anesthesia                        & 1: Yes, 0: No \\
No Maternal Morbidity Reported    & 1: True, 0: False \\
Attendant at Birth                & 1: Doctor of Medicine, 2: Doctor of Osteopathy, \\
                                  & 3: Certified Midwife, 4: Other Midwife, 5: Other \\
Five Minute APGAR Score           & A score of 0-10 \\
Ten Minute APGAR Score            & A score of 0-10 \\
Sex of Infant                     & 1: Male, 0: Female \\
Combined Gestation                & Week of Gestation \\
Assisted Ventilation (immediately)& 1: Yes, 0: No \\
Assisted Ventilation $>$ 6 hrs      & 1: Yes, 0: No \\
Admission to NICU                 & 1: Yes, 0: No \\
Surfactant                        & 1: Yes, 0: No \\
Antibiotics for Newborn           & 1: Yes, 0: No \\
Infant Transferred                & 1: Yes, 0: No \\
Infant Living at Time of Report   & 1: Yes, 0: No \\
Infant Breastfed at Discharge     & 1: Yes, 0: No \\
Mother's Education\_2             & 1: 9th through 12th grade with no diploma, 0: Otherwise \\
Mother's Education\_3             & 1: High school graduate or GED completed, 0: Otherwise \\
Mother's Education\_4             & 1: Some college credit, but not a degree., 0: Otherwise \\
Mother's Education\_5             & 1: Associate degree, 0: Otherwise \\
Mother's Education\_6             & 1: Bachelor's degree, 0: Otherwise \\
Mother's Education\_7             & 1: Master's degree, 0: Otherwise \\
Mother's Education\_8             & 1: Doctorate or Professional Degree, 0: Otherwise \\
Father's Education\_2             & 1: 9th through 12th grade with no diploma, 0: Otherwise \\
Father's Education\_3             & 1: High school graduate or GED completed, 0: Otherwise \\
Father's Education\_4             & 1: Some college credit, but not a degree., 0: Otherwise \\
Father's Education\_5             & 1: Associate degree, 0: Otherwise \\
Father's Education\_6             & 1: Bachelor's degree, 0: Otherwise \\
Father's Education\_7             & 1: Master's degree, 0: Otherwise \\
Father's Education\_8             & 1: Doctorate or Professional Degree, 0: Otherwise \\
Mother's Race/Hispanic Origin\_2  & 1: Non-Hispanic Black (only), 0: Otherwise \\
Mother's Race/Hispanic Origin\_4  & 1: Non-Hispanic Asian (only), 0: Otherwise \\
Mother's Race/Hispanic Origin\_6  & 1: Non-Hispanic more than one race, 0: Otherwise \\
Mother's Race/Hispanic Origin\_7  & 1: Hispanic, 0: Otherwise \\
Father's Race/Hispanic Origin\_2  & 1: Non-Hispanic Black (only), 0: Otherwise \\
Father's Race/Hispanic Origin\_3  & 1: Non-Hispanic AIAN (only), 0: Otherwise \\
Father's Race/Hispanic Origin\_4  & 1: Non-Hispanic Asian (only), 0: Otherwise \\
Father's Race/Hispanic Origin\_6  & 1: Non-Hispanic more than one race, 0: Otherwise \\
Father's Race/Hispanic Origin\_7  & 1: Hispanic, 0: Otherwise \\
Father's Race/Hispanic Origin\_8  & 1: Origin unknown or not stated, 0: Otherwise \\
Fetal Presentation at Delivery\_2 & 1: Breech, 0: Otherwise \\
Fetal Presentation at Delivery\_3 & 1: Other, 0: Otherwise \\
Delivery Method\_2                & 1: C-Section, 0: Otherwise \\
Payment\_2                        & 1: Private Insurance, 0: Otherwise \\
Payment\_3                        & 1: Self Pay, 0: Otherwise \\
Payment\_4                        & 1: Other, 0: Otherwise \\
Birth Season\_Spring              & 1: Spring, 0: Otherwise \\
Birth Season\_Summer              & 1: Summer, 0: Otherwise \\
Birth Season\_Winter              & 1: Winter, 0: Otherwise
\label{table:real_variable}
\end{longtable}
}

\subsection{Additional results}

In this section, we provide additional results for Section \ref{sec:real} with the bandwidths $h$ and $h'$ possibly different.

\subsubsection{Risk factors for low birth weight}

In this section, we use the proposed method to investigate the association between low birth weight in twins and various risk factors. 
We use the tuning parameter of $\lambda=0.4 \sqrt{\log p / n}$, the bandwidths for selection and inference $h = \max \big\{0.05, \sqrt{\tau(1-\tau)} (\log (p) / n)^{1 / 4}\big\} $ and $h' = \{(q+\log n) / n\}^{2 / 5}$.
We draw Gaussian white noise $\omega_n$ from $\cN\left(0_p, \frac{1}{2n}I_{p, p}\right)$. 
For data splitting, we use two-thirds of the samples for model selection and the remaining one-third of the samples for constructing confidence intervals for the effects of the selected variables on the $10\%$ conditional quantile. 
We focus on $\tau=0.1$ in our analysis.
In Figure \ref{fig:birth_CI_old}, we show the $90\%$ confidence intervals for the variables selected by ``Baseline''. 
For this same set of variables, we show the results from ``Proposed'' and ``Splitting'' when implemented on a randomly drawn subsample of size $n=500$.

\begin{figure}[h]
	\centering
	\includegraphics[scale = 0.68]{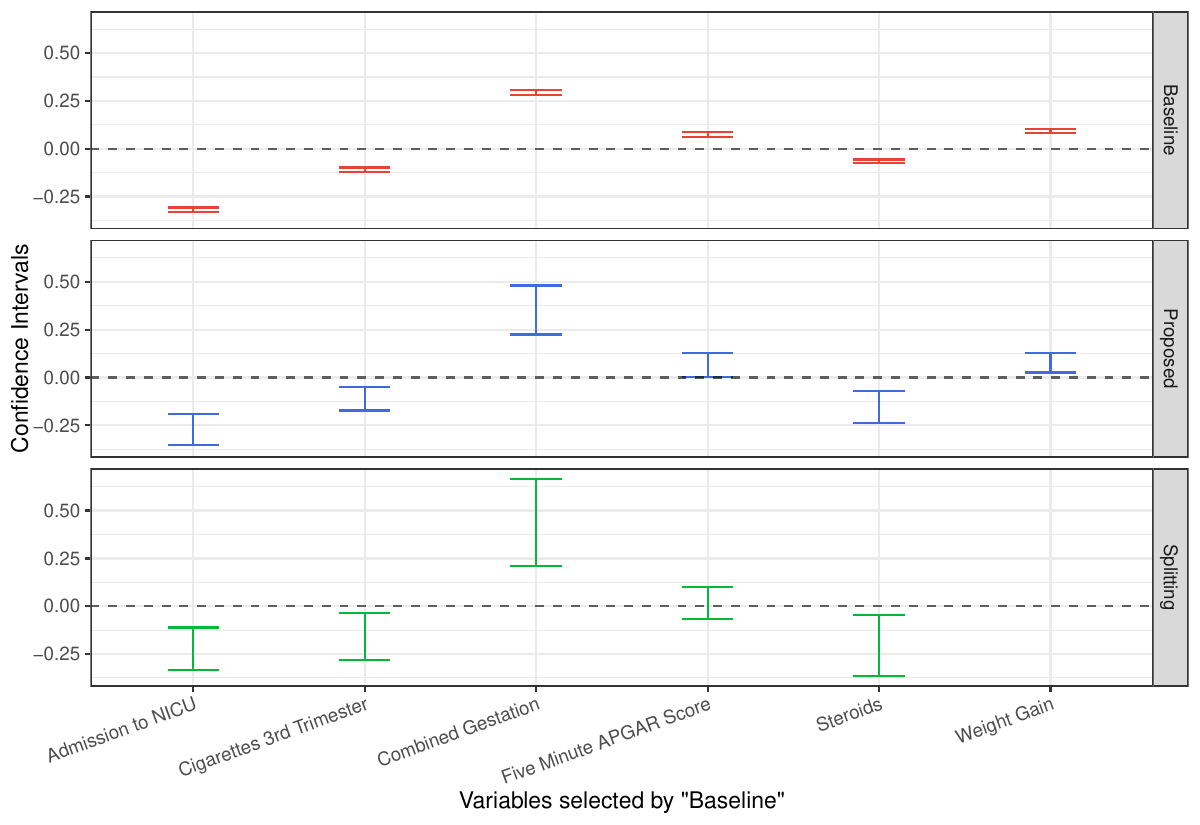}
	\caption{
   \small{$90\%$ confidence intervals for the variables chosen as significant by the full data analysis. 
     ``Splitting'' method does not select the ``Weight Gain'' factor and fails to identify the significance of ``Five Minute APGAR Score'' based on subsamples.
     In contrast, ``Proposed'' identifies the association between these factors and the low birth weight in twins even when a ``Baseline'' on the full data. 
     The average length of the confidence intervals produced by ``Proposed'' is $0.156$, while ``Splitting'' results in an average interval length of $0.283$}.}
	\label{fig:birth_CI_old}
\end{figure}

\subsubsection{Replicating our analysis at $10\%$ quantile level}

We repeat the above-described analysis on $100$ random subsamples of size of $500$. We report in Table \ref{table:birth_length_10_old}, the fraction of times that a selected variable was reported as significant with the post-selection interval estimators.
Additionally, in Figure \ref{fig:birth_length_10_old}, we display the interval lengths for each selected variable as well as the average interval lengths for all selected variables.

\begin{table}[h]
\caption{\small{The fraction of times the post-selection interval for a variable did not include zero out of the total number of times the same variable was selected.
}}
\vspace{10pt}
\small
\centering
\begin{tabular}{ccccccc}
\hline 
\hline
Variables & \makecell{Admission to\\ NICU} & \makecell{Cigarettes 3rd\\ Trimester} & \makecell{Combined\\ Gestation} & \makecell{Five Minute\\ APGAR Score} & \makecell{Steroids} & \makecell{Weight\\ Gain} \\ 
\hline 
\hline
Proposed  & 0.979 & 0.728 & 0.980 & 0.456 & 0.404 & 0.622  \\ 
Splitting & 0.960 & 0.660 & 0.940 & 0.135 & 0.462 & 0.488  \\
\hline 
\hline
\end{tabular}
\label{table:birth_length_10_old}
\end{table}

\begin{figure}[H]
	\centering
	\includegraphics[scale = 0.7]{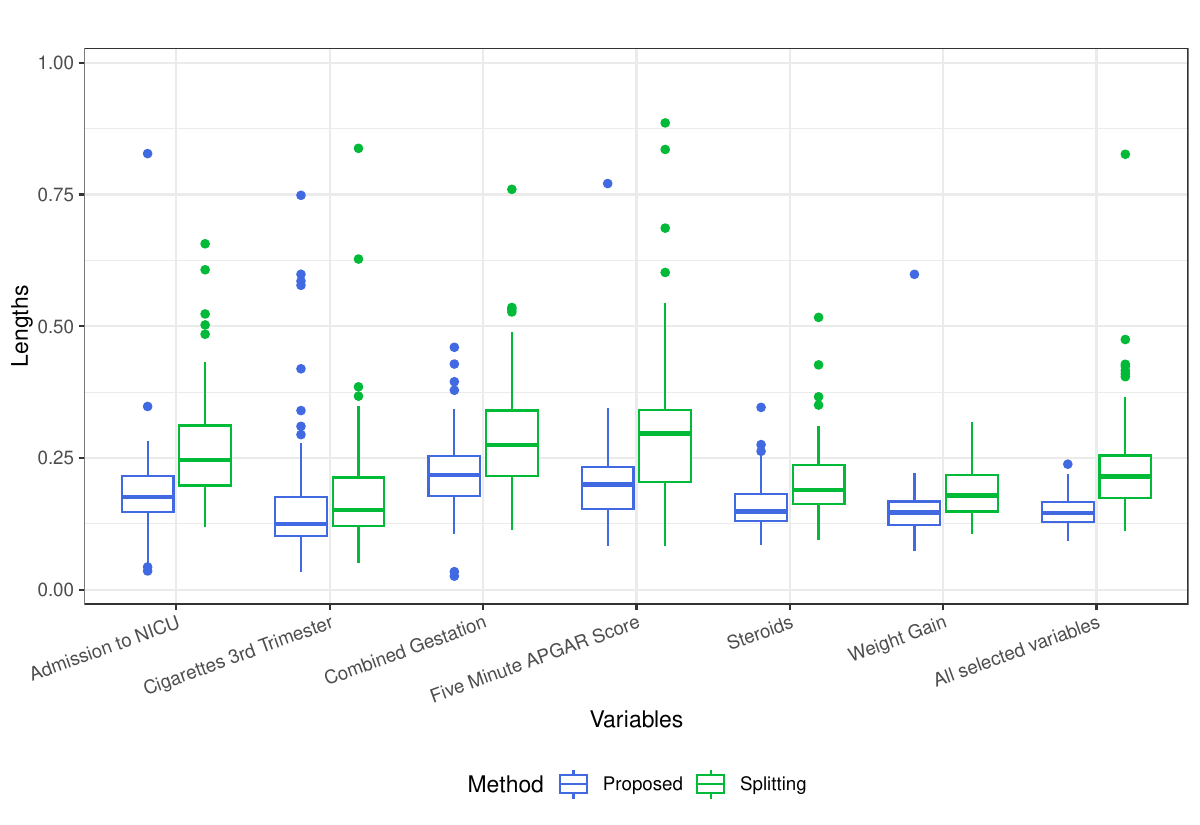}
	\caption{
   \small{
   Box plots for the lengths of $90\%$ confidence intervals of each selected variable and the average lengths of all selected variables at $10\%$ quantile level. 
   ``Proposed'' results in shorter intervals compared to ``Splitting'' for each variable and overall on average for all variables.}}
	\label{fig:birth_length_10_old}
\end{figure}

\subsubsection{Replicating our analysis at $1\%$ quantile level}

We repeat the above-described analysis on $100$ random subsamples of size of $500$. 
In Figure \ref{fig:birth_length_1_old}, we display the interval lengths for each selected variable as well as the average interval lengths for all selected variables.

\begin{figure}[H]
	\centering
	\includegraphics[scale = 0.7]{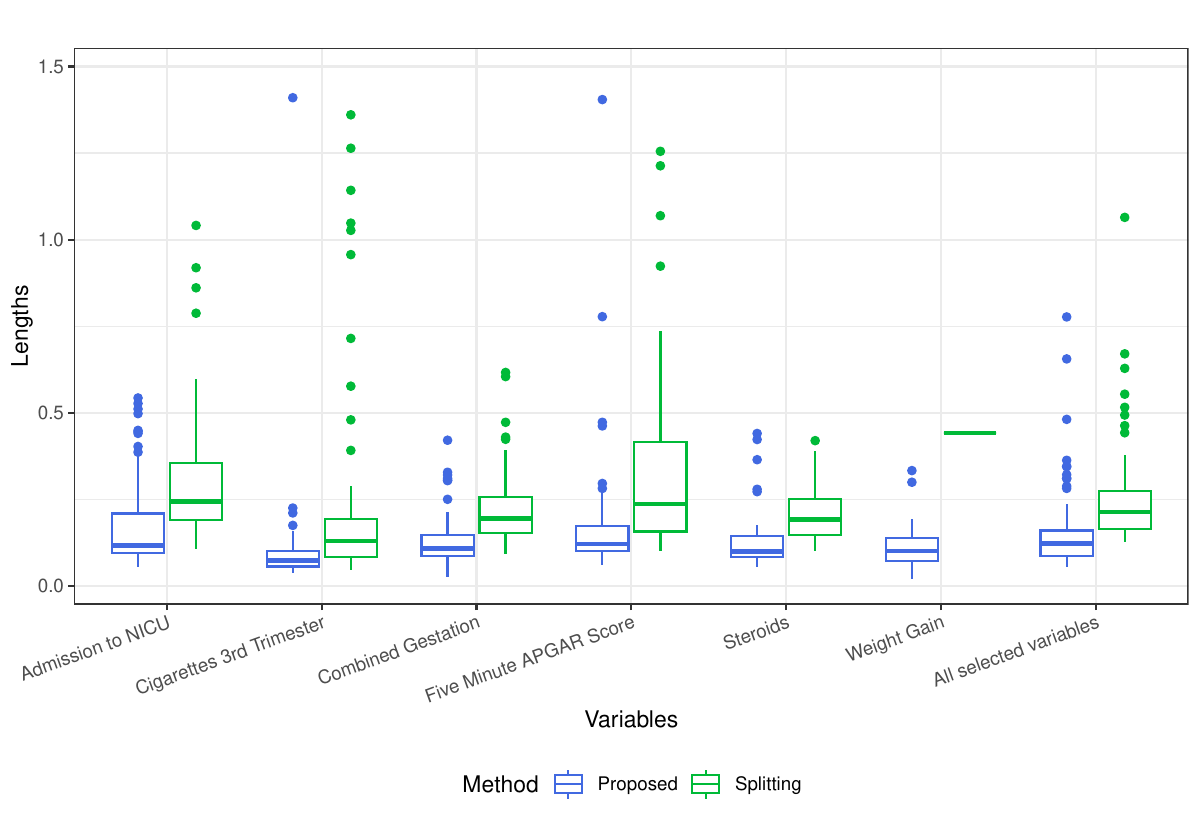}
	\caption{
   \small{
   Box plots for the lengths of $90\%$ confidence intervals of each selected variable and the average lengths of all selected variables at $1\%$ quantile level. 
   ``Proposed'' results in shorter intervals compared to ``Splitting'' for each variable and overall on average for all variables.}}
	\label{fig:birth_length_1_old}
\end{figure}

\end{spacing}
\end{document}